\documentclass[onecolumn,12pt]{IEEEtran}
\usepackage[cmex10]{amsmath}
\usepackage{amsfonts}
\usepackage{color,graphicx}
\usepackage{subfig}

\usepackage{xspace}
\usepackage{mathrsfs}

\newtheorem{definition}{Definition}[section]
\newtheorem{prop}{Proposition}[section]

\newcommand{\A}{\ensuremath{\mathrm{A}}\xspace}
\newcommand{\B}{\ensuremath{\mathrm{B}}\xspace}
\newcommand{\C}{\ensuremath{\mathrm{C}}\xspace}
\newcommand{\Ci}[1]{\ensuremath{{{C}}_\text{\sf #1}}}
\newcommand{\Ri}[1]{\ensuremath{{{R}}_\text{\sf #1}}}
\newcommand{\Cii}[2]{\ensuremath{{{C}}_\text{\sf #1}^\text{\sf #2}}}
\newcommand{\CRo}{\ensuremath{C_{R_0}}}
\newcommand{\CRoi}[1]{\ensuremath{C_{R_0}^\text{\sf #1}}}
\newcommand{\Rii}[2]{\ensuremath{{{R}}_\text{\sf #1}^\text{\sf #2}}}
\newcommand{\SNR}{{\sf SNR}}
\newcommand{\INR}{{\sf INR}}
\newcommand{\CNR}{{\sf CNR}}
\newcommand{\Ve}{\ensuremath{V}}
\newcommand{\ve}{\ensuremath{v}}

\newcommand{\IFconf}{\ensuremath{\text{IF}^\text{coop}}}

\newcommand{\bp}[2]{\ensuremath{{\sf{C}}_{#1#2}}}
\newcommand{\rly}[3]{\ensuremath{{\Delta \sf{R}}_{#1#2#3}}}
\newcommand{\spR}[2]{\ensuremath{\tilde{\sf{R}}^{#1}_{#2}}}
\newcommand{\spU}[2]{\ensuremath{\tilde{U}^{#1}_{#2}}}
\newcommand{\spu}[2]{\ensuremath{\tilde{u}^{#1}_{#2}}}
\newcommand{\ceil}[1]{\ensuremath{\lceil{#1}\rceil}}
\renewcommand{\d}[1]{\ensuremath{\delta_{#1}}}
\newcommand{\Var}[1]{\ensuremath{\mathrm{Var}\left({#1}}\right)}

\newcommand{\x}{x}
\newcommand{\y}{y}
\newcommand{\z}{z}

\newtheorem{thm}{Theorem}[section]
\newtheorem{lem}{Lemma}[section]
\newcommand{\qed}{\hfill \mbox{\raggedright \rule{.1in}{.1in}}}

\begin{document}

\title{Interference Channels with Half-Duplex Source Cooperation}

\author{\IEEEauthorblockN{Rui Wu, Vinod Prabhakaran, Pramod Viswanath and Yi  Wang\footnote{R. Wu and P. Viswanath are with the Coordinated Science Laboratory and the department of Electrical and Computer Engineering at the University of Illinois, Urbana-Champaign. V. Prabhakaran is with the Tata Institute of Fundamental Research, Bombay, India. Y. Wang is with Huawei Technologies. This paper was presented in part at the IEEE ISIT 2010 in  Austin, Texas. This work has been supported in part by grants from NSF CCF 1017430, NSF ECCS 1232257 and Intel¡¯s University Research Office. Copyright \copyright\ 2013 IEEE. Personal use of this material is permitted.  }}}

\maketitle

\begin{abstract}
The performance gain by allowing half-duplex source cooperation is studied for Gaussian interference channels.
The source cooperation is {\em in-band}, meaning that each source can listen to the other source's transmission, but
there is no independent (or orthogonal) channel between the sources. The half-duplex constraint supposes that at each time instant  the sources can either transmit or listen, but not do both. 

Our main result is a characterization of the sum capacity when the cooperation is bidirectional and the channel gains are symmetric. With unidirectional cooperation, we essentially have a cognitive radio channel. By requiring the primary to achieve a rate  close to
its link capacity, the best possible rate for the secondary is characterized within a constant. Novel inner and outer bounds are derived as part of these characterizations. 
\end{abstract}

\section{Introduction}
A basic characteristic of the wireless medium is its {\em broadcast}
nature. This manifests itself as interference when multiple users
try to share the medium. An active area of research which
investigates efficient schemes for managing interference has focused
on interference
channels~\cite{IFalign,ElGamalCosta,EtkinTseWang,HanKobayashi}.
However, the broadcast feature is also a blessing in disguise in
that the same transmission could be heard by multiple receivers,
opening up the possibility of cooperation. Traditionally, the
cooperation aspect has been investigated separately using relay
channels in which only one source-destination pair is
present~\cite{CoverElGamal}. Recently, the role of cooperation in
managing interference has come under scrutiny (\cite{CaoChen07,CaoChen09,Host-Madsen06,MYK07,Mohajer08,OLT07,PrabhakaranVi09S,PrabhakaranVi09D,RFL09,SSNPS07,Thejaswi08,Tuninetti07,WT11,YangTuninetti08,YuZhou08}
is an incomplete list of references).

In this paper we investigate reliable communication over the two-user interference channel,
where the two sources may not only transmit but also receive
(Figure~\ref{Fig:HD_channel}). This ability to receive will allow the sources to
cooperate. However, to be realistic about the gains that can be
derived from this cooperation, we impose two key restrictions:
\begin{itemize}

\item {\em In-band cooperation.} No extra orthogonal band is available for
the source nodes to transmit to each other over; all transmission and
reception must happen over the same band. Thus, the sources cooperate by
transmitting and receiving over the same band that is originally available
for the interference channel.

\item {\em Half-duplex operation.} Each source node may either transmit or
receive at a time but cannot do both. This respects the limitations of current
hardware technology.

\end{itemize}

\begin{figure}
  \centering\scalebox{0.6}{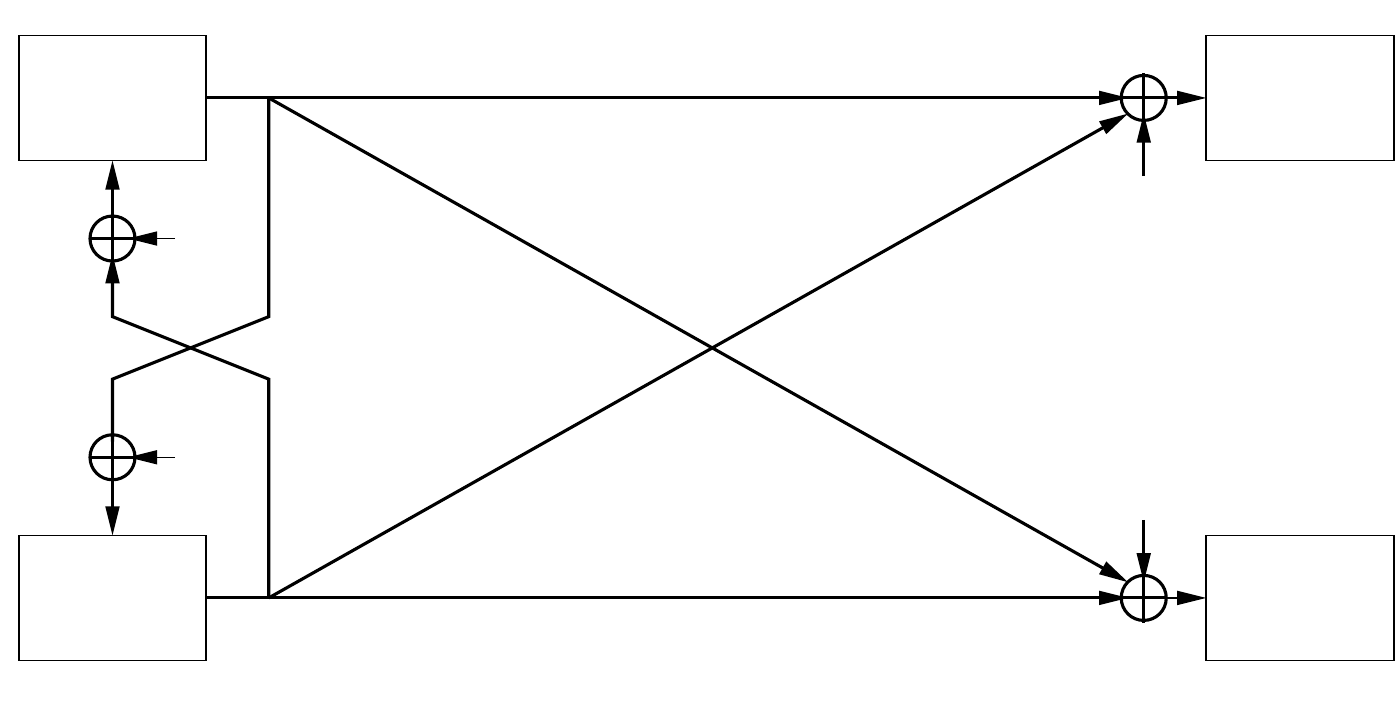}\label{Fig:hd_channel}\\
  \vfil
  \centering{\subfloat[mode $A$]{\scalebox{0.2}{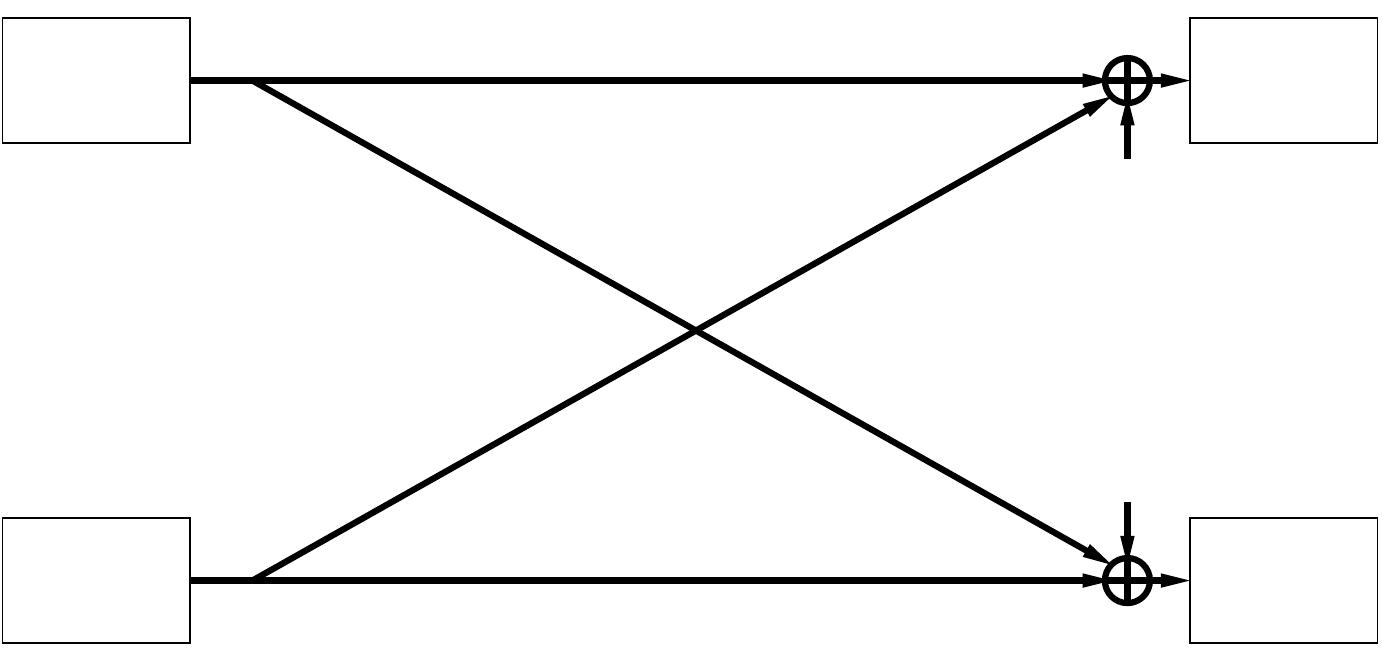}\label{Fig:hd_channel_A}}\hfil
  \subfloat[mode $B$]{\scalebox{0.2}{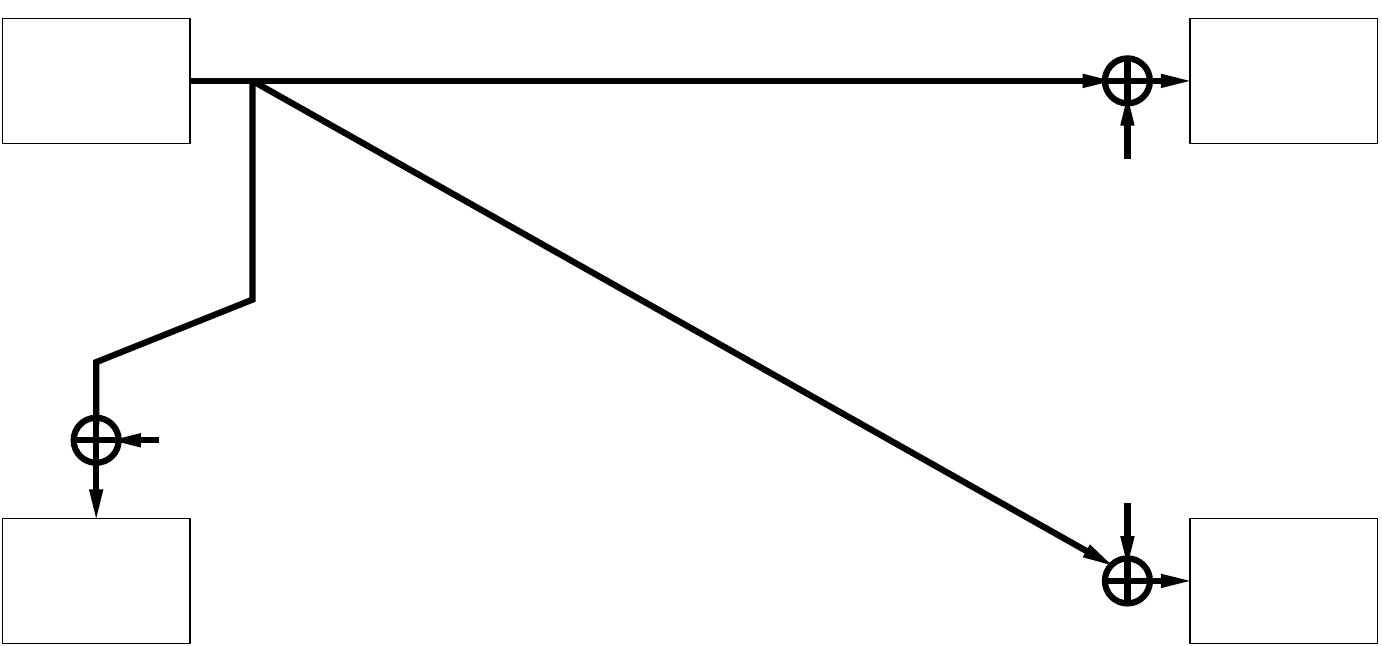}\label{Fig:hd_channel_B}}\hfil
  \subfloat[mode
  $C$]{\scalebox{0.2}{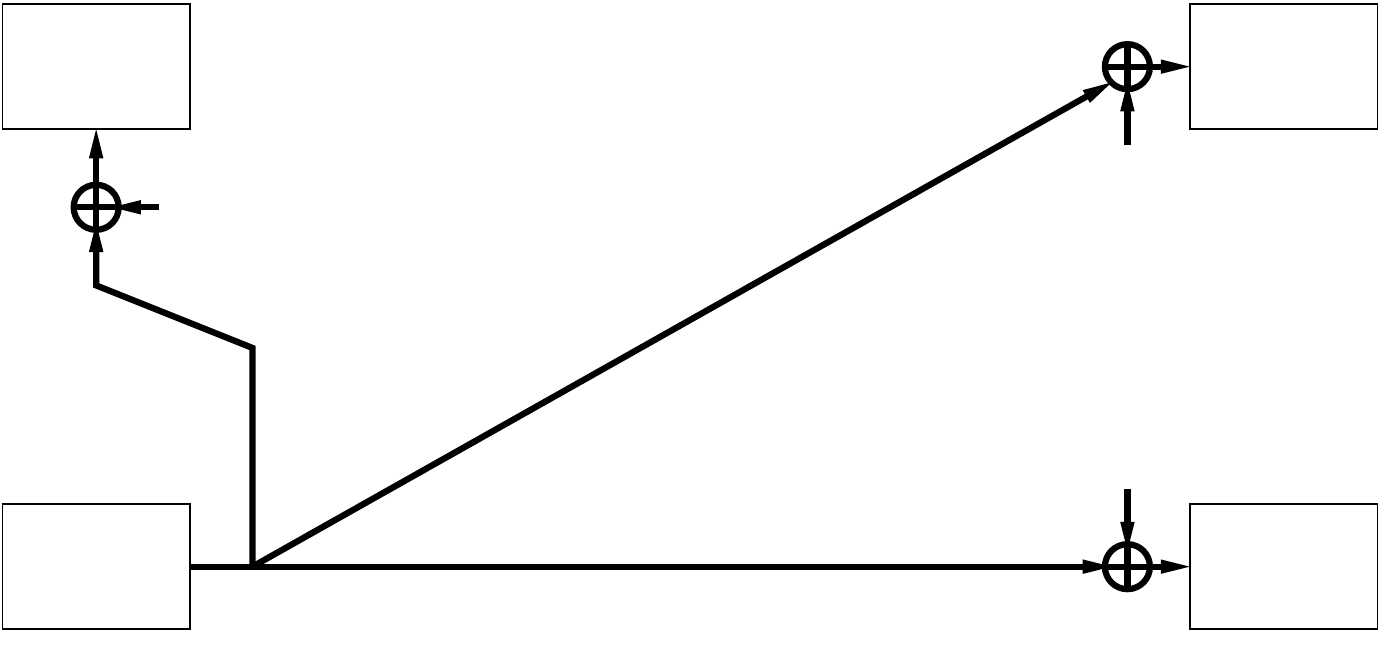}\label{Fig:hd_channel_C}}}\vfil
  \caption{Interference channel with half-duplex source cooperation. The sources can work in three modes: $(\A)$ both sources transmit, $(\B)$ source~1
transmits while source~2 receives, and $(\C)$ source~2 transmits and
source~1 receives.}
  \label{Fig:HD_channel}
\end{figure}

Reliable communication with a half-duplex constraint has been previously studied
in the context of the relay channel in \cite{HD_Relay}. In
\cite{Coop_Diversity}, half-duplex cooperation was used to provide spatial diversity for
fading channels. recent work on full duplex cooperation in interference channels ~\cite{PrabhakaranVi09S,WT11,HD_Coop} are closely related to the present manuscript.
\cite{PrabhakaranVi09S} studied source cooperation under full-duplex assumption.
In \cite{WT11} cooperation is over conferencing links orthogonal to the original channel.
in contrast to in-band cooperation here.
Our
model is identical to the source cooperation part in \cite{HD_Coop}. In \cite{HD_Coop},
an achievable rate region is provided, but the outer bound is only studied the case when
cooperation is very strong.
The work \cite{CRC} considers the half duplex cooperation for relay channels, which are special cases of the interference channels considered in this paper.

The characterization of the capacity region in this setting is quite challenging -- it includes the canonical interference and relay channels as special cases. A complete characterization of the capacity region is also made further complicated by the huge amount of notation and description complexity of the  region. While we do not characterize the entire capacity region in this work, we nevertheless make significant progress in understanding the nature of near optimal communication schemes in this setting.
We present our results in  two different scenarios, aiming at minimizing the notation and description complexity while providing the maximum intuition to the nature of the capacity region as well as gains of interference mitigation via cooperation.
\begin{itemize}
\item In the first scenario, the cooperation is bidirectional and the channel gains are symmetric. Our main result is a characterization of  the sum
capacity of this channel within a constant. Maintaining symmetric channel gains is primarily aimed at reducing the  notational burden.
\item In the second scenario, the cooperation is unidirectional; i.e., source
2 can listen to source 1's transmission but not the other way around. This setting is essentially what is also known in the literature as a ``cognitive radio channel", where  we cansider source
1 and destination 3 the ``primary user" and source 2 and destination 4 the ``secondary user". The main question we address is the following: what rate can the secondary  user achieve without affecting the primary user's
performance by much? The largest such rate, known as  the capacity of the cognitive  radio channel, is characterized in this manuscipt up to a constant.
\end{itemize}

The coding scheme we use to enable reliable communication is quite general and can be applied to all interference channels with half-duplex source cooperation. The
key idea is to turn the half-duplex cooperation problem to a virtual channel problem. A virtual channel is an interference
channel with rate-limited bit-pipes between the two sources and from each source to the destination where it causes interference.
This virtual channel is similar to the channel considered in \cite{WT11} except that there they do not have bit-pipes from sources
to destinations.

The coding scheme for the virtual channel is an extension of the superposition coding scheme for the interference channels
\cite{HanKobayashi}. In addition to public and private messages, we further introduce cooperative private and pre-shared public
messages. Cooperative private messages are shared over the bit-pipes between the two sources so they can be sent using
source beamforming. Pre-shared public messages are shared over the bit-pipes from the sources to the destinations so the
signals corresponding to such messages can be canceled at the other destination and do not cause interference.

To reduce the original channel to a virtual channel, we schedule the transmission in two steps. In the first step, only one
source transmits and the other source listens. The active source can send data to its destination, share information with
the other nodes, or relay data from the other source to the other destination. In the second step, both sources transmit. The
shared information from the previous step and the interference channel together is indeed a virtual channel, and the scheme
mentioned above is applied to this channel. In the end, we optimize over the scheduling parameters to get the best
achievable rate.

An important tool we use to study the Gaussian channel is the linear deterministic model introduced in \cite{ADT08}. The linear deterministic model focuses on modeling the broadcast and interference of the signals. It assumes that the signals are quantized and the noise is negligible, which can be a good approximation in the high SNR regime. For each problem, we will first study the corresponding linear deterministic model and then consider an achievable scheme that mimics the scheme of the linear deterministic model for the Gaussian case. We note that it is possible to get the constant gap result for the Gaussian case directly, but the linear deterministic model allows us to get a clearer understanding of the coding schemes as it is much simpler to deal with.

The rest of the paper is organized as follows. In section~\ref{sec:problem}, we formally state the two problems and in section~\ref{sec:result} the
main results about the sum capacity and the cognitive capacity are given. Section~\ref{sec:sym_LDM} and section~\ref{sec:sym_Gaussian} deal with the
symmetric case and section~\ref{sec:cog_LDM} and section~\ref{sec:cog_Gaussian} are for the cognitive case. In both cases, we start by examining the
corresponding linear deterministic model and use the intuition derived to work with the more complicated Gaussian model. 

\section{Problem Statement}\label{sec:problem}
\subsection{The Symmetric Case}\label{sec:problem_sym}
The Gaussian interference channel with bidirectional source
cooperation is depicted in Figure~\ref{Fig:HD_channel}.

The source nodes 1 and 2 want to communicate with destination nodes
3 and 4, respectively. The communication is over discrete time slots $t = 1, \dots, L$. We assume that the additive noise processes are memoryless and independent across receivers. Without loss of generality, we also assume that the channel is normalized. 
i.e., the additive noise processes $(Z_{it}),
i = 1, 2, 3, 4$ are independent $\mathcal{CN}(0, 1)$, i.i.d. over
time, and the codeword $(X_{it})$ at source i satisfies the power
constraint
\begin{align*}
  \frac{1}{L}\sum_{t=1}^L{E\left[|X_{it}|^2\right]}\leq 1, i = 1, 2.
\end{align*}
Further, we assume the channel is symmetric, i.e.,
$|h_{13}|^2=|h_{24}|^2=\SNR, |h_{14}|^2=|h_{23}|^2=\INR, |h_{12}|^2=|h_{21}|^2=\CNR$.

As the cooperation is half-duplex, the first source chooses to transmit (send) or listen at each time $t=1,2,\ldots,n$ based on its message $W_1$ and what it has received so far $Y_1^{t-1}$. Thus, the first source's input to the channel is $(X_1,S_1)\in{\mathbb C}\times\{1,0\}$ and the encoding function is $(X_{1,t},S_{1,t})=f_{1,t}(W_1,Y_1^{t-1})$. Furthermore the power constraint at the first source's transmitter is $(1/n)\sum_{t=1}^n E[|X_{1,t}|^21_{S_{1,t}=1}]$. Similary for the second source. The channel outputs are as follows:
\begin{align*}
Y_{1,t}&= (h_{21}X_{2,t}+Z_{1,t})1_{S_{1,t}=0}\\
Y_{2,t}&= (h_{12}X_{1,t}+Z_{2,t})1_{S_{2,t}=0}\\
Y_{3,t}&= h_{13}X_{1,t}1_{S_{1,t}=1} + h_{23}X_{2,t}1_{S_{2,t}=1} + Z_{3,t},\\
Y_{4,t}&= h_{14}X_{1,t}1_{S_{1,t}=1} + h_{24}X_{2,t}1_{S_{2,t}=1} + Z_{4,t}.
\end{align*}
More specifically, the channel can be in one of the following three modes. In mode $A$, both sources transmit. The
nodes receive
\begin{align*}
  &Y_{1t} = 0,\\
  &Y_{2t} = 0,\\
  &Y_{3t} = h_{13}X_{1t}+h_{23}X_{2t}+Z_{3t},\\
  &Y_{4t} = h_{14}X_{1t}+h_{24}X_{2t}+Z_{4t}.
\end{align*}
In mode $B$, source 1 transmits and source 2 listens. Then
\begin{align*}
  &Y_{1t} = 0,\\
  &Y_{2t} = h_{12}X_{1t}+Z_{2t},\\
  &Y_{3t} = h_{13}X_{1t}+Z_{3t},\\
  &Y_{4t} = h_{14}X_{1t}+Z_{4t}.
\end{align*}
In mode $C$, source 2 transmits, source 1 listens, and
\begin{align*}
  &Y_{1t} = h_{21}X_{2t}+Z_{1t},\\
  &Y_{2t} = 0,\\
  &Y_{3t} = h_{23}X_{2t}+Z_{3t},\\
  &Y_{4t} = h_{24}X_{2t}+Z_{4t}.
\end{align*}

A block length-$L$ codebook of rate $(R_1, R_2)$ for the channel consists of a schedule function $\varphi(t)\in\{\A,\B,\C\}$ and a sequence of encoding functions $f_{it}$ and decoding functions
$g_{i+2}$, $i=1, 2$, $t = 1, 2, \dots, L$. The scheduling function specifies which mode the channel is in at time $t$. The source messages $W_i\in \{1, 2, \dots, 2^{LR_i}\}$, $i = 1, 2$ are independent and uniformly distributed. The sources transmit $X_{it} = f_{it}(W_i, Y_i^{t-1})$, where $Y_i^{t-1} = (Y_{i1}, \dots, Y_{i(t-1)})$. Note that the encoding functions are causal. Further, the encoding functions also are constrained by a scheduling function $\varphi(t)$; i.e., we have $X_{2t} = f_{2t}(W_2, Y_2^{t-1}) = 0$ when $\varphi(t)=\B$ and $X_{1t} = f_{1t}(W_1, Y_1^{t-1}) = 0$ when $\varphi(t) = \C$. \mbox{Destination-$(i+2)$} estimates the message intended for it as $\hat{W_i} = g_{i+2}(Y_{i+2}^L)$, $i=1,2$. We say that a rate pair $(R_1, R_2)$ is {\em achievable} if there is sequence of rate $(R_1,R_2)$ codebooks such that as $L\to \infty$,
\begin{align*}
P(\hat{W_i} \ne W_i)\to 0, i = 1, 2.
\end{align*}
The capacity region $\mathscr{C}$ is the collection of all
achievable $(R_1, R_2)$. The {\em sum-capacity} $C_{sum}$ of the channel is
defined as the largest $R_1+R_2$ such that $(R_1,R_2)\in
\mathscr{C}$. In Section \ref{sec:result} we will provide a
characterization of the sum-capacity within a constant.

\subsection{The Cognitive Case}\label{sec:problem_cog}
The Gaussian interference channel with unidirectional source
cooperation is depicted in Figure~\ref{Fig:HD_channel_cog}. This channel has no cooperation link from source 2 to
source 1.

The source nodes 1 and 2 want to communicate with destination nodes
3 and 4, respectively. The communication is over discrete time slots $t = 1, \dots, L$. Without loss of generality, we assume the
channel is normalized; i.e., the additive noise processes $(Z_{it}),
i = 2, 3, 4$ are independent $\mathcal{CN}(0, 1)$, i.i.d. over time,
and the codeword $(X_{it})$ at source i satisfies the power
constraint
\begin{align*}
  \frac{1}{L}\sum_{t=1}^L{E\left[|X_{it}|^2\right]}\leq 1, i = 1, 2.
\end{align*}
Here, we assume that the channel gains are asymmetric in general. We
can view source 1 as the primary user and source 2 as the secondary
user, and the secondary can listen to the primary's transmission and
adapt its behavior accordingly. Hence, this case corresponds to the
cognitive scenario.

\begin{figure}
  \centering\scalebox{0.6}{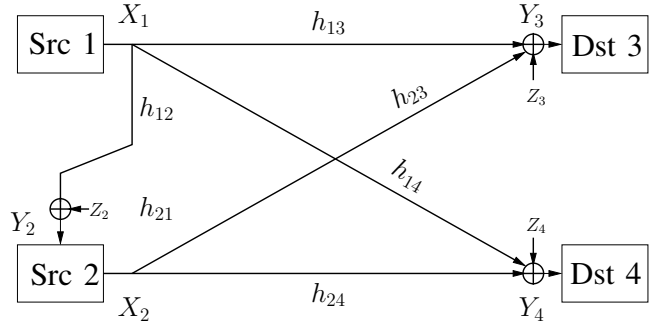}
  \caption{Interference channel with unidirectional half-duplex source cooperation.}\label{Fig:HD_channel_cog}
\end{figure}

As there is only one-side half-duplex cooperation, the secondary sender chooses to transmit (send) or listen at each time $t=1,2,\ldots,n$ based on its message $W_2$ and what it has received so far $Y_2^{t-1}$. Secondary transmitter's input to the channel is $(X_2,S_2)\in{\mathbb C}\times\{1,0\}$ and the encoding function is $(X_{2,t},S_{2,t})=f_{2,t}(W_2,Y_2^{t-1})$. Furthermore the power constraint at the secondary transmitter is $(1/n)\sum_{t=1}^n E[|X_{2,t}|^21_{S_{2,t}=1}]$. The channel outputs are as follows:
\begin{align*}
Y_{1,t}&=0\\
Y_{2,t}&= (h_{12}X_{1,t}+Z_{2,t})1_{S_{2,t}=0}\\
Y_{3,t}&= h_{13}X_{1,t} + h_{23}X_{2,t}1_{S_{2,t}=1} + Z_{3,t},\\
Y_{4,t}&= h_{14}X_{1,t} + h_{24}X_{2,t}1_{S_{2,t}=1} + Z_{4,t}.
\end{align*}
More specifically, the channel can be in one of the following two modes. In mode $A$, both sources
transmit. The nodes receive
\begin{align*}
  &Y_{1t} = 0,\\
  &Y_{2t} = 0,\\
  &Y_{3t} = h_{13}X_{1t}+h_{23}X_{2t}+Z_{3t},\\
  &Y_{4t} = h_{14}X_{1t}+h_{24}X_{2t}+Z_{4t}.
\end{align*}
In mode $B$, source 1 transmits and source 2 listens. Then
\begin{align*}
  &Y_{1t} = 0,\\
  &Y_{2t} = h_{12}X_{1t}+Z_{2t},\\
  &Y_{3t} = h_{13}X_{1t}+Z_{3t},\\
  &Y_{4t} = h_{14}X_{1t}+Z_{4t}.
\end{align*}

Let $\SNR_1 = |h_{13}|^2, \SNR_2 = |h_{24}|^2, \INR_1 = |h_{23}|^2, \INR_2 = |h_{14}|^2, \CNR =
|h_{12}|^2$.

The codebook definition is similar to that in the symmetric case except that now the scheduling function
$\varphi(t)$ only takes value in $\{A, B\}$ and the encoding function $f_{1t}$ is only a function of
$W_1$, as $Y_{1t}$ is always 0. In this case, instead of the sum
capacity, we are more interested in another question from the cognitive
perspective: what can the secondary achieve if we do not sacrifice the
primary's performance? This motivates us to consider the following definition.
\begin{definition}
  Let $C_0 = \log(1+\SNR_1)$ be the capacity achieved by source 1 when $X_{2t}=0, \forall
  t$. Then $R_0$-capacity for the secondary user is defined as
  \begin{align*}
    C_{R_0} = \max_{\begin{subarray}{c} (R_1, R_2)\in \mathscr{C}\\R_1\geq C_0-R_0
    \end{subarray}}R_2.
  \end{align*}
\end{definition}
This definition specifies the best performance the secondary user can get, given that the primary user
backs off less than $R_0$ from its link capacity. In
Section~\ref{sec:result}, the $R_0$-capacity is characterized when $R_0$ is
larger than some constant.

To see why we introduce a back-off in the primary rate, consider the
Z-channel where $\CNR=\INR_2 = 0, \SNR_1 = \SNR_2 = \SNR>1, \INR_1 = \INR = 1$.
Let $R_{ij}$ denote the rate from source $i$ to destination $j$. For a complex Gaussian
Z-channel with weak interference (i.e., $\INR < \SNR$), the achievable rate-tuple
$(R_{11}, R_{21}, R_{22})$ must satisfy~\cite[Theorem 2]{Z_channel}
\begin{align*}
  R_{21}\leq & \log\left(1+\frac{(1-\beta)\INR}{1+\beta \INR}\right)\\
  R_{22}\leq & \log(1+\beta \SNR)\\
  R_{11}+R_{21}\leq & \log\left(1+\frac{\SNR+(1-\beta)\INR}{\beta \INR+1}\right),
\end{align*}
for some $0\leq \beta\leq 1$. In our case, $R_{11} = R_1, R_{22} = R_2$ and $R_{21} =
0$, thus the above constraints reduce to
\begin{align*}
  R_2\leq & \log(1+\beta \SNR)\\
  R_1\leq & \log\left(1+\frac{\SNR+(1-\beta)\INR}{\beta \INR+1}\right),
\end{align*}
for some $0\leq \beta\leq 1$. If no back-off is allowed, i.e., we insist that
$R_1=C_0 = \log(1+\SNR)$, then we must have $\beta\leq \frac{1}{1+\SNR}$, which gives
$R_2\leq \log(1+\frac{\SNR}{1+\SNR})\leq 1$ bit. However, if the primary can back off its rate
by 1 bit, then the secondary can send to its destination at full power and achieve a nonconstant rate
$R_2 = \log(1+\SNR)$. Notice that the gap between the two is unbounded when $\SNR$ scales to
$\infty$. Since we are more interested in the high-SNR region and would want to
characterize capacity only up to a constant, the definition above with back-off
better serves our purpose.

We further remark that this definition is not a constant gap
characterization of the upper-right corner point of the capacity region
$\mathscr{C}$. In fact, with the help of the
secondary, the primary can do strictly better than $C_0$ in some channel parameter settings.

\section{Results}\label{sec:result}
The main result of this paper is the approximate characterization of the
sum capacity of the symmetric case and the $R_0$-capacity of the cognitive
case for $R_0$ larger than some constant. We state them in the
following two theorems and highlight the gains we can get
from half-duplex cooperation.
To prove these theorems, we first motivate the schemes we use by studying the corresponding linear deterministic model in Section~\ref{sec:sym_LDM} and \ref{sec:cog_LDM}. We then sketch the proofs in Section \ref{sec:sym_Gaussian} and
\ref{sec:cog_Gaussian}, with details taken up in the appendices.

\subsection{The Symmetric Case}\label{sec:sym result}

Let $\theta_{ij}$ be the phase angle of $h_{ij}$ and define $\theta$ to be the angle difference between the direct links and the interference links, i.e., $\theta = \theta_{13}+\theta_{24}-\theta_{14}-\theta_{23}$ . We say the channel is aligned if $\SNR = \INR$ and $\theta = 0$. The following theorem characterizes the sum capacity of the symmetric channel within a constant.

\begin{thm}\label{thm:sumrate}
  Define $\overline{\Ci {sum}} = \max_{\delta}\overline{\Ci {sum}}(\delta) = \max_{\delta}\min(u_1, u_2, u_3,
u_4)$, where
  \begin{align*}
  u_1 = &\frac{2}{2+\delta}\Big[\delta \log(1+\SNR)+\log(1+\SNR+\CNR)\Big]\\
  u_2 = &\frac{1}{2+\delta}\Big[\delta \log(1+2\SNR+2\INR)+\log(1+\SNR)+\log(1+\SNR+\INR+\CNR)\\&+\delta\log(1+\frac{\SNR}{1+\INR})\Big]\\
  u_3 = &\frac{2}{2+\delta}\Big[\delta\max\{\log(1+\INR+\frac{2\SNR+\INR}{1+\INR}), \log(1+2\INR)\}+\log(1+\SNR+\INR+\CNR)]\\
  u_4 = &\frac{1}{2+\delta}\Big[\delta\log(1+4\SNR+4\INR+\SNR^2+\INR^2-2\SNR\INR\cos\theta)+2\log(1+\SNR+\INR)\Big].
  \end{align*}
  Then the sum capacity $\Ci {sum}$ of the symmetric channel defined in section~\ref{sec:problem_sym} satisfies
  $\overline{\Ci {sum}}-17 \leq \Ci {sum} \leq \overline{\Ci {sum}}+7$.
\end{thm}

In the coding scheme, we consider a symmetric scheduling and the number of time slots spent in mode $\B$ and $\C$ are the same, i.e., $|\varphi^{-1}(B)| = |\varphi^{-1}(\C)|$ where by definition $\varphi^{-1}(\B)$ is the set of time slots scheduled for mode $\B$ and $|S|$ denotes the cardinality set $S$. We define the scheduling parameter $\delta = \frac{|\varphi^{-1}(\A)|}{|\varphi^{-1}(B)|}$, which is also the optimization parameter in the above theorem.

\begin{figure*}[!t]
\centering{\subfloat{\includegraphics[width=0.5\textwidth]{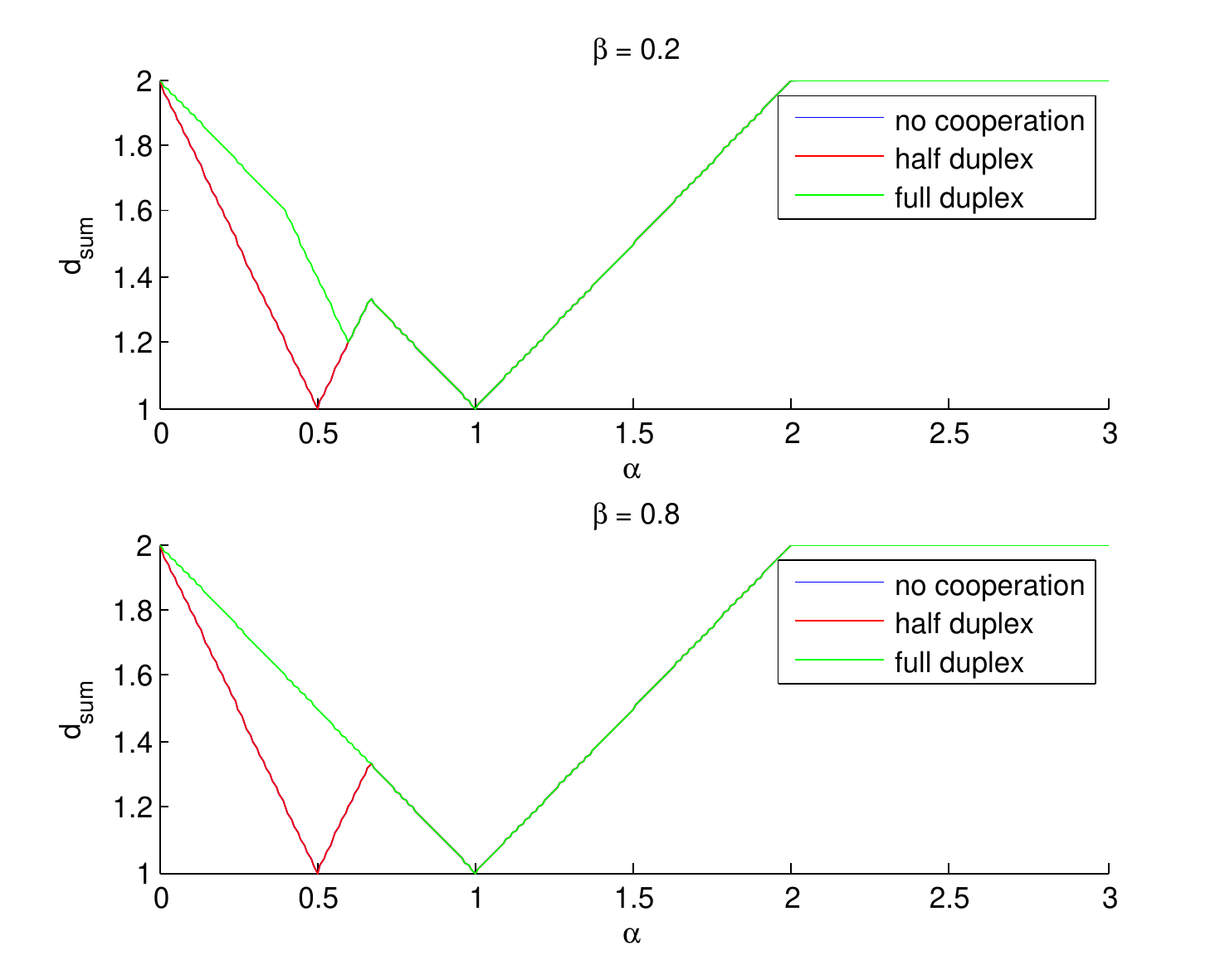}}\hfil \subfloat{\includegraphics[width=0.5\textwidth]{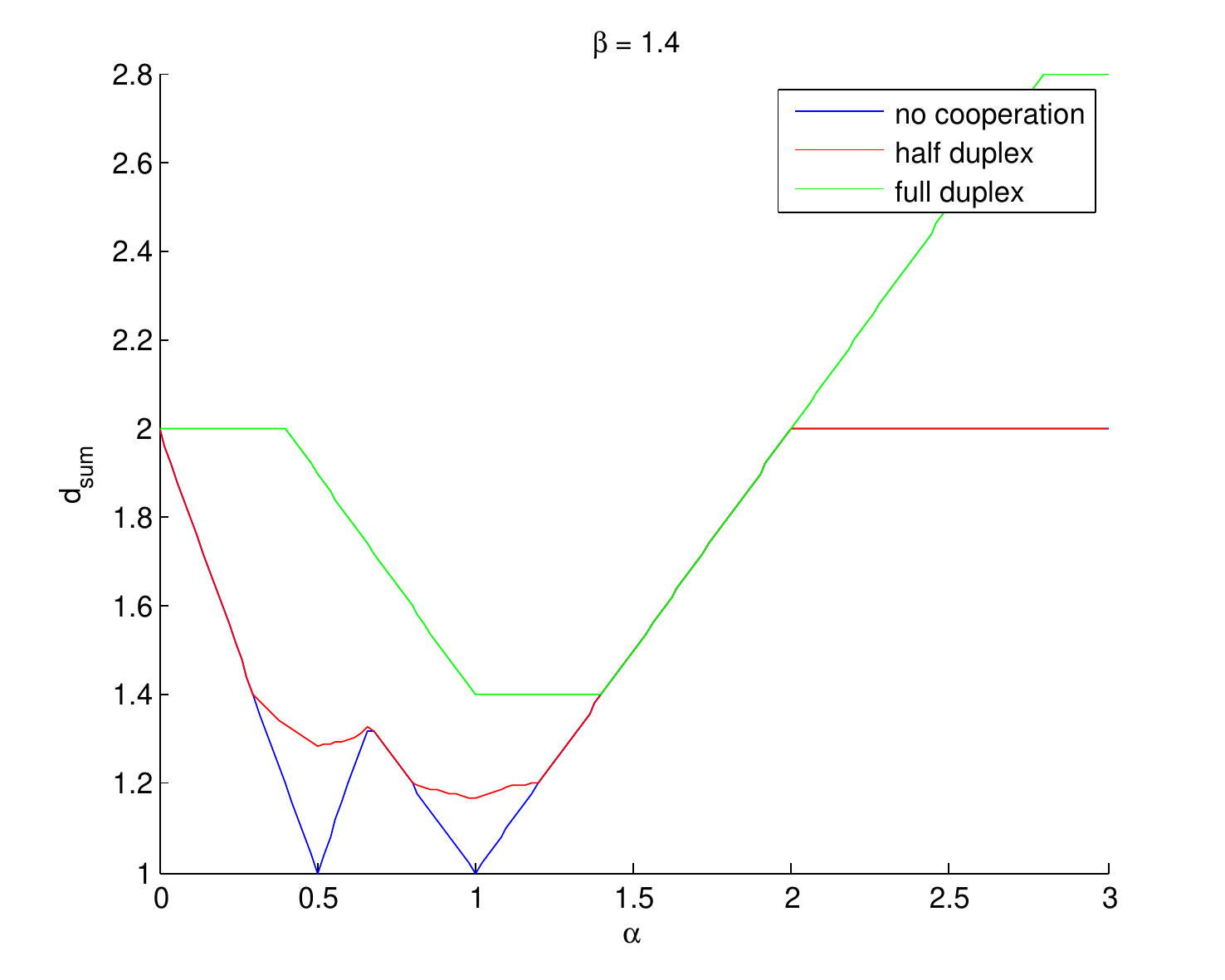}}\\
\subfloat{\includegraphics[width=0.5\textwidth]{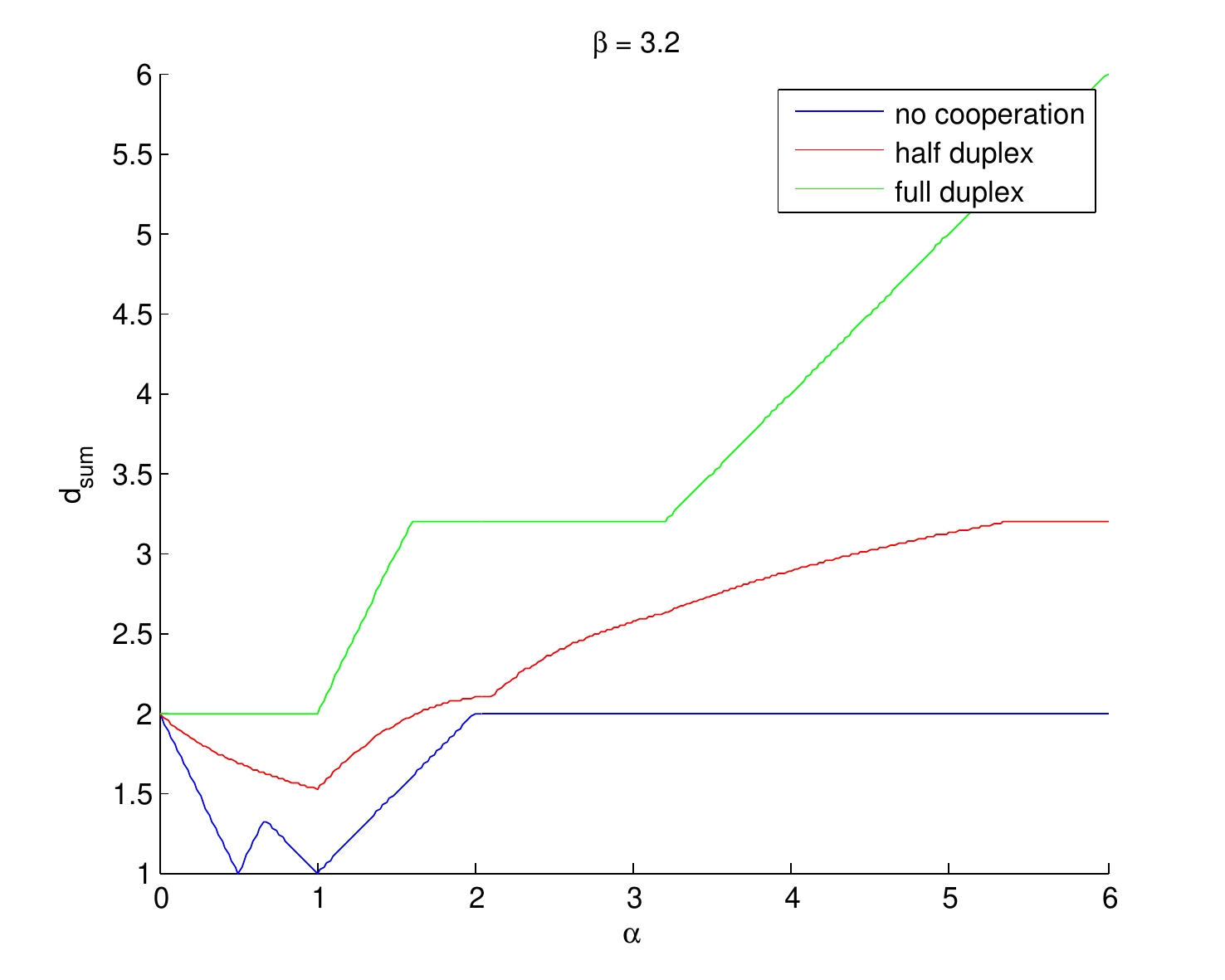}}\hfil \subfloat{\includegraphics[width=0.5\textwidth]{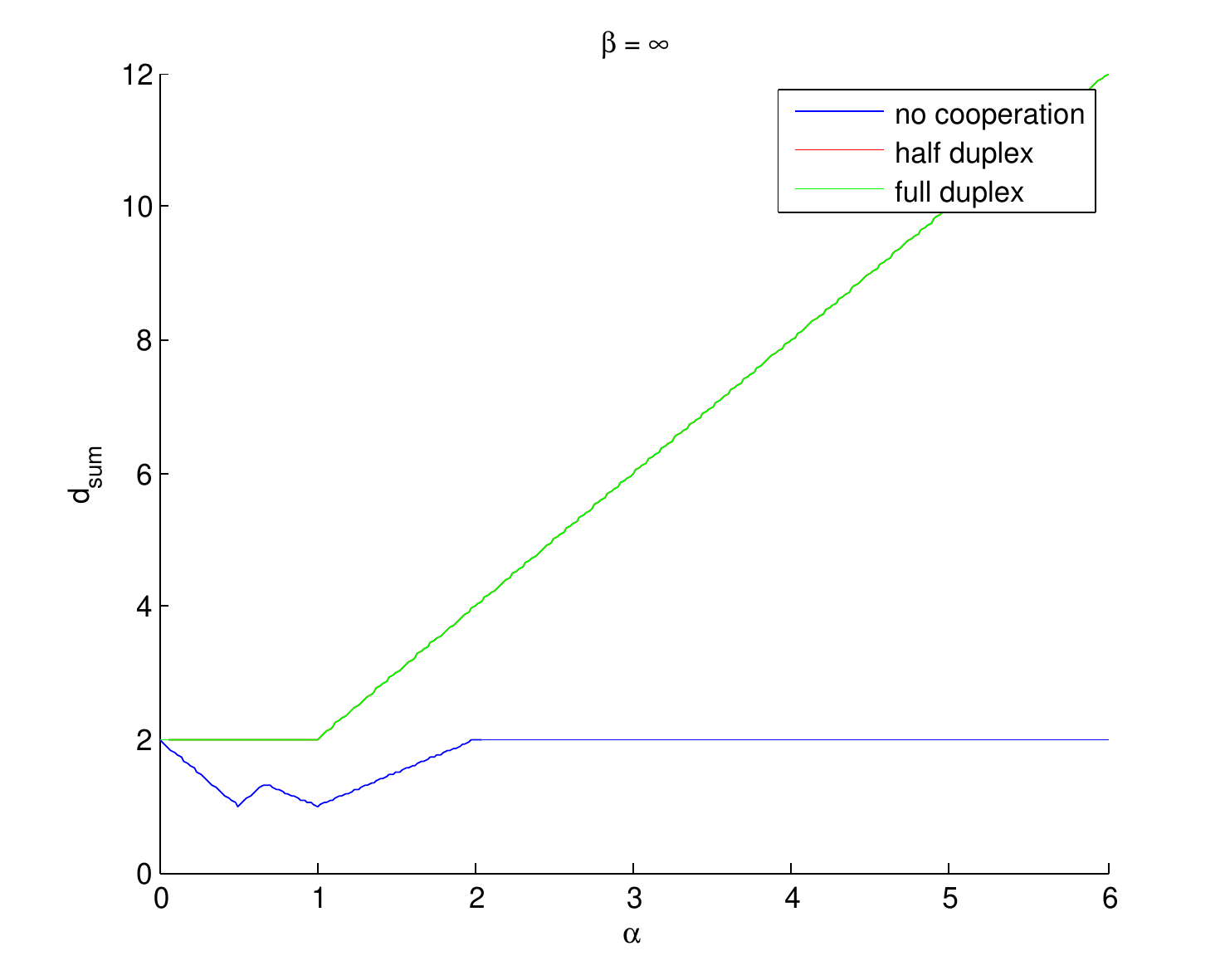}}}\\
\caption{Sum capacity of the interference channel with half-duplex
source cooperation.}\label{Fig:sum_performance}
\end{figure*}

To demonstrate the gains from cooperation, we plot the
generalized degree of freedom \cite{EtkinTseWang} of the sum capacity.
Here we use the natural generalization of the original definition given in
\cite{WT11}. Assume
\begin{align*}
  \lim_{\SNR\to \infty}\frac{\log \INR}{\log \SNR} = \alpha,
  \lim_{\SNR\to \infty}\frac{\log \CNR}{\log \SNR} = \beta.
\end{align*}
Then the generalized degree of freedom for fixed $\alpha, \beta$ is
\begin{align*}
  d_{sum}(\alpha, \beta) = \lim_{\begin{subarray}{c} \text{fix} (\alpha,
  \beta)\\\SNR\to\infty\end{subarray}}\frac{\Ci {sum}}{\log \SNR}.
\end{align*}

Note that $d_{sum}$ is well-defined for $\alpha\ne 1$. When $\alpha=1$,
$d_{sum}$ can take two different values, and we need to treat them
separately.
\begin{enumerate}
  \item $h_{13}h_{24}=h_{14}h_{23}$. Consider the cut-set bound with
  sources on one side and destinations on the other. The upper bound on the sum
  capacity of the interference channel reduces to the capacity of
  a degenerated multiple input multiple output (MIMO) point-to-point channel. As the degree of freedom of the latter channel
  is only 1, therefore we get $d_{sum} = 1$.
  \item $h_{13}h_{24}\ne h_{14}h_{23}$. For this setting, the
  channel is well-conditioned and $d_{sum}$ is a continuous function
  with respect to $\alpha$ at $\alpha=1$.
\end{enumerate}

In Figure~\ref{Fig:sum_performance}, we show plots of $d_{sum}$ against $\alpha$ for different $\beta's$
under the more interesting assumption $h_{13}h_{24}\ne h_{14}h_{23}$. We
also compare it with the result for full-duplex source cooperation \cite{PrabhakaranVi09S}. In
\cite{PrabhakaranVi09S}, the sources are allowed both to listen and transmit at the same
time instant. Under such full-duplex assumption, the channel has only one mode: both sources transmit
and listen. The resulted $d_{sum}$ is a piecewise linear function of $\alpha$. For our
half-duplex channe, however, we need to switch between three modes, and
the optimization over the scheduling parameter $\delta$ makes each piece a smoothed curve rather than a linear function.
From the plots, we first observe that half-duplex cooperation is helpful only when $\beta>1$,
while full-duplex cooperation is helpful for all $\beta>0$. When $\beta$ is large enough
(for example, $\beta=3.2$), the sum capacity of our channel can be strictly better
than that of the usual interference channel. Moreover, when $\beta = \infty$, the sources can
get to know both messages in negligible amount of time with either half-duplex or full-duplex
cooperation. Therefore the channel essentially become a broadcast channel with two antennas at the source, and
the channel with half-duplex source cooperation has the same sum capacity as the channel with full-duplex source cooperation.

\subsection{The Cognitive Case}
The following theorem characterizes the $R_0$-capacity of the cognitive channel within a constant.

\begin{thm}\label{thm:cog}
  Define $\overline{C_{R_0}} = \max_{\delta}\overline{C_{R_0}}(\delta)=
  \max_{\delta}\min(u_1, u_2, u_3, u_4)$, where
  \begin{align*}
    u_1= & \frac{1}{1+\delta}\log(1+\SNR_2)+1\\
    u_2= &\frac{1}{1+\delta}\Big[\log (1+2\SNR_2+2\INR_2)-\log(1+\SNR_1)+\delta\log(1+\frac{\INR_2+\CNR}{1+\SNR_1})\\&+\log(1+\frac{\SNR_1}{1+\INR_2})\Big]+2+R_0\\
    u_3= &\frac{1}{1+\delta}\left[\log(1+2\SNR_1+2\INR_1)- \log(1+\SNR_1)+\log (1+\frac{\SNR_2}{1+\INR_1})\right]+2+R_0\\
    u_4= &\frac{1}{1+\delta}\Big[\log(1+2\SNR_1+2\INR_1)-\log(1+\SNR_1)+\log(1+\frac{\SNR_1}{1+\INR_2})-\log(1+\SNR_1)\\
    &+\max(\log(1+\INR_2+\frac{2\SNR_2+\INR_2}{1+\INR_1}), \log(1+2\INR_2))+\delta
    \log(1+\frac{\INR_2+\CNR}{1+\SNR_1})\Big]\\&+3+2R_0.
  \end{align*}
  Then when $R_0\geq 7$, the $R_0$-capacity $C_{R_0}$ of the cognitive channel defined in
  section~\ref{sec:problem_cog}
  satisfies $\overline{C_{R_0}}-23-2R_0\leq C_{R_0}\leq \overline{C_{R_0}}$.
\end{thm}

In the coding scheme, we define the scheduling parameter $\delta$ to be the ratio of the number of time slots spent in mode $\B$ and $\A$, i.e., $\delta = \frac{|\varphi^{-1}(\B)|}{|\varphi^{-1}(\A)|}$. We note that this definition of $\delta$ is a little bit different from the one for the
symmetric case, as it is now proportional to $\varphi(B)$, which is more convenient for presenting the result.

To demonstrate the gains from cooperation, we plot the
generalized degree of freedom \cite{EtkinTseWang} of the $R_0$-capacity.
To be consistent in notations with later sections, we consider a reference $\SNR$ that goes to infinity, and assume
\begin{align*}
  &\lim_{\SNR\to \infty}\frac{\log \SNR_1}{\log \SNR} = n_1,
  \lim_{\SNR\to \infty}\frac{\log \SNR_2}{\log \SNR} = n_2,\\
  &\lim_{\SNR\to \infty}\frac{\log \INR_1}{\log \SNR} = \alpha_1,
  \lim_{\SNR\to \infty}\frac{\log \INR_2}{\log \SNR} = \alpha_2,\\
  &\lim_{\SNR\to \infty}\frac{\log \CNR}{\log \SNR} = \beta.
\end{align*}
For the discussion with generalized degree of freedom in this section, we simply take $\SNR = \SNR_1$ thus $n_1=1$ . Then the generalized degree of freedom for given $n_2, \alpha_1, \alpha_2, \beta$ is
\begin{align*}
  d_{\text{\sf cog}}(n_2, \alpha_1, \alpha_2, \beta) = \lim_{\SNR\to \infty}\frac{C_{R_0}}{\log
  \SNR}.
\end{align*}
Unlike the symmetric case, this limit always exists, i.e., $d_{\text{\sf cog}}$ is always well-defined.
In particular, when $|h_{13}||h_{24}| = |h_{23}||h_{14}|$, $d_{\text{\sf cog}}$ is the same as that of
an interference channel without cooperation, which is essentially saying that cooperation
is not quite helpful even when the absolute value of the channel gains are aligned. Phases do
not matter here. Moreover, $d_{\text{\sf cog}}$ is continuous when the channel gains are close to being aligned.
Figure~\ref{Fig:cog_performance} shows two typical plots of $d_{\text{\sf cog}}$ against $\alpha_1$ for various
$\beta$ while $n_2, \alpha_2$ are held fixed.

In our model, $\beta=0$ corresponds to an interference channel without
cooperation. The above plot shows that when $\beta\leq \alpha_2\vee n_1$, where $x\vee y = \max(x, y)$, the
generalized degree of freedom is the same as that of $\beta = 0$. Hence, cooperation is
not very helpful unless it is above the threshold. This behavior is the same as what happens to the symmetric
channel case. On the other hand, when $\beta= \infty$, the cooperation link is so strong
that the secondary can decode the primary's message in a negligible amount of time. This
case is equivalent to the cognitive radio channel model in \cite{JovicicViswanath09}, where
the secondary is assumed to know both messages. One other interesting thing to notice is that
when $n_2\leq \alpha_1\leq n_1$, $d_{\text{\sf cog}}$ is always 0, even with infinite cooperation.
This is because in this region, what destination 4 gets from source 2 is only a noisy version of
what destination 3 gets from source 2, which implies that destination $3$ can further decode $W_2$ after decoding $W_1$.
Since we require the primary to achieve a rate near its link capacity, the rate allowed for $W_2$ must be at most a constant in the high-SNR region.

\begin{figure*}[!t]
\centering{\subfloat{\includegraphics[width=0.5\textwidth]{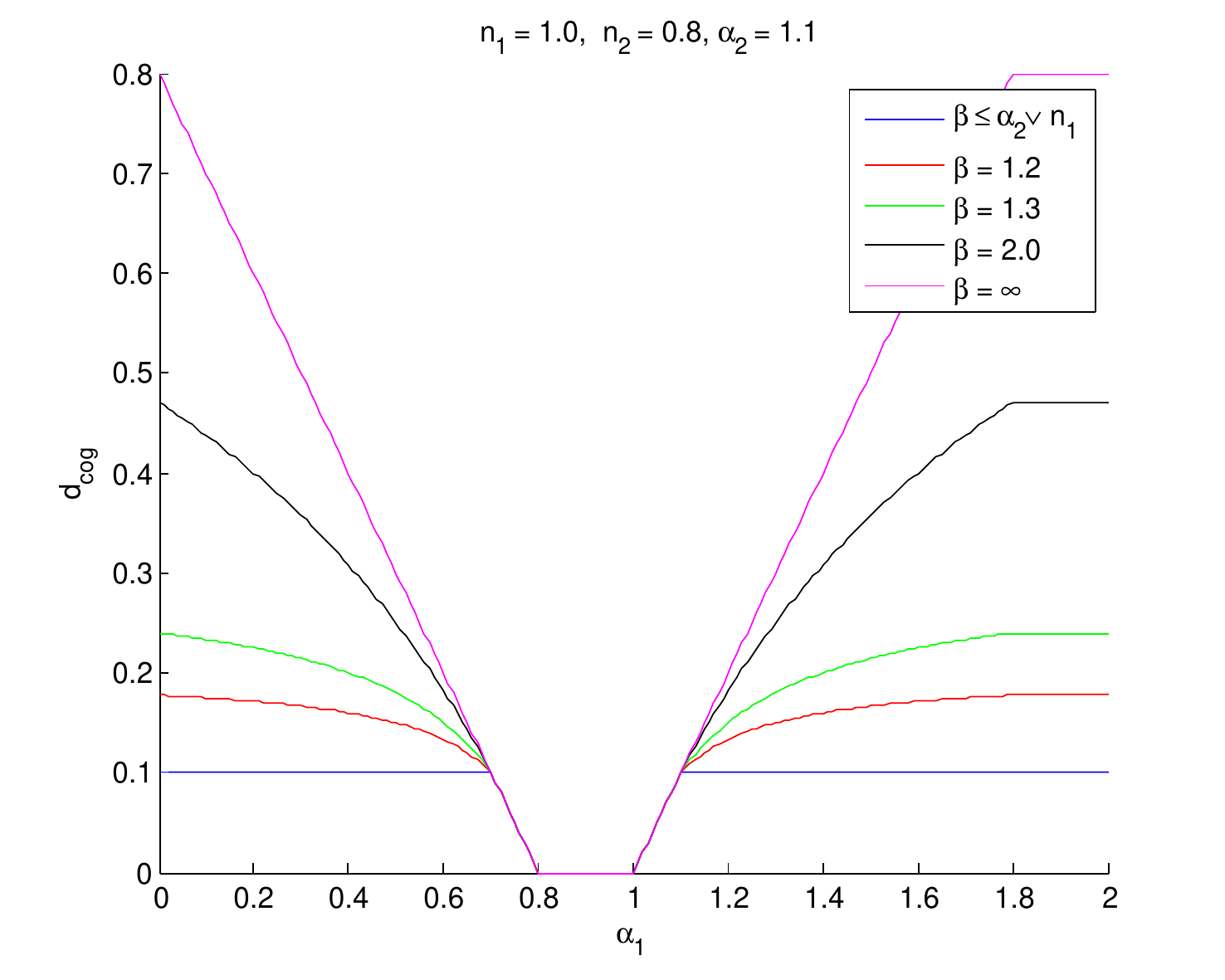}}\hfil \subfloat{\includegraphics[width=0.5\textwidth]{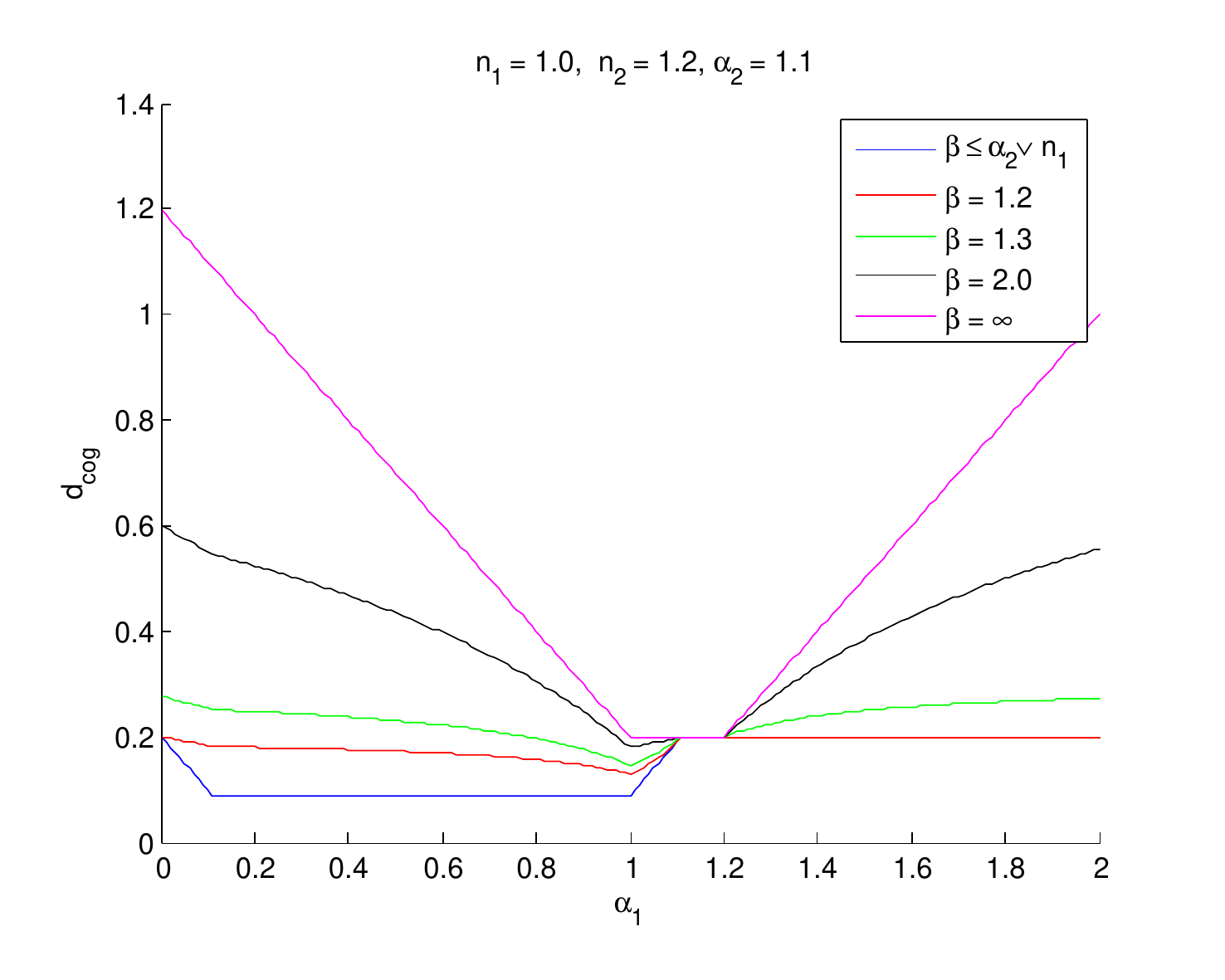}}}\\
\caption{Cognitive capacity of the interference channel with half-duplex
source cooperation.}\label{Fig:cog_performance}
\end{figure*}

\section{Achievability}\label{sec:achievability}
Our coding scheme turns the two-user interference channel of mode $\A$
(Fig.~\ref{Fig:hd_channel_A}) into a virtual two-user interference channel
(Fig.~\ref{Fig:IFconf}), with rate-limited (noiseless) bit-pipes between the
two sources and from each source to the destination node where it causes
interference.  The bit-pipes are realized by operating in modes $\B$ and
$\C$ (Fig.\ref{Fig:hd_channel_B} and \ref{Fig:hd_channel_C}) where only one
of the source nodes transmits while the other receives. In these modes, in addition to sending data to its own destination, the
transmitting source sends messages to the other nodes as well to establish the noiseless links.
In this section, we first describe a coding scheme
and characterize an achievable rate region for the virtual channel. Then we
will use this characterization to obtain an achievable rate region for the
two-user interference channel with half-duplex source cooperation. We note that it is possible to obtain an achievable scheme using strategies in \cite{PrabhakaranVi09S, YangTuninetti11},
However, we do not pursue this route in this paper, as it is as complicated specializing known schemes as describing our coding scheme.

\subsection{Interference Channel with Bit-pipes}

We denote the {\em virtual channel} in Fig.~\ref{Fig:IFconf} by
$\IFconf(p_{Y_3,Y_4|X_1,X_2},\bp 12,\bp 21,\bp 14,\bp 23)$, where $\bp ij$
are the rates of the bit-pipes between node $i\in\{1,2\}$ and node
$j\in\{1,2,3,4\}$.
The virtual channel is converted from mode $\A$, therefore it lasts for the duration of mode $\A$. With a little abuse of notation,
we assume the communication is over discrete time slots $t = 1, \dots, L$ in this subsection for simplicity.
For this new channel, we limit ourselves to
block-coding schemes of the following type:
\begin{enumerate}
  \item First, the sources send at most $L \bp ij$ bits over the bit-pipes,
where $L$ is the blocklength. These bits are functions only of the message
of the source sending the bits.
  \item Then, the sources transmit over the interference channel with each
of their channel inputs (of blocklength $L$) being functions of their
message and the bits exchanged in the first step. For the Gaussian channel,
these transmissions are required to satisfy average power
constraints of unity.
\end{enumerate}

\begin{figure}
\centering\scalebox{0.5}{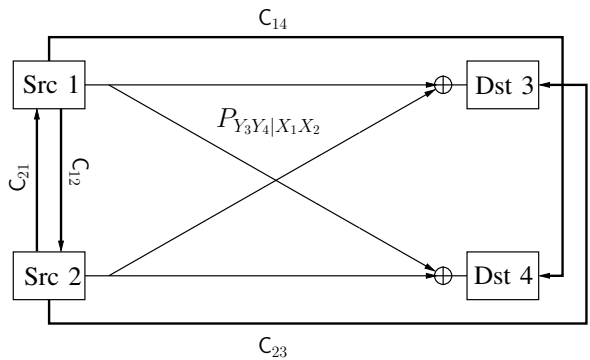}
\caption{Interference channel with bit-pipes. The rate-limited bit-pipes (shown in bold) run between the two sources and from each
source to the destination node where it causes interference.}
\label{Fig:IFconf}
\end{figure}

A rate pair $(R_1,R_2)$ is defined to be {\em achievable} for this channel
along the same lines as in Section~\ref{sec:problem}.  In the rest of the section,
we first discuss an achievable region ${\mathcal R}_\text{\sf virtual}(\bp
12,\bp 21, \bp 14, \bp 23)$ for the virtual channel\footnote{We drop the
channel $p_{Y_3,Y_4|X_1,X_2}$ from the notation since the channel will be
clear from the context.}. Then using this result, an achievable region for
the half-duplex channel will be presented.

Our coding scheme for this virtual channel is a
generalization of the superposition coding scheme given by Han and Kobayashi
for interference channels.
The scheme of Han and
Kobayashi in this context may be interpreted as follows. Each source
node transmits its information in two parts:
\begin{itemize}
\item {\em public message} is decoded by both destinations (even though it
is meant for only one of the destinations),
\item {\em private message} is decoded only by one of the destinations, the
one to which it is intended.
\end{itemize}
Our scheme also uses superposition coding and involves two additional parts
each of which takes advantage of one of the two types of bit-pipes available.
\begin{enumerate}
\item {\em cooperative private message.} These messages are shared in
advance between the sources over the bit-pipes between them. The
messages are then sent out cooperatively by the two sources. But
they are only decoded by the intended destination. Below, we will use
superposition coding and beamforming for transmitting these messages.
\item {\em pre-shared public message.} Each source shares this type of
message with the unintended destination in advance over the bit-pipes to
that destination. This ensures that when it appears as interference in the
transmission over the interference channel, the destination can treat it as
known interference while decoding.
\end{enumerate}
In slightly greater detail, our coding scheme is as follows: We fix
the input distribution
\begin{align*}
  &p(x_{V_1'}, x_{V_2'}, x_{W_1}, x_{W_2}, x_{U_1}, x_{U_2}, v_1, v_2, x_{V_1}, x_{V_2}, x_1, x_2)\\
  = &p(x_{V_1'})p(x_{V_2'})p(x_{W_1}|x_{V_1'})p(x_{W_2}|x_{V_2'})p(x_{U_1}|x_{W_1}, x_{V_1'})p(x_{U_2}|x_{W_2}, x_{V_2'})\\
  & p(v_1)p(v_2)p(x_{V_1}, x_{V_2}|v_1, v_2)p(x_1|x_{U_1}, x_{V_1})p(x_2|x_{U_2}, x_{V_2}).
\end{align*}
\noindent{\em Codebook construction and encoding:} Source $i\in\{1,2\}$ divides its message into
four parts $m_i = (m_{W_i}, m_{U_i}, m_{V_i},m_{V_i'})$, where $W$ stands
for (noncooperative) {\em public}, $U$ for (noncooperative) {\em private},
$V$ for {\em cooperative private}, and $V'$ for {\em pre-shared public}.
First, $m_{V_i}$ is shared with the other source and $m_{V_i'}$ is shared
with the other destination over the bit-pipes. Superposition codewords are
then transmitted over the interference channel. A random codebook
construction for these codewords is as follows:
\begin{enumerate}
  \item At source $i\in\{1,2\}$, generate the pre-shared public codeword
$X_{V_i'}^L(m_{V_i'})$ independently according to distribution
$p(x_{V_i'}^L)=\prod_{t=1}^{L}p(x_{V_i', t})$, where
$m_{V_i'}\in\{1,2,\ldots,2^{L(R_{V_i'}-\epsilon)}\}$.
  \item At source $i$, for each $m_{V_i'}$, generate the public codeword
$X_{W_i}^L(m_{V_i'},m_{W_i})$ independently according to distribution
$p(x_{W_i}^L|x_{V_i'}^L(m_{V_i'}))=
\prod_{t=1}^{L}p(x_{W_i, t}|x_{V_i', t}(m_{V_i'}))$, where
$m_{W_i}\in\{1,2,\ldots,2^{L(R_{W_i}-\epsilon)}\}$.
  \item At source $i$, for each pair of $(m_{W_i},m_{V_i'})$,
generate the private codeword $X_{U_i}^L(m_{U_i},m_{W_i},m_{V_i'})$
according to distribution
$p(x_{U_i}^L|x_{W_i}^L(m_{W_i},m_{V_i'}),x_{V_i'}^L(m_{V_i'}))=
\prod_{t=1}^{L}p(x_{U_i, t}|x_{W_i, t}(m_{W_i},m_{V_i'}),x_{V_i', t}(m_{V_i'}))$, where $m_{U_i}\in\{1,2,\ldots,2^{L(R_{U_i}-\epsilon)}\}$.
  \item Generate, for $i\in\{1,2\}$, the {\em auxiliary} cooperative private
codewords $\Ve_i^L({m}_{V_i})$, according to distribution
$p_{\ve_i^L}=\prod_{t=1}^Lp(\ve_{i, t})$, where
$m_{V_i}\in\{1,2,\ldots,2^{L(R_{V_i}-\epsilon)}\}$. For every pair
$(m_{V_1},m_{V_2})$, define the cooperative private codewords
$(X_{V_1}^L, X_{V_2}^L)(m_{V_1}, m_{V_2})$ according to distribution
$$p(x_{V_1}^L, x_{V_2}^L|\ve_1^L(m_{V_1}), \ve_2^L(m_{V_2}))=
\prod_{t=1}^{L}p(x_{V_1, t}, x_{V_2, t}| \ve_{1, t}(m_{V_1}),
 \ve_{2, t}(m_{V_2})).$$
  \item At source $1$, generate the codewords to be transmitted $X_1^L(m_{W_1},m_{U_1},m_{V_1'},m_{V_1},m_{V_2})$ according to distribution
  \begin{align*}
    p(x_1^L|x_{U_1}^L(m_{U_1},m_{W_1},m_{V_1'}), x_{V_1}^L(m_{V_1},m_{V_2})) = \prod_{t = 1}^L p(x_{1, t}|x_{U_1, t}(m_{U_1},m_{W_1},m_{V_1'}), x_{V_1, t}(m_{V_1},m_{V_2})).
  \end{align*}
  At source $2$, generate $X_2^L(m_{W_2},m_{U_2},m_{V_2'},m_{V_2},m_{V_1})$ similarly.

\end{enumerate}

\noindent{\em Decoding:} Destination~3 looks for a unique
$(m_{W_1}, m_{U_1}, {m}_{V_1}, m_{V_1'})$ such that
\begin{align*}
  (Y_3^L, X_{V_1'}^L(m_{V_1'}), X_{W_1}^L(m_{W_1},m_{V_1'}),
X_{U_1}^L(m_{U_1},m_{W_1},m_{V_1'}), \Ve_1^L({m}_{V_1}),
X_{W_2}^L(\hat{m}_{W_2}),X_{V_2'}^L(m_{V_2'}))
\end{align*}
is jointly typical, for some $\hat{m}_{W_2}$. Note that $m_{V_2'}$ is
available to destination~3 via the bit-pipe from source~2.  Destination~4
uses the same decoding rule with index 1 and 2 exchanged.

\begin{thm}\label{thm:IFconf}
The rate pair $(R_{W_1}+R_{U_1}+R_{V_1}+R_{V_1'},
R_{W_2}+R_{U_2}+R_{V_2}+R_{V_2'})$ is achievable if
$R_{W_1},R_{W_2},R_{U_1},R_{U_2},R_{V_1},R_{V_2},R_{V_1'},R_{V_2'}$ are
non-negative reals which satisfy the following constraints.

Constraints at destination $3$: 
\begin{align*}
  R_{V_1'}&\leq  \bp 14\\
  R_{U_1}&\leq I(X_{U_1}; Y_3 | X_{W_1}, \Ve_1, X_{V_1'},
X_{W_2}, X_{V_2'})\\
  R_{W_1}+R_{U_1}&\leq I(X_{W_1}, X_{U_1}; Y_3 | \Ve_1, X_{V_1'},
X_{W_2}, X_{V_2'})\\
  R_{V_1'}+R_{W_1}+R_{U_1}&\leq I(X_{W_1}, X_{U_1}, X_{V_1'}; Y_3 | \Ve_1, X_{W_2}, X_{V_2'})\\
  R_{V_1}&\leq I(\Ve_1; Y_3 | X_{W_1}, X_{U_1}, X_{V_1'},
X_{W_2}, X_{V_2'})\\
  R_{V_1}+R_{U_1}&\leq I(X_{U_1}, \Ve_1; Y_3 | X_{W_1}, X_{V_1'},
X_{W_2}, X_{V_2'})\\
  R_{V_1}+R_{W_1}+R_{U_1}&\leq I(X_{W_1}, X_{U_1}, \Ve_1; Y_3 | X_{V_1'},
X_{W_2}, X_{V_2'})\\
  R_{V_1}+R_{V_1'}+R_{W_1}+R_{U_1}&\leq I(X_{W_1}, X_{U_1}, \Ve_1, X_{V_1'}; Y_3 |
X_{W_2}, X_{V_2'})\\
  R_{W_2}+R_{U_1}&\leq I(X_{W_2}, X_{U_1}; Y_3 | X_{W_1}, \Ve_1, X_{V_1'},
X_{V_2'})\\
  R_{W_2}+R_{W_1}+R_{U_1}&\leq I(X_{W_2}, X_{W_1}, X_{U_1}; Y_3 | \Ve_1, X_{V_1'},
X_{V_2'})\\
  R_{W_2}+R_{V_1'}+R_{W_1}+R_{U_1}&\leq I(X_{W_2}, X_{W_1}, X_{U_1}, X_{V_1'}; Y_3 |\Ve_1,
X_{V_2'})\\
  R_{W_2}+R_{V_1}&\leq I(X_{W_2}, \Ve_1; Y_3 | X_{W_1}, X_{U_1}, X_{V_1'},
X_{V_2'})\\
  R_{W_2}+R_{V_1}+R_{U_1}&\leq I(X_{W_2}, X_{U_1}, \Ve_1; Y_3 | X_{W_1}, X_{V_1'},
X_{V_2'})\\
  R_{W_2}+R_{V_1}+R_{W_1}+R_{U_1}&\leq I(X_{W_2}, X_{W_1}, X_{U_1}, \Ve_1; Y_3 | X_{V_1'},
X_{V_2'})\\
  R_{W_2}+R_{V_1}+R_{V_1'}+R_{W_1}+R_{U_1}&\leq
     I(X_{W_1}, X_{U_1}, X_{V_1'}, \Ve_1, X_{W_2}; Y_3 | X_{V_2'}).
\end{align*}

Constraints at destination $4$: Above with index $1, 2$ exchanged and index $3, 4$ exchanged.

Constraints at sources:
\begin{align*}
  R_{V_1}\leq \bp 12, \quad R_{V_2}\leq \bp 21.
\end{align*}
for some
\begin{align*}
  p( x_{W_1}, x_{U_1}, x_{V_1}, x_{V_1'}, x_{W_2}, x_{U_2}, x_{V_2},
x_{V_2'}, \ve_1, \ve_2)
&= p(x_{V_1'},x_{W_1},x_{U_1})p(x_{V_2'},x_{W_2},x_{U_2})
p(\ve_1)p(\ve_2)p(x_{V_1},x_{V_2}|\ve_1, \ve_2).
\end{align*}
For the Gaussian channel, the joint distribution must satisfy
\begin{align*}
\Var{X_{U_i}}+\Var{X_{V_i}} \leq 1,\qquad i\in\{1,2\}.
\end{align*}

\end{thm}
We denote this rate region by ${\mathcal R}_\text{\sf virtual}(\bp 12, \bp
21, \bp 14, \bp 23)$.
\begin{proof} The proof is omitted since it follows from standard
arguments for superposition coding.\end{proof}

\subsection{Achievablity for Half-Duplex
Channel}\label{sec:ach.half-duplex}

Now we give a scheme for the original channel. The rate region will be
given in terms of ${\mathcal R}_\text{\sf virtual}$ in
Theorem~\ref{thm:IFconf}. Our coding scheme consists of a sequence of
blocks.  Each block is $\ceil{\d \A L} + \ceil{\d \B L} + \ceil{\d \C L}$
long $(\d \A, \d \B, \d \C\geq 0)$. Let us denote, $L_\A=\ceil{\d \A L}$,
$L_\B=\ceil{\d \B L}$ and $L_\C=\ceil{\d \C L}$.  In each block, the first
$1,2,\ldots,L_\B$ and $L_\B+1,L_\B+2,\ldots,L_\B+L_\C$, respectively are
operated in modes $\B$ and $\C$ respectively. The rest $L_\A$ long duration
is in mode $\A$. During mode~$\B$ and $\C$ of each block, we
will realize the bit-pipes of the virtual channel. This will allow us to
implement our coding scheme for the virtual channel during mode~$\A$. In {\em addition} to realizing the virtual channels, modes~$\B$ and $\C$ also involve communication of additional data directly to the intended destination as well as by relaying through the other source node as explained next.

Notice that in mode~$\B$ (resp. $\C$), we have a broadcast channel with source node 1 (resp. 2) as the sender and three receivers, namely, the two destinations nodes 3 \& 4 and the other souce node 2 (resp. 1). We describe mode~$\B$; mode~$C$ is symmetric.
In addition to realizing the bit-pipes of the virtual channel,
during mode~$\B$, the source node 1
\begin{enumerate}
\item[(i)] sends data to its own destination node 3, and
\item[(ii)] implements a simple block Markov decode-and-forward scheme in conjunction with source node 2 by
 (a) sending data to the other source node 2 to be {\em relayed} by source node 2 to
the intended destination node 3 in mode~$\C$ of the next block, and (b) relaying data
received from the other source node 2 during mode~$\C$ of the previous block to its
intended destination node 4.
\end{enumerate}

In mode $\B$, source node~1 uses superposition coding to send
messages to each of the other nodes. In particular, it sends at a rate of
$R_{1\B}$ to destination~3, at a rate $\frac{\d \A}{\d \B}\bp 12 + \rly
123$ to the other source (node~2) and at a rate of $\frac{\d \A}{\d \B}\bp
14 + \frac{\d \C}{\d \B} \rly 214$ to destination node~4. The transmissions
at rates $\frac{\d \A}{\d \B}\bp 12$ and $\frac{\d \A}{\d \B} \bp 14$ are
used to realize the bit-pipes originating from source node~1 to nodes~2 and
4, respectively in the virtual channel. Similarly, source node $2$ realizes the bit-pipes to the other nodes in mode~$\C$. With these
bit-pipes in place, the channel in the following mode~$\A$ is effectively
transformed into the virtual channel we described before. The transmission at rate $\rly 123$ is
meant to be relayed on by source node~2 to destination node~3 in the
following mode $\C$. And the transmission at rate $\frac{\d \C}{\d
\B} \rly 214$ is of the data node~1 received from source node~2 in mode
$\C$ of the previous block that is intended to be relayed to destination
node~4. Similarly, in mode $\C$, source node~2 sends using
superposition coding at rates $R_{2\C}$, $\frac{\d \A}{\d \C} \bp 21 + \rly
214$, and $\frac{\d \A}{\d \C} \bp 23 + \frac{\d \B}{\d \C} \rly 123$ to
nodes~4, 1, and 3, respectively. Note that in mode $\B$ for the first block, there is no relay data available for node~1 to relay to
node~4. But, by increasing the number of blocks, the resulting deficit in
rate can be made as small as desired.

For the degraded broadcast channel of mode $B$ (resp. $C$), we will use the natural ordering of users for
superposition coding-successive cancellation decoding, i.e., the strongest user's message is superposed on the codeword resulting from superposing the next stronger user's message on the weakest user's codeword. To denote all possibilities together, we adopt the
following notation. Let
\begin{align*}
\spR \B 3 &= R_{1\B}, & \spR \C 4 &= R_{2\C},\\
\spR \B 2 &= \frac{\d \A}{\d \B}\bp 12 + \rly 123, &
\spR \C 1 &= \frac{\d \A}{\d \C}\bp 21 + \rly 214,\\
\spR \B 4 &= \frac{\d \A}{\d \B}\bp 14 + \frac{\d \C}{\d \B}\rly 214,
\text{ and} &
\spR \C 3 &= \frac{\d \A}{\d \C}\bp 23 + \frac{\d \B}{\d \C}\rly 123.
\end{align*}
Then, by superposition coding, the above rates are achievable if there are
permutations $\phi^{\B}$ of $\{2,3,4\}$ and $\phi^{\C}$ of $\{1,3,4\}$, and
a joint distribution\\ $p(\spu \B 1)p(\spu \B 2)p(\spu
\B 3)p(x_1|\spu \B 1,\spu \B 2, \spu \B 3)p(\spu \C 1)p(\spu \C 2)p(\spu
\C 3)p(x_2|\spu \C 1,\spu \C 2, \spu \C 3)$, (which satisfies the condition
$\Var{X_1}$, $\Var{X_2} \leq 1$ for the Gaussian case) such that
\begin{align}
\sum_{j=1}^i \spR \B {\phi^{\B}(j)} \leq I(\spU \B 1,\ldots,\spU \B i;Y_{\phi^{\B}(i)}),
\qquad i\in\{1,2,3\}, \label{eq:spcond1}\\
\sum_{j=1}^i \spR \C {\phi^{\C}(j)} \leq I(\spU \C 1,\ldots,\spU \C i;Y_{\phi^{\C}(i)}),
\qquad i\in\{1,2,3\}. \label{eq:spcond2}
\end{align}
Note that, for a given channel, we will use only the permutations $\phi^{\B},\phi^{\C}$ corresponding to the natural ordering described above. Also, note that the $\tilde{U}^B$'s are auxiliary random variables corresponding to the messages superposition coded in mode~$\B$ (similary, $\tilde{U}^C$ for mode~$\C$).
Thus, we have proved the following theorem:
\begin{thm}\label{thm:half-duplex}
The rate pair $(R_1,R_2)$ is achievable for the half-duplex channel, where
\begin{align*}
R_1=\frac{\d \A R_{1\A} + \d \B R_{1\B} + \d \B \rly 123 }{\d \A+\d \B + \d \C},\\
R_2=\frac{\d \A R_{2\A} + \d \C R_{2\C} + \d \C \rly 214 }{\d \A+\d \B + \d \C},
\end{align*}
for parameters as defined in the above discussion such that
\eqref{eq:spcond1}-\eqref{eq:spcond2} hold and
\[ (R_{1\A},R_{2\A})\in {\mathcal R}_\text{\sf virtual}(\bp 12,\bp 21,\bp
14, \bp 23).\]
\end{thm}

\section{The Symmetric Case: LDM}\label{sec:sym_LDM}

In this section, we study the linear deterministic model (LDM) of the symmetric half
duplex source cooperation problem, and characterize the sum capacity for this LDM. In particular, there is a natural way to divide this problem into several different parameter regions, and in each region we explicitly characterize how the achievable scheme allocates rates for various messages. For the Gaussian model in the next section, we will divide the problem into parameter regions that correspond to the regions for the LDM. Our achievable scheme for the Gaussian model in each region mostly follows from the intuition we gain from the LDM.

\subsection{Channel Model and Sum Capacity}
The linear deterministic channel \cite{ADT08} corresponding to the symmetric case is parameterized by nonnegative integers
\begin{align*}
  n_D = \lfloor\log \SNR\rfloor^+, n_I=\lfloor\log \INR\rfloor^+, n_C=\lfloor\log
\CNR\rfloor^+.
\end{align*}
The channel is depicted in Figure~\ref{Fig:HD_channel_LDM}. Let $S_n$ be the shift matrix in $\mathbb{F}_2^{n \times n}$, where $\mathbb{F}_2$ is the finite field with two elements, i.e., 

\begin{align*}
S_n = \left[
\begin{array}{ccccc}
0&0&0&\cdots&0\\
1&0&0&\cdots&0\\
0&1&0&\cdots&0\\
\vdots&\ddots&\ddots&\ddots&\vdots\\
0&\cdots&0&1&0
\end{array}
\right]_{n\times n}. 
\end{align*}

\begin{figure}
  \centering\scalebox{0.6}{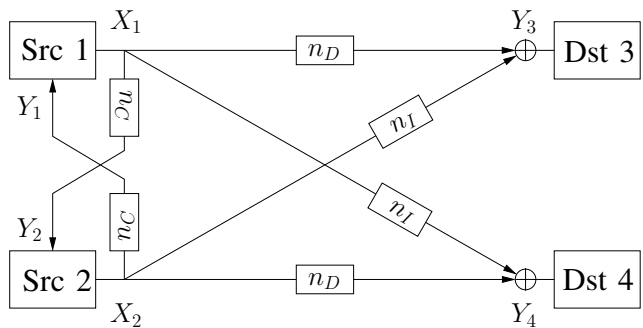}
  \caption{Linear deterministic interference channel with half-duplex source cooperation.}\label{Fig:HD_channel_LDM}
\end{figure}

The sources can work in one of the three modes. In mode $\A$, both sources transmit and the channel inputs $X_{1t}, X_{2t}$ are in $\mathbb{F}_2^{\max\{n_D, n_I\}}$. The nodes receive:
\begin{align*}
Y_{1t} &= 0,\\
Y_{2t} &= 0,\\
Y_{3t} &= S_{\max\{n_D, n_I\}}^{\max\{n_D, n_I\}-n_D}X_{1t}\oplus S_{\max\{n_D, n_I\}}^{\max\{n_D, n_I\}-n_I}X_{2t},\\
Y_{4t} &= S_{\max\{n_D, n_I\}}^{\max\{n_D, n_I\}-n_D}X_{2t}\oplus S_{\max\{n_D, n_I\}}^{\max\{n_D, n_I\}-n_I}X_{1t}.
\end{align*}

In mode $\B$, source 2 listens and the channel inputs $X_{1t}, X_{2t}$ are in $\mathbb{F}_2^{\max\{n_D, n_I, n_C\}}$. Then,
\begin{align*}
Y_{1t} &= 0,\\
Y_{2t} &= S_{\max\{n_D, n_I, n_C\}}^{\max\{n_D, n_I, n_C\}-n_C}X_{1t},\\
Y_{3t} &= S_{\max\{n_D, n_I, n_C\}}^{\max\{n_D, n_I, n_C\}-n_D}X_{1t},\\
Y_{4t} &= S_{\max\{n_D, n_I, n_C\}}^{\max\{n_D, n_I, n_C\}-n_I}X_{1t}.
\end{align*}

In mode $\C$, source 1 listens and the channel inputs $X_{1t}, X_{2t}$ are in $\mathbb{F}_2^{\max\{n_D, n_I, n_C\}}$. Then,
\begin{align*}
Y_{1t} &= S_{\max\{n_D, n_I, n_C\}}^{\max\{n_D, n_I, n_C\}-n_C}X_{2t},\\
Y_{2t} &= 0,\\
Y_{3t} &= S_{\max\{n_D, n_I, n_C\}}^{\max\{n_D, n_I, n_C\}-n_I}X_{2t},\\
Y_{4t} &= S_{\max\{n_D, n_I, n_C\}}^{\max\{n_D, n_I, n_C\}-n_D}X_{2t}.
\end{align*}

\begin{thm} The sum capacity of the interference
channel in Figure~\ref{Fig:HD_channel_LDM} is
\begin{align*}
C_{sum} = \max_{\delta\geq 0}\min\{l_1(\delta), l_2(\delta),
l_3(\delta), l_4(\delta)\},
\end{align*}
\indent where
\begin{align*}
l_1(\delta) =& \frac{2}{2+\delta}\left(\delta n_D+\max\{n_D,
n_C\}\right),\\
l_2(\delta) =& \frac{1}{2+\delta}(\delta \max\{2n_D-n_I,
n_I\}+n_D\\
&+\max\{n_D, n_I, n_C\}),\\
l_3(\delta) =& \frac{2}{2+\delta}\left(\delta\max\{n_I,
n_D-n_I\}+\max\{n_D, n_I, n_C\}\right)\\
l_4(\delta) =&
\begin{cases}
\frac{2(1+\delta)}{2+\delta}\max\{n_D, n_I\}, &n_D \ne n_I\\
n_D, &n_D = n_I\\
\end{cases}. 
\end{align*}
\end{thm}
The parameter $\delta$ is a scheduling parameter the same as the scheduling parameter used in Theorem~\ref{thm:sumrate}. The proof for the converse of the theorem is similar to that of the Gaussian case and is omitted in this paper. Below we describe the achievable coding scheme for the LDM. Note that when $n_I = n_D$ or $n_C\leq n_D$, the sum capacity reduces to that of the interference channel without cooperation. Hence, it can be achieved with the optimal interference channel scheme. In the following discussions, we assume $n_I\ne n_D$ and $n_C>n_D$.

\subsection{Coding Scheme}

To characterize the sum capacity, it is sufficient to consider only symmetric
schemes. The induced virtual channel is also symmetric. The
symmetric virtual channel has an interference channel determined by
$(n_D, n_I)$ and its bit-pipes have rates \bp 12 = \bp 21 = \bp ss and \bp 14 = \bp 23 = \bp sd. We denote this type of virtual
channel by \IFconf $((n_D, n_I), \bp ss, \bp sd)$.

For simplicity, let $n = \max\{n_D, n_I\}$. For source $i\in \{1, 2\}$, we define the public,
pre-shared and  private  auxiliary random variables $W_i, V_i', U_i$ to be independent random variables on $\mathbb{F}_2^n$. In particular, the public and pre-shared auxiliary random
variables are uniformly distributed over $\mathbb{F}_2^n$. The private auxiliary random variables are uniformly distributed over the set of length $n$ vectors in $\mathbb{F}_2^n$ whose upper $n-(n_D-n_I)^+$ elements are fixed to be 0.
In Theorem~\ref{thm:IFconf}, we set
\begin{align*}
  X_{V_i'} &= V_i',\\
  X_{W_i} &= V_i' + W_i,\\
  X_{U_i} &= V_i' + W_i + U_i.
\end{align*}
Note that the private auxiliary random variable $U_i$ occupies the lower
$(n_D-n_I)^+$ levels so that it does not appear at the other destination.
This is similar to the choice made in~\cite{EtkinTseWang} for the
(non-cooperative) intereference channel.

For the cooperative private codebook, we choose the auxiliary random
variables $V_i, i = 1,2$ independent of each other and all the other
auxiliary random variables, and distributed uniformly over $\mathbb{F}_2^n$.
We choose $(X_{V_1},X_{V_2})$ as deterministic functions of $(V_1,V_2)$
such that
\begin{align*}
  \left[
  \begin{array}{c}
    V_1\\
    V_2
  \end{array}
  \right] = \left[
  \begin{array}{cc}
    S_n^{n-n_D}&S_n^{n-n_I}\\
    S_n^{n-n_I}&S_n^{n-n_D}
  \end{array}
  \right]\left[
  \begin{array}{c}
    X_{V_1}\\
    X_{V_2}
  \end{array}
  \right]
\end{align*}
As the channel matrix is invertible, we can always find such $X_{V_i}$
for arbitrary $V_i$. For the particular choice of $X_{V_i}$, the sources are effectively doing {\em zero-forcing} beamforming such that each destination receives the message $V_i$ intended for it.

Using these definitions, source $i$ sends $X_{U_i}^L+X_{V_i}^L, i = 1, 2$. The
induced channel $p_{Y_3,Y_4|V_1',V_2',W_1,W_2,U_1,U_2,V_1,V_2}$ is as follows:
\begin{align*}
  Y_3 &= S_n^{n-n_D}(W_1+U_1+V_1')+S_n^{n-n_I}W_2+V_1\\
  Y_4 &= S_n^{n-n_D}(W_2+U_2+V_2')+S_n^{n-n_I}W_1+V_2,
\end{align*}
where the unintended pre-shared public signals which the receivers know in advance are removed.
We choose symmetric rates for the four types of messages: i.e.,
$R_{V'_1}=R_{V'_2}=R_{V'}$, and so on. When $n_I<n_D$, the sources only send data to their own destinations in modes $\B$ and $\C$, thus we set $\bp sd = 0$ and the pre-shared message rate $R_{V'}=0$. By Theorem~\ref{thm:IFconf} the rate pair $(R_W+R_U+R_V,
R_W+R_U+R_V)$ is achievable if
\begin{align*}
  2R_W+R_V+R_U &\leq n_D\\
  R_U+R_W &\leq \max\{n_I,
  n_D-n_I\}\\
  R_U&\leq n_D-n_I
\end{align*}
with $R_W\geq 0, R_U\geq 0, 0\leq R_V\leq \bp ss$.
When $n_I>n_D$ we set the private message rate $R_U = 0$ as the interference is strong. By Theorem~\ref{thm:IFconf} the rate
pair $(R_W+R_V+R_V', R_W+R_V+R_V')$ is achievable if
\begin{align*}
  2R_W+R_V+R_{V'}&\leq n_I\\
  R_W+R_{V'}&\leq n_D
\end{align*}
with $R_W\geq 0, 0\leq R_V\leq \bp ss, 0\leq R_V'\leq \bp sd$.
By the Fourier-Motzkin elimination, we arrive at
\begin{thm}
The following is an achievable sum rate $R_\text{sum}^\text{virtual}$ for
$\IFconf((n_D, n_I), \bp ss, \bp sd)$.
  \begin{enumerate}
    \item When $n_I< n_D, \bp sd=0$,
    \begin{align*}
      R_\text{sum}^\text{virtual} = 2\min \left\{
      \begin{array}{c}
      n_D,\\
      n_D-\frac{1}{2}n_I+\frac{1}{2}\bp ss,\\
      \max\{n_I, n_D-n_I\}+\bp ss
      \end{array}
      \right\},
  \end{align*}
  \item when $n_I>n_D$,
  \begin{align*}
    R_\text{sum}^\text{virtual} = 2\min \left\{
    \begin{array}{c}
    n_D+\bp ss,\\
    \frac{n_I+\bp ss+\bp sd}{2},\\
    n_I
    \end{array}
    \right\}.
  \end{align*}\\
  \end{enumerate}

\end{thm}

Now we can show the achievability of the sum capacity $\Ci {sum}$ using a
symmetric version of the scheme in Section~\ref{sec:ach.half-duplex}. Set $\delta_B = \delta_C = 1, \delta_A  = \delta$. For superposition
coding in modes $\B$ and $\C$, the sources set the data rates $R_{1B}=R_{2C} =
n_D$ and choose the shared rates \bp 12 = \bp 21 = \bp ss, \bp 14 = \bp 23
= \bp sd and relay rates \rly 123 = \rly 214 = $\Delta R$. The constraints
\eqref{eq:spcond1}-\eqref{eq:spcond2} translate to
\begin{align*}
  \delta \bp ss+\Delta R &\leq (n_C-n_D)^+,\\
  \delta \bp sd+\Delta R &\leq (n_I-n_D)^+,\\
  \delta \bp ss+\delta \bp sd+2\Delta R&\leq (\max\{n_I, n_C\}-n_D)^+.
\end{align*}
By Theorem~\ref{thm:half-duplex}, the sum rate achieved by this scheme is
\begin{align*}
  R_\text{sum} = \max_{\delta\geq 0}\frac{1}{2+\delta}(2n_D+2\Delta R+\delta
R_\text{sum}^\text{virtual}(n_D, n_I, \bp ss, \bp sd)).
\end{align*}

The optimization problem for $R_\text{sum}$ naturally divides in to the following parameter regions. For our choice of rates $\bp ss, \bp sd$ and $\Delta R$, tt is not hard to verify that the above constraints are satisfied and $R_\text{sum} = C_\text{sum}$ in all regions.
\begin{enumerate}
  \item $n_I<n_D<n_C$. $\bp ss = (n_C-n_D)/\delta$, $\bp sd = 0$ and
  $\Delta R = 0$. The interference link is weak in this region. We do not use it for sharing information or relay.
  \item $n_D<n_I\leq n_C$. $\bp sd=0$. The cooperation link dominates the interference link in this region, so we do not share data over the interference link. When the cooperation is strong enough, we use the additional capacity to relay data.
  \begin{enumerate}
    \item $n_C-n_D\leq \delta n_I$. $\bp ss =(n_C-n_D)/\delta$ and $\Delta R = 0$.
    \item $n_C-n_D>\delta n_I$. $\bp ss = n_I$ and
    \begin{align*}
      \Delta R = \min \left(\frac{n_C-n_D-\delta n_I}{2}, n_I-n_D\right)
    \end{align*}
  \end{enumerate}
  \item $n_D<n_C<n_I$. The interference link dominates in this region. We always use it for sharing data. When the cooperation link and the interference link are both strong enough, we further use them to relay data.
  \begin{enumerate}
    \item $n_I-n_D\leq \delta n_I$ or $n_C-n_D\leq \delta (n_I-n_D)$. $\bp ss = (n_C-n_D)/\delta,
    \bp sd = (n_I-n_C)/\delta$ and $\Delta R=0$.
    \item $n_I-n_D> \delta n_I$ and $n_C-n_D> \delta (n_I-n_D)$. $\bp ss = n_I-n_D, \bp ss+\bp
    sd=
    n_I$ and
    \begin{align*}
      \Delta R = \min \left(n_C-n_D-\delta (n_I-n_D), \frac{n_I-n_D-\delta
n_I}{2}\right)
    \end{align*}
  \end{enumerate}
\end{enumerate}

{\em Remark:} Primarily, cooperation enables better rates of
transmission over the interference channel. When both $n_C$ and $n_I$ are
large relative to $n_D$, relaying also comes into play. In the Gaussian model, we divide the problem into parameter regions as above. The basic idea for the coding scheme is to allocate the power
for the signals according to the intuition provided by the LDM, such that the rates for the messages in the Gaussian model and the corresponding LDM differ by at most a constant. Then it is sufficient to apply the achievable coding scheme for the LDM. Note that when $\SNR \approx \INR$, which corresponds to the case $n_D = n_I$, the achievable rate obtained by directly applying the LDM result is not tight with respect to the upper bound. In fact, we need to further consider the angle difference $\theta$ for the channel gains to show the constant gap result.

\section{The Symmetric Case: Gaussian Model}\label{sec:sym_Gaussian}
\label{sec:sumrateachievability}

We follow the intuition from the linear deterministic channel and consider a
symmetric version of the coding scheme in section~\ref{sec:achievability}
as well. The auxiliary random variables in Theorem~\ref{thm:IFconf} for the
induced symmetric virtual channel are chosen as follows: for source
$i=1,2,$ we define the auxiliary random variables $W_i,U_i,V_i'$ to be
independent, zero-mean Gaussian random variables with variances
$\sigma_W^2, \sigma_U^2, \sigma_{V'}^2$, respectively. Set
\begin{align*}
  X_{V_i'} &= V_i',\\
  X_{W_i} &= V_i' + W_i,\\
  X_{U_i} &= V_i' + W_i + U_i.
\end{align*}
The variance $\sigma_U^2$ for the private message is set below the
noise power level at the destination where it causes interference.  Following the
intuition from the linear deterministic case, we will employ {\em
zero-forcing beamforming} for the cooperative private messages. We choose
$V_1, V_2$ to be zero-mean Gaussian random variables with variance $\sigma_V^2$, independent of each other and all previously defined
auxiliary random variables, . When the channel matrix is
invertible, $X_{V_i}, i = 1,2$ are chosen such that
\begin{align*}
  \left[
    \begin{array}{c}
    V_1\\
    V_2
    \end{array}
    \right] =
    \left[
    \begin{array}{cc}
    h_{13} & h_{23}\\
    h_{14} & h_{24}
    \end{array}
    \right]
    \left[
    \begin{array}{c}
    X_{V_1}\\
    X_{V_2}
    \end{array}
    \right]
\end{align*}
where $X_{V_i}, i= 1, 2$ are correlated Gaussian random variables with variance
\begin{align*}
  \Var{X_{V_i}} =
\frac{\SNR+\INR}{\SNR^2+\INR^2-2\SNR\,\INR\cos\theta}\sigma_V^2.
\end{align*}
When the channel matrix is not invertible, we simply set $\sigma_V^2=0$ and $X_{V_1}=X_{V_2}=0$, i.e., there will be no cooperative private message.
The variance parameters must satisfy the power constraint
\begin{align*}
  \sigma_W^2+\sigma_U^2+\sigma_{V'}^2+\Var{X_{V_i}}\leq 1, \quad i = 1, 2.
\end{align*}
After removing the unintended pre-shared public signals, the destinations receive
\begin{align*}
  Y_3 &= h_{13}(W_1+U_1+V_1')+h_{23}W_2+V_1+h_{23}U_2+Z_3\\
  Y_4 &= h_{24}(W_2+U_2+V_2')+h_{24}W_1+V_2+h_{14}U_1+Z_4.
\end{align*}
We set the rates for the four types of messages to be symmetric, i.e.,
$R_{W_1}=R_{W_2}=R_W$ and so on. Also, in
Theorem~\ref{thm:half-duplex}, we set $\bp 12 = \bp 21 = \bp ss$, $\bp
14 = \bp 23 = \bp sd$, and $\rly 123 = \rly 214 = \Delta R$.

With the above definitions of
auxiliary random variables, there exist power and rate allocations such
that the rate $\overline{\Ci {sum}}$, defined in
Theorem~\ref{thm:sumrate}, is achievable within a constant. Specifically,
\begin{thm}\label{thm:sumrateachievability}
$\Ci {sum} \geq \overline{\Ci {sum}}-17$.
\end{thm}
\begin{proof}
  We sketch how we prove the theorem and refer the reader to Appendix~\ref{app:sumrateachievability} for details. We show achievability in the following five parameter regions. In the first four regions, we consider the coding schemes for the corresponding LDM and show that the sum capacity of the LDM can be achieved within a constant. The last region is unique for Gaussian channel, where the scheme according to the LDM can be strictly suboptimal.
  \begin{enumerate}
    \item $\CNR\leq \SNR$ or $\CNR\leq 1$ or $\INR\leq 1$. In this region, the condition implies that either the cooperation is not helpful or there is little interference. Therefore, the previous schemes for the interference channel are enough to achieve the upper bound within a constant.
    \item $2\INR<\SNR<\CNR$. This region corresponds to the case $n_I<n_D<n_C$.
    \item $2\SNR<\INR<\CNR$. This region corresponds to the case $n_D<n_I\leq n_C$. We further divide this region into two subregions as for the LDM.
    \item $\SNR<\CNR<\INR$. This region corresponds to the case $n_D<n_C\leq n_I$. We further divide this region into two subregions as for the LDM.
    \item $\SNR\approx \INR<\CNR$. This region corresponds to the case $n_D = n_I$. In LDM, if $n_D = n_I$, the channel is degenerated and the channel matrix $S$ has only rank $n_D$. However, in the Gaussian case, whether the channel is degenerated further depends on the angles of the channel gains. In particular, when $\cos \theta\approx 0$, the channel matrix $H$ is well conditioned and cooperation is still helpful.
  \end{enumerate}
\end{proof}

The following theorem provides an upperbound to the sum-rate. It is proved in Appendix~\ref{app:sumrateconverse}. This theorem together with the previous one imply Theorem~\ref{thm:sumrate}.
\begin{thm}\label{thm:converse}
Let
  \begin{align*}
    Cut(\delta) =& \frac{1}{2+\delta}\Big[\delta \log(1+\SNR P_{1A})+\delta \log(1+\SNR P_{2A})\\
        &\log(1+(\SNR+\CNR)P_{1B})+\log(1+(\SNR+\CNR)P_{2C})\Big]\\
    Z(\delta) = & \frac{1}{2+\delta}\Big[\delta \log(1+2\SNR P_{1A}+2\INR P_{2A})+\log(1+\SNR P_{1B})\\
        &+\log(1+(\SNR+\INR+\CNR)P_{2C})+\delta\log(1+\frac{\SNR P_{2A}}{1+\INR P_{2A}})\Big]\\
    V(\delta) =& \frac{1}{2+\delta}\Big[\delta\log\left(1+\INR P_{2A}+\frac{2\SNR P_{1A}+\INR P_{2A}}{1+\INR P_{1A}}\right)+\log(1+(\SNR+\INR+\CNR)P_{1B})\\
        &
        +\delta\log\left(1+\INR P_{1A}+\frac{2\SNR P_{2A}+\INR P_{1A}}{1+\INR P_{2A}}\right)+\log(1+(\SNR+\INR+\CNR)P_{2C})\Big]\\
    Cut'(\delta) =& \frac{1}{2+\delta}\Big[\delta\log(1+2(\SNR+\INR)(P_{1A}+P_{2A})+P_{1A}P_{2A}(\SNR^2+\INR^2-2\SNR\INR\cos\theta))\\
        &+\log(1+(\SNR+\INR)P_{1\B})+\log(1+(\SNR+\INR)P_{2\C})\Big]
  \end{align*}
Define
  $\displaystyle\overline{C_{sum}^{HD}} = \max_{\delta, P_{1A},P_{1B}}\min(Cut(\delta), Z(\delta),
  V(\delta), Cut'(\delta))$, where the maximization is over all
non-negative $\delta, P_{1A}, P_{1B}, P_{2A}, P_{2C}$ which satisfy the power
constraints
  \begin{align*}
    \frac{\delta P_{1A}+P_{1B}}{2+\delta}\leq 1\text{  and  }\frac{\delta P_{2A}+P_{2C}}{2+\delta}\leq
    1.
  \end{align*}
Then
\[ \Ci {sum} \leq \overline{\Cii {sum} {HD}}\leq\overline{\Ci {sum}}+7.\]
\end{thm}

\section{The Cognitive Case: LDM}\label{sec:cog_LDM}

In this section, we study the linear deterministic model (LDM) of the cognitive channel. We first characterize the cognitive capacity of the LDM, which is the counterpart of the $R_0$-capacity for the Gaussian case. Next we describe the coding scheme for the channel and provide a simple interpretation of the coding scheme. We then briefly discuss the converse. The intuition from the LDM will be our guideline for studying the Gaussian channel in the next section.

\subsection{Channel Model and Cognitive Capacity}\label{sec:cog_LDM_model}
The LDM of the cognitive channel is parameterized by the nonnegative
integers
\begin{align*}
  &n_1 = \lfloor\log \SNR_1\rfloor^+, n_2 = \lfloor\log \SNR_2\rfloor^+ , \alpha_1 = \lfloor\log \INR_1\rfloor^+,\\
  &\alpha_1 = \lfloor\log \INR_2\rfloor^+ , \beta = \lfloor\log \CNR\rfloor^+
\end{align*}
The channel is depicted in Figure~\ref{Fig:HD_channel_LDM_cog}. Let $S_n$ be the shift matrix in $\mathbb{F}_2^{n\times n}$, as defined in Section~\ref{sec:sym_LDM}. As the cooperation is
only unidirectional, the sources can work in mode $A$ and $B$. In mode $\A$, both
sources transmit and the channel inputs $X_{1t}, X_{2t}$ are in $\mathbb{F}_2^{\max\{n_1, \alpha_1, n_2, \alpha_2\}}$. The nodes receive:
\begin{align*}
Y_{1t} &= 0,\\
Y_{2t} &= 0,\\
Y_{3t} &= S_{\max\{n_1, \alpha_1, n_2, \alpha_2\}}^{\max\{n_1, \alpha_1, n_2, \alpha_2\}-n_1}X_{1t}\oplus S_{\max\{n_1, \alpha_1, n_2, \alpha_2\}}^{\max\{n_1, \alpha_1, n_2, \alpha_2\}-\alpha_1}X_{2t},\\
Y_{4t} &= S_{\max\{n_1, \alpha_1, n_2, \alpha_2\}}^{\max\{n_1, \alpha_1, n_2, \alpha_2\}-n_2}X_{2t}\oplus S_{\max\{n_1, \alpha_1, n_2, \alpha_2\}}^{\max\{n_1, \alpha_1, n_2, \alpha_2\}-\alpha_2}X_{1t}.
\end{align*}

In mode $\B$, source 2 listens and the channel inputs $X_{1t}, X_{2t}$ are in $\mathbb{F}_2^{\max\{n_1, \alpha_1, n_2, \alpha_2, \beta\}}$. Then,
\begin{align*}
Y_{1t} &= 0,\\
Y_{2t} &= S_{\max\{n_1, \alpha_1, n_2, \alpha_2, \beta\}}^{\max\{n_1, \alpha_1, n_2, \alpha_2, \beta\}-\beta}X_{1t},\\
Y_{3t} &= S_{\max\{n_1, \alpha_1, n_2, \alpha_2, \beta\}}^{\max\{n_1, \alpha_1, n_2, \alpha_2, \beta\}-n_1}X_{1t},\\
Y_{4t} &= S_{\max\{n_1, \alpha_1, n_2, \alpha_2, \beta\}}^{\max\{n_1, \alpha_1, n_2, \alpha_2, \beta\}-\alpha_2}X_{1t}.
\end{align*}

\begin{figure}
  \centering\scalebox{0.6}{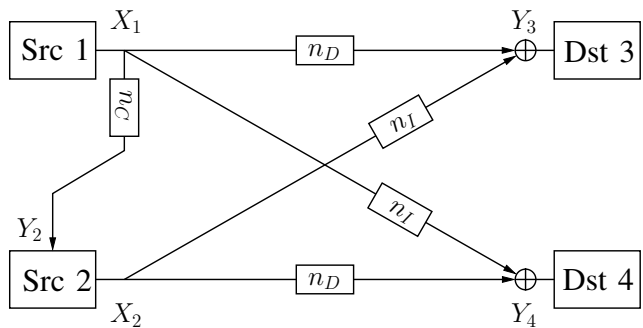}
  \caption{Linear deterministic interference channel with unidirectional half-duplex source cooperation.}\label{Fig:HD_channel_LDM_cog}
\end{figure}

For this channel, source 1 is the primary user and source 2 is the secondary user. As mentioned in Section~\ref{sec:problem}, we would like to know the best rate the secondary can get when the primary is communicating at its link capacity, which is $R_1 = n_1$. We define the cognitive capacity for this LDM as follows, which is similar to the $R_0$-capacity for the Gaussian case.

\begin{definition}
  Assume the capacity region of the channel in Figure~\ref{Fig:HD_channel_LDM_cog} is $\mathscr{C}$. The cognitive
  capacity of the channel is defined as
  \begin{align*}
    \Ci{cog} = \max_{\begin{subarray}{c} (R_1, R_2)\in \mathscr{C}\\R_1 = n_1
    \end{subarray}}R_2.
  \end{align*}
\end{definition}

Note that in this definition, the primary does not need to back-off as in the
$R_0$-capacity. This back-off is not necessary because the linear deterministic model is a coarser
description of the true channel. It characterizes the channel capacity only up
to degree of freedom. Therefore, a constant back-off in the Gaussian model is negligible in this LDM.

\begin{thm}\label{thm:LDcog}
  The cognitive capacity $\Ci{cog}$ of channel in Figure~\ref{Fig:HD_channel_LDM_cog} is given by
  \begin{align*}
    \Ci{cog} =\max_{\delta\geq 0}\min(u_1, u_2, u_3, u_4),
  \end{align*}

  where
  \begin{align*}
    u_1 &= \frac{1}{1+\delta}n_2\\
    u_2 &= \frac{1}{1+\delta}[n_2\vee\alpha_2-\alpha_2\wedge n_1+\delta(\beta\vee\alpha_2\vee n_1-n_1)]\\
    u_3 &= \frac{1}{1+\delta}[(\alpha_1-n_1)^++(n_2-\alpha_1)^+]\\
    u_4 &= \frac{1}{1+\delta}[(\alpha_1-n_1)^+-\alpha_2\wedge n_1+(n_2-\alpha_1)\vee\alpha_2+\delta(\beta\vee\alpha_2\vee
    n_1-n_1)].
  \end{align*}
\end{thm}
The parameter ¦Ä is a scheduling parameter the same as the scheduling parameter used in Theorem~\ref{thm:cog}. Before continuing to the coding scheme and the converse proof, we summarize here the result for cognitive capacity of the interference channel without cooperation for comparison.
\begin{prop}\label{prop:LDcogIFC}
  The cognitive capacity of linear deterministic interference channel parameterized by
  $n_1, n_2, \alpha_1, \alpha_2$ is
  \begin{align*}
    \Cii{cog}{IFC} = \min(v_1, v_2, v_3, v_4),
  \end{align*}
  \indent where
  \begin{align*}
    v_1 &= n_2\\
    v_2 &= n_2\vee\alpha_2-\alpha_2\wedge n_1\\
    v_3 &= (\alpha_1-n_1)^++(n_2-\alpha_1)^+\\
    v_4 & = (\alpha_1-n_1)^+-\alpha_2\wedge n_1+(n_2-\alpha_1)\vee\alpha_2.
  \end{align*}
\end{prop}

\begin{proof}
  The capacity region of the linear deterministic interference channel \cite{BT08} is given by the set of
  $(R_1, R_2)$ satisfying
  \begin{align*}
    R_1\leq & n_1\\
    R_2\leq & n_2\\
    R_1+R_2\leq & (n_1-\alpha_2)^++n_2\vee \alpha_2\\
    R_1+R_2\leq & (n_2-\alpha_1)^++n_1\vee \alpha_1\\
    R_1+R_2\leq & \alpha_1\vee (n_1-\alpha_2)+\alpha_2\vee (n_2-\alpha_1)\\
    2R_1+R_2\leq & n_1\vee \alpha_1+(n_1-\alpha_2)^++\alpha_2\vee (n_2-\alpha_1)\\
    R_1+2R_2\leq & n_2\vee \alpha_2+(n_2-\alpha_1)^++\alpha_1\vee(n_1-\alpha_2).
  \end{align*}
  Evaluating the inequalities at $R_1 = n_1$, the maximum $R_2$ gives the cognitive capacity above.
\end{proof}

Using the notation in the proposition, we can rewrite the cognitive capacity of the
cognitive channel as
\begin{align*}
  \Ci{cog} = \max_{\delta}\frac{1}{1+\delta}\min(v_1, v_2+\delta(\beta\vee\alpha_2\vee n_1-n_1), v_3, v_4+\delta(\beta\vee\alpha_2\vee
  n_1-n_1)).
\end{align*}
When $\beta = 0$, clearly the cognitive channel reduces to the original interference
channel and $\Ci{cog}(\beta = 0) = \Ci{cog}^{IFC}$. When $\beta\leq \alpha_2\vee n_1$, we can see
that $\Ci{cog}(\beta) = \Ci{cog}(\beta = 0)= \Ci{cog}^{IFC}$.  Moreover,
when the channel is aligned, i.e., $n_1+n_2 =\alpha_1+\alpha_2$, we have
\begin{align*}
  \Ci{cog}\leq \max_{\delta}u_3 = v_3 = \max(n_1, n_2, \alpha_1, \alpha_2)-n_1 =
  C_{cog}^{IFC}.
\end{align*}
In both cases, the cooperation link is useless and the optimal interference channel scheme is enough.
Therefore, in the following discussions, we assume $\beta > \alpha_2\vee n_1$, $n_1+n_2\ne
\alpha_1+\alpha_2$, and
\begin{align*}
  \Ci{cog} = \max_{\delta}\frac{1}{1+\delta}\min(v_1, v_2+\delta(\beta-n_1), v_3,
  v_4+\delta(\beta-n_1)).
\end{align*}

\subsection{Coding Scheme}

We consider general asymmetric schemes for the cognitive LDM. Compared with the
symmetric case, we have several differences: (a) the interference channel
is asymmetric and is determined by $(n_1, \alpha_1, n_2, \alpha_2)$; (b)
for the virtual channel, as $n_{21} = 0$, we always have $\bp 21 = 0$.

In our coding scheme, we do not use the pre-shared message and set
$\bp 14 = \bp 23 = 0$. Hence the virtual channel is denoted as
$\IFconf(n_1, \alpha_1, n_2, \alpha_2, \bp 12)$. Moreover, relay is also
not used in this case and we set the relay rates $\rly 123 = \rly 214 = 0$.
By definition of the cognitive capacity, we have $R_1 = n_1$ and our scheme
sets $R_{1B} = R_{1A} = n_1$.

To choose the auxiliary random variables in Theorem~\ref{thm:IFconf} for
this asymmetric virtual channel, let $n = n_1\vee \alpha_1\vee n_2\vee
\alpha_2$ for simplicity.  For source $i\in \{1, 2\}$, we define the public and  private
auxiliary random variables $W_i,U_i$ to be independent random variables on $\mathbb{F}_2^n$. The public auxiliary
random variables are uniformly distributed over $\mathbb{F}_2^n$. The private auxiliary random variables are uniformly distributed over the set of length $n$ vectors in $\mathbb{F}_2^n$ whose
upper $n-(n_i-\alpha_i)^+$ elements are fixed to be 0.
In Theorem~\ref{thm:IFconf}, we set $V_i'=0$ and
\begin{align*}
  X_{W_i} &= W_i,\\
  X_{U_i} &= W_i + U_i.
\end{align*}
Note that $U_i$ occupies the lower $(n_i-\alpha_i)^+$ levels so that it does not
appear at the other destination.This is similar to the choice made
in~\cite{EtkinTseWang} for the (non-cooperative) intereference channel.

For the cooperative private codebook, we set the auxiliary random variable
$V_2 = 0$ and choose $V_1$ independent of the auxiliary random variables
and distributed uniformly over the set of length $n$ vectors in $\mathbb{F}_2^n$ whose
upper $n-k$ elements are fixed to be 0.
The choice of $k$ will be specified later. We choose $(X_{V_1},X_{V_2})$
as deterministic functions of $V_1$ such that
\begin{align}
  \left[
  \begin{array}{c}
    V_1\\
    0
  \end{array}
  \right] = \left[
  \begin{array}{cc}
    S_n^{n-n_1}&S_n^{n-\alpha_1}\\
    S_n^{n-\alpha_2}&S_n^{n-n_2}
  \end{array}
  \right]\left[
  \begin{array}{c}
    X_{V_1}\\
    X_{V_2}
  \end{array}
  \right] \label{eq:invertibility}
\end{align}
For the particular choice of $X_{V_i}$, the sources are effectively doing zero-forcing beamforming such that the primary destination receives $V_1$ and the signal cancels at the secondary destination. For this scheme to be feasible, $k$ is chosen such that for arbitrary $V_1$
in $\mathbb{F}_2^n$ with the upper $n-k$ elements being 0, there exist
$X_{V_1}, X_{V_2}$ satisfying the above equation. Such $k$ is called {\em
realizable}, and we have the following lemma.
\begin{lem}\label{lem:cog_LDM_k}
  For channel with parameters $(n_1, n_2, \alpha_1, \alpha_2)$,
  the largest realizable $k$ is $[n_1-(\alpha_2-n_2)^+]\vee[\alpha_1-(n_2-\alpha_2)^+]$
\end{lem}
\begin{proof} Clearly we have $k\leq n_1\vee \alpha_1$. Assume
$\alpha_2\geq n_2$. As $V_2 = 0$ and the upper $\alpha_2-n_2$ bits of $V_2$
and $X_{V_1}$ are the same, those bits of $X_{V_1}$ must be zero. After
removing the corresponding first $\alpha_2-n_2$ columns, the channel matrix
is equivalent to a channel with parameters $(n_1-(\alpha_2-n_2), n_2,
\alpha_1, n_2)$. Hence we have $k\leq (n_1-(\alpha_2-n_2))\vee \alpha_1$.
Ignoring the all zero rows of this new channel matrix, it is not hard to
see that it is of full row rank and for any $V_1\in\mathbb{F}_2^k$ with its
upper $n-k$ elements being 0, where $k=(n_1-(\alpha_2-n_2))\vee \alpha_1$,
there exists $X_{V_1},X_{V_2}$ satisfying \eqref{eq:invertibility}. Hence
the maximum realizable $k$ is $(n_1-(\alpha_2-n_2))\vee  \alpha_1$. A
similar argument can be made for $\alpha_2<n_2$ and combining the two we
have the lemma.  \end{proof}

According to the above lemma, we set
$k=[n_1-(\alpha_2-n_2)^+]\vee[\alpha_1-(n_2-\alpha_2)^+]$. Source 1 sends
$X_{U_1}^L+X_{V_1}^L$ and source 2 sends $X_{U_2}^L$. The induced channel
$p_{Y_3,Y_4|W_1,W_2,U_1,U_2,V_1}$ is
\begin{align*}
  Y_3 &= S_n^{n-n_1}(W_1+U_1)+S_n^{n-\alpha_1}W_2+V_1\\
  Y_4 &= S_n^{n-n_2}(W_2+U_2)+S_n^{n-\alpha_2}W_1.
\end{align*}
By Theorem~\ref{thm:IFconf} the rate pair $(R_{W_1}+R_{U_1}+R_{V_1},
R_{W_2}+R_{U_2})$ is achievable if the rates
$R_{W_1},R_{U_1},R_{V_1},R_{W_2},R_{U_2}$ are non-negative and they satisfy
the following conditions:
\begin{align*}
  R_{W_1}+R_{U_1}+R_{W_2}+R_{V_1}&\leq \max(\alpha_1, n_1)\\
  R_{U_1}+R_{W_2}+R_{V_1}&\leq \max(\alpha_1, k)\\
  R_{W_1}+R_{U_1}+R_{V_1}&\leq \max(n_1, k)\\
  R_{W_1}+R_{U_1}&\leq n_1\\
  R_{U_1}+R_{W_2}&\leq \max(n_1-\alpha_2, \alpha_1)\\
  R_{U_1}+R_{V_1}&\leq k\\
  R_{U_1}&\leq (n_1-\alpha_2)^+\\
  R_{V_1}&\leq \bp 12\\
  R_{W_1}+R_{W_2}+R_{U_2}&\leq \max(\alpha_2, n_2)\\
  R_{W_1}+R_{U_2}&\leq \max(n_2-\alpha_1, \alpha_2)\\
  R_{W_2}+R_{U_2}&\leq n_2\\
  R_{U_2}&\leq (n_2-\alpha_1)^+
\end{align*}
Set $R_1 = R_{W_1}+R_{U_1}+R_{V_1} = n_1$ and $R_2 =
R_{W_2}+R_{U_2}$. Applying Fourier-Motzkin elimination to the
above inequalities we get the following theorem.
\begin{thm}
  The following is an achievable cognitive rate for
$\IFconf(n_1, \alpha_1, n_2, \alpha_2, \bp
  12)$,
  \begin{align*}
    \Rii{cog}{virtual} = \min(v_1, v_2+\bp 12, v_3, v_4+\bp 12)
  \end{align*}
  in which $v_i, i = 1, 2, 3, 4$ are defined in
Proposition~\ref{prop:LDcogIFC}.
\end{thm}

With this theorem in hand, showing the achievability of the cognitive capacity for the original
half-duplex channel $\Ci{cog}$ is quite straightforward. Set $\delta_B = \delta,
\delta_C = 0, \delta_A = 1$. For the superposition coding in mode $\B$, source 1 sets rate $R_{1\B} = n_1$ and the shared rate $\frac{\bp 12}{\delta} = \beta - n_1$ or
$\bp 12 = \delta (\beta-n_1)$. As $R_{1\B}= R_{1A} = n_1$, the total rate for the
primary is $R_1 = n_1$. Then by Theorem~\ref{thm:half-duplex}, the cognitive rate achieved by the secondary is
\begin{align*}
  \Ri{cog} = \max_{\delta\geq 0}\frac{1}{1+\delta}\Rii{cog}{virtual}
   =\max_{\delta\geq 0} \min(u_1,u_2,u_3,u_4),
\end{align*}
where $u_1,u_2,u_3,u_4$ were defined in Theorem~\ref{thm:LDcog}.

\subsection{An Interpretation of the Scheme}

For the interesting region $\beta > \alpha_2\vee n_1$ and $n_1+n_2\ne
\alpha_1+\alpha_2$, we can obtain a simple interpretation of the scheme by
optimizing over $\delta$. Let
\begin{align*}
  \Ci{cog}(\delta) = &\frac{1}{1+\delta}\min(v_1, v_2+\delta(\beta-n_1), v_3,
  v_4+\delta(\beta-n_1))\\
  = & \frac{1}{1+\delta}\min(v_1\wedge v_3, v_2\wedge v_4+\delta(\beta-n_1)).
\end{align*}
Define $\delta_0 =\frac{v_1\wedge v_3-v_1\wedge v_2\wedge v_3\wedge v_4}{\beta-n_1}\geq
  0$. When $\delta \geq \delta_0$,
\begin{align*}
  \Ci{cog}(\delta) = \frac{1}{1+\delta}[v_1\wedge v_3]\leq \frac{1}{1+\delta_0}[v_1\wedge
  v_3].
\end{align*}
When $0\leq \delta< \delta_0$, we must have $\delta_0>0$, which means $v_1\wedge
  v_3>v_1\wedge v_2 \wedge v_3\wedge v_4$; hence, $v_2\wedge v_4 = v_1\wedge v_2 \wedge v_3\wedge
  v_4$.
\begin{align*}
  \Ci{cog}(\delta) &= \frac{1}{1+\delta}[v_2\wedge v_4+\delta(\beta-1)]\\
  &\leq \max\left(v_2\wedge v_4, \frac{1}{1+\delta_0}[v_2\wedge v_4+\delta_0(\beta-1)]\right)\\
  & = \max\left(v_1\wedge v_2\wedge v_3\wedge v_4, \frac{1}{1+\delta_0}[v_1\wedge
  v_3]\right).
\end{align*}
The second inequality is due to the fact that $\Ci{cog}(\delta)$ is a monotone function in this region and its
maximum is achieved at the end points. The last equality follows from the fact that $v_2\wedge v_4 = v_1\wedge v_2 \wedge v_3\wedge v_4$.

In summary,
\begin{align*}
  &\Ci{cog}(\delta) \leq \max\left(v_1 \wedge v_2\wedge v_3 \wedge v_4, \frac{1}{1+\delta_0}[v_1\wedge
  v_3]\right)\\
  &\Ci{cog} = \max_{\delta}\Ci{cog}(\delta) = \max\left(v_1 \wedge v_2\wedge v_3 \wedge v_4, \frac{1}{1+\delta_0}[v_1\wedge
  v_3]\right).
\end{align*}
The equality is achieved by taking either $\delta = 0$ or $\delta =  \delta_0$. As
defined in Section~\ref{sec:cog_LDM_model}, $\Cii{cog}{IFC} = v_1 \wedge v_2\wedge v_3 \wedge v_4$. If we
let $\alpha_2 = 0$, the interference channel reduces to the corresponding Z-channel and we
can define its cognitive capacity as
\begin{align*}
  \Cii{cog}{Z} = \Cii{cog}{IFC}(\alpha_2 = 0) = v_1\wedge v_2 \wedge v_3\wedge v_4|_{\alpha_2 =
  0}=  v_1\wedge v_3,
\end{align*}
Then $\Ci{cog}$ can be rewritten as
\begin{align*}
  \Ci{cog} = \max\left(\Cii{cog}{IFC}, \frac{1}{1+\delta_0}\Cii{cog}{Z}\right).
\end{align*}

This expression of $\Ci{cog}$ provides a new interpretation of our scheme. It consists of two optional
schemes. One is the optimal scheme for the interference channel that achieves its
cognitive capacity. In the second scheme, the secondary first listens in mode $B$
long enough to collect information of the interference from source 1 during mode
$A$. In each time instant, it gets $\beta-n_1$ bits. Then in mode $A$, it uses this
information to perform dirty paper coding to fully ``cancel'' the interference. Thus the original channel is
now equivalent to a Z-channel and $\Cii{cog}{Z}$ is achieved for the secondary. The amount
of information needed to cancel interference is $\Cii{cog}{Z}-\Cii{cog}{IFC}$; hence, the time
to listen is $\delta_0 = \frac{\Cii{cog}{Z}-\Cii{cog}{IFC}}{\beta-n_1}$, as defined above. It is easy to see that this scheme achieves
rate $\frac{1}{1+\delta_0}\Cii{cog}{Z}$. Our optimal scheme picks the better of the two and
achieves capacity $\Ci{cog}$.

\subsection{Converse}
To prove the converse, we need the following theorem.
\begin{thm}\label{thm:cogLDconverse}
  The capacity region $\mathscr{C}$ is contained within $\bigcup_{\delta}\mathscr{C}(\delta)$, where $\mathscr{C}(\delta)$
  is the set of rate pairs $(R_1, R_2)$ satisfying
  \begin{align*}
    R_2 \leq & \frac{1}{1+\delta}n_2\\
    R_1+R_2\leq & \frac{1}{1+\delta}[\max(n_2, \alpha_2)+\delta\max(\beta,
\alpha_2,n_1)+(n_1-\alpha_2)^+]\\
    R_1+R_2\leq & \frac{1}{1+\delta}[\max(\alpha_1, n_1)+\delta n_1+(n_2-\alpha_1)^+]\\
    2R_1+R_2\leq & \frac{1}{1+\delta}[\max(\alpha_1,
n_1)+\delta+(n_1-\alpha_2)^++\max(n_2-\alpha_1,
\alpha_2)+\delta\max(\beta, \alpha_2, n_1)]
  \end{align*}
\end{thm}

For schemes with scheduling parameter $\delta$, $\mathscr{C}(\delta)$ can
be shown as an outer bound on the achievable rate region.  The first upper
bound is proved by assuming no interference. The second and third upper
bounds are proved along the lines of the Z-channel bound in
\cite{EtkinTseWang} and the last bound has similarities to the $2R_1+R_2$
upper bound in the same reference. The full details are provided for the
Gaussian model.

By evaluating the upper bounds with $R_1 = n_1$ and optimizing over
$\delta$ we get an upper bound on $R_2$, which matches the cognitive capacity
given in Theorem~\ref{thm:LDcog}.

\section{The Cognitive Case: Gaussian Model}\label{sec:cog_Gaussian}

We follow the intuition in the previous section to approximately
characterize the $R_0$-capacity of the Gaussian cognitive channel.
The auxiliary random variables for the virtual channel in
Theorem~\ref{thm:IFconf} are chosen as follows: For source $i=1,2$,
we define respectively the public and the private auxiliary random variables
$W_i$ and $U_i$ to be independent, zero-mean Gaussian random variables with
variances $\sigma_{W_i}^2,\sigma_{U_i}^2$, respectively.
In Theorem~\ref{thm:IFconf}, we define
\begin{align*}
  X_{W_i} &=  W_i,\\
  X_{U_i} &=  W_i + U_i.
\end{align*}
The variance $\sigma_{U_i}^2$ for the private message is set below the noise power level
at the destination where it causes interference. Following
the intuition from the linear deterministic case, we will employ {\em
zero-forcing beamforming} for the cooperative private messages. We choose
$V_2=0$ and $V_1$ to be zero-mean
Gaussian random variables with variance $\sigma_{V_1}^2$, independent of each other and all previously
defined auxiliary random variables. When the channel
matrix is invertible, $X_{V_1}$ and $X_{V_2}$ are chosen such that
\begin{align*}
  \left[
    \begin{array}{c}
    V_1\\
    0
    \end{array}
    \right] =
    \left[
    \begin{array}{cc}
    h_{13} & h_{23}\\
    h_{14} & h_{24}
    \end{array}
    \right]
    \left[
    \begin{array}{c}
    X_{V_1}\\
    X_{V_2}
    \end{array}
    \right]
\end{align*}
In this case, $X_{V_i}, i=1,2$ are correlated Gaussian random varaibles
with variances
\begin{align}
  \Var{X_{V_1}} &=
\frac{|h_{24}|^2}{|h_{13}h_{24}-h_{14}h_{23}|^2}\sigma_{V_1}^2\notag\\
 &=
\frac{\SNR_2}{\SNR_1\SNR_2+\INR_1\INR_2-2\sqrt{\SNR_1\SNR_2\INR_1\INR_2}\cos\theta}\sigma_{V_1}^2 \label{eq:cogcooppvtpower1}\\
  \Var{X_{V_2}} &=
\frac{|h_{14}|^2}{|h_{13}h_{24}-h_{14}h_{23}|^2}\sigma_{V_1}^2\notag\\
&=\frac{\INR_2}{\SNR_1\SNR_2+\INR_1\INR_2-2\sqrt{\SNR_1\SNR_2\INR_1\INR_2}\cos\theta}\sigma_{V_1}^2 \label{eq:cogcooppvtpower2}
\end{align}
When the channel matrix is singular, we set\footnote{In fact, in a
region where the channel matrix is ill-conditioned, we do not employ
cooperative private message.} $\sigma_{V_1}^2=0$, i.e., there is no
cooperative private message.
The variance parameters must satisfy the power constraint
\begin{align*}
  \Var{X_{U_i}}+\Var{X_{V_i}}\leq 1, \quad i = 1, 2
\end{align*}
The destinations receive
\begin{align*}
  Y_3 = & h_{13}(W_1+U_1)+h_{23}W_2+V_1+h_{23}U_2+Z_3\\
  Y_4 = & h_{24}(W_2+U_2)+h_{24}W_1+h_{14}U_1+Z_4
\end{align*}
In Theorem~\ref{thm:half-duplex}, as mentioned earlier, we set $\bp 21 =
\bp 14 = \bp 23 = \rly 123 = \rly 214 = 0$, i.e., only $\bp 12$ is
non-zero, in general.

In appendix~\ref{app:cogachievability} we show that with the above choice
of auxiliary random variables, there are power and rate allocations under
which we achieve an $R_1$ which is within $R_0$ of the point-to-point
capacity $C_0=\log(1+\SNR_1)$ of the primary link and an $R_2$ which is
within a constant of $\overline{C_{R_0}}$ as defined in
Theorem~\ref{thm:cog}. Specifically, we prove that
\begin{thm}\label{thm:cogachievability}
If $R_0>7$,
\begin{align*}
\CRo \geq \overline{\CRo}-23-2R_0.
\end{align*}
\end{thm}
To prove the converse part of Theorem~\ref{thm:cog}, we need the following theorem that
is similar to Theorem~\ref{thm:cogLDconverse}. It is proved in appendix~\ref{app:cogouterboundtheorem}.
\begin{thm}\label{thm:cog_outerbound}
  The capacity region $\mathscr{C}$ is contained within $\bigcup_{\delta}\mathscr{C}(\delta)$, where $\mathscr{C}(\delta)$
  is the set of rate pairs $(R_1, R_2)$ satisfying
  \begin{align*}
    R_2\leq & 1+\frac{1}{1+\delta}\log(1+\SNR_2P_{2A})\\
    R_1+R_2\leq &1+\frac{1}{1+\delta}\Big[\log
    (1+2\SNR_2P_{2A}+2\INR_2P_{1A})+\delta\log(1+(\SNR_1+\INR_2+\CNR)P_{1B})\\
    &+\log(1+\frac{\SNR_1P_{1A}}{1+\INR_2P_{1A}})\Big]\\
    R_1+R_2\leq &2+\frac{1}{1+\delta}\left[\log(1+2\SNR_1P_{1A}+2\INR_1P_{2A})+\delta \log(1+\SNR_1P_{1B})+\log (1+\frac{\SNR_2P_{2A}}{1+\INR_1P_{2A}})\right]\\
    2R_1+R_2\leq &3+\frac{1}{1+\delta}\Big[\log(1+2\SNR_1P_{1A}+2\INR_1P_{2A})+\delta
    \log(1+\SNR_1P_{1B})+\log(1+\frac{\SNR_1P_{1A}}{1+\INR_2P_{1A}})\\
    &+\log(1+\INR_2P_{1A}+\frac{2\SNR_2P_{2A}+\INR_2P_{1A}}{1+\INR_1P_{2A}})+\delta \log(1+(\SNR_1+\INR_2+\CNR)P_{1B})\Big]
  \end{align*}
  with power constraint
  \begin{align*}
    \frac{P_{1A}+\delta P_{1B}}{1+\delta}\leq 1, \quad
    \frac{P_{2A}}{1+\delta}\leq 1, \quad P_{2B} = 0.
  \end{align*}
\end{thm}

Setting the power terms to their maximum possible value, i.e., $P_{iA} =
1+\delta, P_{1B} =\frac{1+\delta}{\delta}, i = 1, 2$, we get a new
outer bound on the capacity region that is easier to use. The following lemma is shown in
appendix~\ref{app:cogouterboundlemma}.
\begin{lem}\label{lem:cog_outerbound}
  The capacity region $\mathscr{C}$ is contained within $\bigcup_{\delta}\mathscr{C}(\delta)$, where $\mathscr{C}(\delta)$
  is the set of rate pairs $(R_1, R_2)$ satisfying
  \begin{align*}
    R_2\leq & \frac{1}{1+\delta}\log(1+\SNR_2)+2\\
    R_1+R_2\leq &\frac{1}{1+\delta}\left[\log (1+2\SNR_2+2\INR_2)+\delta\log(1+(\SNR_1+\INR_2+\CNR))+\log(1+\frac{\SNR_1}{1+\INR_2})\right]+3\\
    R_1+R_2\leq &\frac{1}{1+\delta}\left[\log(1+2\SNR_1+2\INR_1)+\delta \log(1+\SNR_1)+\log (1+\frac{\SNR_2}{1+\INR_1})\right]+4\\
    2R_1+R_2\leq &\frac{1}{1+\delta}\Big[\log(1+2\SNR_1+2\INR_1)+\delta
    \log(1+\SNR_1)+\log(1+\frac{\SNR_1}{1+\INR_2})\\
    &+\max(\log(1+\INR_2+\frac{2\SNR_2+\INR_2}{1+\INR_1}), \log(1+2\INR_2))+\delta \log(1+(\SNR_1+\INR_2+\CNR))\Big]+6
  \end{align*}
\end{lem}
Setting $R_1 = \log(1+\SNR_1)-R_0$ in this lemma we get $C_{R_0}\leq
\overline{C_{R_0}}$.






\appendices
\section{Proof of Theorem~\ref{thm:sumrateachievability}}
\label{app:sumrateachievability}

We prove this sum-rate achievability result in two steps. Instead of
directly comparing $\overline{\Ci {sum}}$ with the rate achievable by the
coding scheme in section~\ref{sec:achievability}, we will first show that
the $\overline{\Ci {sum}}$ is within a constant of $\overline{\Cii {sum}
{LDM}}$, a quantity we define below inspired by the result for the linear
deterministic model. We will then prove that the coding scheme in
section~\ref{sec:achievability} can be used to achieve a sum-rate which is
within a constant of $\overline{\Cii {sum}{LDM}}$.  Specifically, we prove
the following two lemmas which together imply
Theorem~\ref{thm:sumrateachievability}. To simplify the notation, let $\x = \SNR, \y = \INR, \z = \CNR$, and define
$n_D = \lfloor \log \x\rfloor^+, n_I = \lfloor \log \y\rfloor^+, n_C = \lfloor \log \z\rfloor^+$.
\begin{lem}\label{lem:GaussianLDMlink}
  Define
\begin{align*}
  \overline{\Cii {sum} {LDM}} = \max_{\delta}\overline{\Cii {sum} {LDM}}(\delta) =
\max_{\delta}\min(u_1'-6, u_2'-4, u_3', u_4'-4, u_4-10)
\end{align*}
\indent where
\begin{align*}
  u_1' &= \frac{2}{2+\delta}\left(\delta n_D+\max\{n_D,
  n_C\}\right)\\
  u_2' &= \frac{1}{2+\delta}\left(\delta \max\{2n_D-n_I,
n_I\}+n_D+\max\{n_D, n_I, n_C\}\right)\\
  u_3' &= \frac{2}{2+\delta}\left(\delta\max\{n_I, n_D-n_I\}+\max\{n_D, n_I, 
  n_C\}\right)\\
  u_4' &= \frac{2(1+\delta)}{2+\delta}\max\{n_D, n_I\}
\end{align*}
and $u_4$ is as defined in Theorem~\ref{thm:sumrate}.
Then $\overline{\Ci {sum}}\leq\overline{\Cii {sum} {LDM}}+10$.
\end{lem}
\begin{proof}
  The following inequality is useful for the proof.
  \begin{align*}
    \lfloor\log \x\rfloor^+\leq (\log \x)^+\leq \log(1+\x)\leq &1+(\log \x)^+\leq 2+\lfloor\log \x\rfloor^+, \quad \forall x>0.
  \end{align*}

  It is easy to verify that $u_1\leq u_1'+4, u_2\leq u_2'+6$ and $u_3\leq u_3'+10$. So we get the result.
\end{proof}

Note that in the definition of $\overline{\Cii {sum} {LDM}}$ we have
preserved the term $u_4$ rather than have all the terms as functions of
$n_D, n_I$ and $n_C$. The reason for this is that the linear deterministic
model is too coarse to model the channel phase information. When the
channel matrix becomes ill-conditioned, the term $u_4$ may dominate
$\overline{\Ci {sum}}$ and also have a large gap with respect to $u'_4$.

Next we show that $\overline{\Cii {sum} {LDM}}$ can be achieved within a constant.
\begin{lem}\label{lem:GaussianLDMachievability}
  $\Ci {sum} \geq \overline{\Cii {sum} {LDM}}-7$.
\end{lem}
\begin{proof}
  To simplify the notation, let
  \begin{align*}
    \beta_1 &= \frac{\x^2+\y^2-2\x\y\cos\theta}{\x(\x+\y)}\\
    \beta_2 &= \frac{\x^2+\y^2-2\x\y\cos\theta}{\y(\x+\y)}.
  \end{align*}
  Then, for the auxiliary random variables in
  section~\ref{sec:sumrateachievability}, we have $\sigma_V^2 = \beta_1 \x
  \Var{X_V}$. We note that $\beta_1\x = \beta_2\y$, and it is easy to show the following properties for $\beta_1$ and
  $\beta_2$.

  \begin{enumerate}
      \item When $\frac{1}{2}\leq \frac{\x}{\y}\leq 2$, we have $\beta_i\leq 3, i = 1, 2$.
      \item When $\frac{\x}{\y}\geq 2$ we have $\beta_1\geq \frac{1}{6}$, and when $\frac{\y}{\x}\geq 2$, we have $\beta_2\geq
      \frac{1}{6}$.
  \end{enumerate}

To satisfy the average power constraints, we always allocate the source powers such that local average power constraints are satisfied, i.e., the average power for each mode is at most $1$.
We consider five different regions which together
cover all possibilities. In the first four regions, we consider the coding schemes for the
corresponding LDM and show that the Gaussian channel can allocate the same rates for all the messages
up to some constant. The last region is unique for Gaussian channel, where following the scheme for the
LDM can be strictly suboptimal. The sum rate is
\begin{align*}
  \Ri {sum} = \frac{1}{2+\delta}(\delta R_\A+R_\B+R_\C+2\Delta R).
\end{align*}

\noindent{\bf Region 1:} $\z\leq \x$ or $\z\leq 1$ or $\y\leq 1$.

In this region we do not use any cooperation ($\delta_B=\delta_C=0$ in
Theorem~\ref{thm:half-duplex}). The scheme reduces to Han and Kobayashi's
scheme for the interference channel\cite{EtkinTseWang, Carleial75}, and it is not hard to show that $\overline{\Cii {sum} {LDM}}$ can be
achieved within 6 bits in this region.

\noindent{\bf Region 2:} $2\y<\x<\z$ and  $\y>1$.

This region corresponds to the case $n_I< n_D< n_C$ for the LDM. The sources share messages with each other and there is no relay.

In this region, $\beta_1\geq \frac{1}{6}$ is a finite constant bounded away
from 0. We set $\bp sd = 0, \Delta R = 0$. In
modes $\B$ and \C, each source uses power $1-\frac{1}{\x}$ to send data to
its own destination and uses power $\frac{1}{\x}$ to share bits with the
other source. By superposition coding, the
following rates are achievable.
\begin{align*}
    R_{\B} = R_\C &=  \log\left(1+\frac{(1-\frac{1}{\x})\x}{2}\right)
           \geq (n_D-1)^+\\
    \delta \bp ss &=  \log\left(1+\frac{\z}{\x}\right)
                  \geq (n_C-n_D-1)^+.
\end{align*}
Therefore we can set $R_\B=R_\C=(n_D-1)^+$ and $\delta \bp ss = (n_C-n_D-1)^+$.

For the virtual channel, we take $R_{V'} = 0$ and set powers
$\sigma_W^2 = \frac{1}{3}, \sigma_U^2 = \frac{1}{3\y}, \Var{X_V} =
\frac{1}{3}$. So destination 3 receives $W_1, U_1, W_2, V_1, U_2$ with
powers $\frac{\x}{3}, \frac{\x}{3\y}, \frac{\y}{3}, \frac{\beta_1\x}{3},
\frac{1}{3}$, respectively, and destination 4 gets $W_2, U_2, W_1, V_2,
U_1$ with powers $\frac{\x}{3}, \frac{\x}{3\y}, \frac{\y}{3},
\frac{\beta_1\x}{3}, \frac{1}{3}$, respectively. It is easy to verify that
the following constraints on non-negatives rates imply all the relevant rate
constraints in Theorem~\ref{thm:IFconf}.
\begin{align*}
  2R_W+R_U+R_V& \leq \log\Big(1+\frac{\beta_1\x}{4}\Big)\\
  R_U+R_W& \leq \log\Big(1+\frac{\frac{\x}{\y}+\y}{4}\Big)\\
  R_U& \leq \log\Big(1+\frac{\x}{4\y}\Big)\\
  R_V& \leq \bp ss
\end{align*}
Hence the following non-negative rates are achievable.
\begin{align*}
  2R_W+R_U+R_V& \leq (n_D-2-\log6)^+\\
  R_U+R_W& \leq (\max(n_D-n_I, n_I)-3)^+\\
  R_U& \leq (n_D-n_I-3)^+\\
  R_V& \leq \bp ss
\end{align*}
Setting $R_\A = 2(R_W+R_U+R_V)$, we can achieve
\begin{align*}
  R_\A &= \min \left\{
  \begin{array}{c}
  (2n_D-4-2\log 6)^+\\
  (2\max(n_D-n_I, n_I)+2\bp ss-6)^+\\
  (2n_D-n_I+\bp ss-5-\log 6)^+\\
  \end{array}
  \right\}.
\end{align*}
Therefore the sum rate is
\begin{align*}
  \Ri {sum} &= \frac{1}{2+\delta}(\delta R_\A+R_\B+R_\C) \geq  \min \left\{u_2'-9, u_3'-6, u_4'-10  \right\}.
\end{align*}
Hence $\overline{\Cii {sum} {LDM}}$ can be achieved within 6 bits in this
region.

\noindent{\bf Region 3:} $2\x<\y\leq \z$ and  $\y>1$.

This region corresponds to the case $n_D < n_I < n_C$ for the LDM. The sources share messages with each other and relay is used when the cooperation link is strong. In particular, we consider two subregions as in the LDM. When $n_C$ is small, we will only use the cooperative private signal to improve the virtual channel sum-rate. But when $n_C$ is big enough to achieve the cut-set bound of the virtual channel, we need to use relaying in modes $\B$ and $\C$ ($\Delta R>0$) to further increase the achievable rate.

We set \bp sd = 0. Firstly, we assume that $x>1$ and consider the following two subregions.
\begin{enumerate}
  \item $\displaystyle \y^{\delta}\x\geq \z$. This subregion corresponds to the case $n_C-n_D\leq \delta n_I$ for the LDM. We set $\Delta R =0$. As in Region~2, we can set $R_\B=R_\C=(n_D-1)^+$ and $\delta \bp ss = (n_C-n_D-1)^+$. For the virtual channel, we choose $R_U = R_{V'} = 0$ and set powers $\sigma_W^2 = \frac{1}{2}, \Var{X_V} = \frac{1}{2}$. Then we apply Theorem~\ref{thm:IFconf} as in Region~2, and get
      $
          \Ri {sum} \geq  \min\left\{u_1'-4, u_2'-11\right\}.
      $
      Hence $\overline{\Cii {sum} {LDM}}$ can be achieved within $7$ bits in this case.
  \item $\displaystyle \y^{\delta}\x\leq \z$. This subregion corresponds to the case $n_C-n_D>\delta n_I$ for the LDM. In modes $\B$ and $\C$, sources use power $\frac{1}{3}$ to send data to its
    own destination and $\frac{1}{3}\sqrt{\frac{\y^{\delta}}{\x\z}}$ and $\frac{1}{3\x}$, respectively, to send to the other source and the other destination, respectively. By superposition coding,
    the following are achievable.
    \begin{align*}
      R_\B =R_\C &=
        \log\left(1+\frac{\frac{\x}{3}}{1+\frac{1}{3}+\frac{1}{3}\sqrt{\frac{\x\y^{\delta}}{\z}}}\right)
        \geq (n_D-\log5)^+\\
      \Delta R &=
        \log\left(1+\frac{\frac{\y}{3\x}}{1+\frac{1}{3}\sqrt{\frac{\y^{\delta+2}}{\x\z}}}\right)
        \geq
        \min(n_I-n_D, \frac{1}{2}(n_C-n_D-\delta
        n_I))-2-\log3-\frac{1}{2}\delta\\
      \Delta R + \delta \bp ss &= \log\left(1+\frac{1}{3}\sqrt{\frac{\z\y^{\delta}}{\x}}\right)
      \geq \frac{1}{2}(n_C+\delta n_I-n_D-1)-\log3
    \end{align*}
    Since the condition $\displaystyle \y^{\delta}\x\leq \z$ implies that $\delta n_I-1\leq n_C-n_D$, it is easy to see that we can set
    \begin{align*}
      R_\B &= R_\C = (n_D - \log 5)^+\\
      \delta \bp ss &= (\delta n_I -1-\log 3)^+\\
      \Delta R &=  (\min(n_I-n_D, \frac{1}{2}(n_C-n_D-\delta
      n_I))-2-\log3-\frac{1}{2}\delta)^+.
    \end{align*}
    For the virtual channel, we use the same scheme as in the previous subregion. Then we apply Theorem~\ref{thm:IFconf} as in Region~2, and get
    $
          \Ri {sum} \geq  \min\left\{u_2', u_4'\right\}-11.
    $
    Hence $\overline{\Cii {sum} {LDM}}$ can be achieved within $7$ bits in this case.
\end{enumerate}

Now we consider the case $x\leq1$. As $n_D=0$, no (significant) direct transmission of data from source to destination is possible; all data must pass through the other source. This can happen in one of two ways: relaying in modes $\B$ and $\C$, and cooperative private message for the virtual channel. We note that the power allocation for $x> 1$ might not satisfy the local power constraints in modes $\B$ and $\C$ now. As in the previous case , we consider the following two subregions separately.

\begin{enumerate}
  \item $\y^{\delta}\geq \z$. In modes $\B$ and $\C$, the sources use all
    their power to send data to the other source and get
    \begin{align*}
      R_\B = R_\C = 0, \quad \delta \bp ss = \log (1+\z)\geq n_C.
    \end{align*}
    For the virtual channel, each source relays the shared data to the other
    destination and the direct link signals are treated as interference. It is easy to show that we can achieve
    \begin{align*}
      R_\A &=  2 \min(\log (1+\frac{\y}{1+\x}), \bp ss) \geq 2(\frac{n_C}{\delta}-2),
    \end{align*}
    Therefore the sum-rate is
    $
      \Ri {sum} \geq u_1'-4.
    $
    Hence $\overline{\Cii {sum} {LDM}}$ can be achieved in this case.
  \item $\y^{\delta} \leq \z$.
    In mode $\B$ and $\C$, each source uses powers
    $\frac{1}{2}\sqrt{\frac{\y^{\delta}}{\z}}$ and $\frac{1}{2}$, respectively,
    to share bits with the other source and the other destination, respectively.
    By superposition coding, the following rates are
    achievable.
    \begin{align*}
      \Delta R =
      &\log\Big(1+\frac{\frac{\y}{2}}{1+\frac{1}{2}\sqrt{\frac{\y^{\delta+2}}{\z}}}\Big)
      \geq \min (n_I, \frac{1}{2}(n_C-\delta n_I))-\frac{\delta}{2}-2\\
      \Delta R+\delta \bp ss&=  \log(1+\frac{1}{2}\sqrt{\y^{\delta}\z})
      \geq  \frac{1}{2}(n_C+\delta n_I)-1.
    \end{align*}
    Therefore we can set
    \begin{align*}
      \delta \bp ss &= \Big(\delta n_I-\frac{3}{2}\Big)^+\\
      \Delta R &=  \Big(\min (n_I, \frac{1}{2}(n_C-\delta
    n_I))-\frac{\delta}{2}-2\Big)^+
    \end{align*}
    For the virtual channel, we use the same scheme as in the previous subregion and achieve
    \begin{align*}
      R_\A = 2\min ((n_I-1), \bp ss) \geq 2\Big(n_I-1-\frac{3}{2\delta}\Big).
    \end{align*}
    Therefore the sum-rate is
    $
      \Ri {sum} \geq \min \left\{u_2', u_4'\right\}-4.
    $
    Hence $\overline{\Cii {sum} {LDM}}$ can be achieved in this case.
\end{enumerate}

\noindent{\bf Region 4:} $\x<\z<\y, 2\x<\y$, and $\z>1$.

This region corresponds to the case $n_D < n_C < n_I$ for the LDM. The sources share messages with each other and the other destinations, and relay is used when both the cooperation link and the interference link are strong. In particular, we consider two subregions as in the LDM. When $n_C$ and $n_I$ are small, we will only use the cooperative private signal and the pre-shared public signal to improve the virtual channel sum-rate. But when $n_C, n_I$ are big enough to achieve the cut-set bound of the virtual channel, we need to use relaying in modes $\B$ and $\C$ (i.e., $\Delta R >0$) to further improve the achievable rate.

Firstly, we assume that $x > 1$ and consider the following two subregions.
\begin{enumerate}
  \item $\y\leq \x\y^{\delta}$ or $n_C-n_D+1\leq \delta(n_I-n_D)$. The condition $\y\leq \x\y^{\delta}$ leads to $n_I\leq n_D+\delta n_I+\delta+1$. In mode $\B, \C$, each source uses power $1-1/\x$
    to send data to its own destination and $1/\x-1/\z$ and $1/\z$ to share bits with the other source and the other destination respectively. By superposition coding, the following rates are achievable.
    \begin{align*}
      R_\B &= R_\C = \log(1+\x)-1 \geq (n_D-1)^+\\
      \delta \bp ss &= \log(1+\frac{\z}{\x})-1\geq (n_C-n_D-2)^+\\
      \delta \bp sd &= \log(1+\frac{\y}{\z})\geq (n_I-n_C-1)^+
    \end{align*}
    Therefore we can set the corresponding rates equal to the lower bounds on right-hand side. By the assumption, we have either $\bp ss\geq n_I-n_D$ or $\bp ss+\bp sd\leq n_I+1$.

    For the virtual channel, we take $R_U=0$ and set powers $\sigma_W^2 =
    \frac{1}{3}, \sigma_{V'}^2 = \frac{1}{3}, \Var{X_V} = \frac{1}{3}$. Then we apply Theorem~\ref{thm:IFconf} as in Region~2, and get
    $
      \Ri {sum} \geq  \min \left\{u_1'-9, u_2'-10 \right\}.
    $
    Hence $\overline{\Cii {sum} {LDM}}$ can be achieved within $6$ bits in this case.
  \item $\y\geq \x\y^{\delta}$. In modes $\B$ and $\C$, each source uses a power of $\frac{1}{3}$ to send data to its own destination and $\frac{1}{3\x}$ and
    $\frac{1}{3\sqrt{\y^{1+\delta}\x^{1-2\delta}}}$ to share bits with the other source and the other destination, respectively.  We note that this is a valid local power allocation since we  have $\y^{1+\delta}\x^{1-2\delta}\geq y^{1-\delta}x\geq 1$. By superposition coding, the following rates are achievable.
    \begin{align*}
      R_\B =  R_\C =& \log\Big(1+\frac{\frac{2}{3}}{1+\frac{1}{3}+\frac{\x}{3\sqrt{\y^{1+\delta}\x^{1-2\delta}}}}\Big)
      \geq (n_D-\log 5)^+\\
      \delta \bp ss+\Delta R =&
      \log\Big(1+\frac{\frac{\z}{3\x}}{1+\frac{\z}{3\sqrt{\y^{1+\delta}\x^{1-2\delta}}}}\Big)
      \geq \min\Big(n_C-n_D,
      \frac{1+\delta}{2}n_I-\frac{1+2\delta}{2}n_D\Big)-2-\log
      3-\frac{\delta}{2}\\
      \delta \bp sd+\Delta R =
      &\log\left(1+\frac{1}{3}\sqrt{\frac{\y^{1-\delta}}{\x^{1-2\delta}}}\right)
      \geq \frac{1-\delta}{2}n_I-\frac{1-2\delta}{2}n_D-\frac{1-\delta}{2}-\log3.
    \end{align*}
    Since the condition $\y\geq \x\y^{\delta}$ implies that $n_I+1\geq n_D+\delta n_I$ and $1-\delta\geq 0$, it is easy to verify that we can set
    \begin{align*}
      \delta \bp ss &=  \delta (n_I-n_D)-3-\log3-\frac{\delta}{2}\\
      \delta \bp sd &=  \delta n_D-\frac{3}{2}+\frac{\delta}{2}-\log 3\\
      \Delta R &=
      \min\Big(n_C-n_D-\delta(n_I-n_D),\frac{1-\delta}{2}n_I-\frac{1}{2}n_D\Big)+1.
    \end{align*}
    For the virtual channel, we use the same scheme as in the previous subregion. Then we apply Theorem~\ref{thm:IFconf} as in Region~2, and get
      $
        R_{sum} \geq \min \left\{u_1', u_2' \right\}-9.
      $
    Hence $\overline{\Cii {sum} {LDM}}$ can be achieved within 6 bits in this case.
\end{enumerate}

Now we consider the case $x\leq1$. As $n_D=0$, no (significant) direct transmission of data from source to destination is possible; all data must pass through the other source. This can happen in one of two ways: relaying in modes $\B$ and $\C$, and cooperative private message for the virtual channel. We note that the power allocation for $x> 1$ might not satisfy the local power constraints in modes $\B$ and $\C$ now. As in the previous case , we consider the following two subregions separately.
\begin{enumerate}
  \item $\y^{\delta}\geq \z$. The analysis here is the same as the corresponding case in Region $3$, i.e., $2\x< \y$, $\x<1<\y\leq \z$ and $\y^{\delta}\geq \z$.
  \item $\y^{\delta}< \z$. In modes $\B$ and $\C$, each source uses powers $\frac{1}{2}$ and $\frac{1}{2\sqrt{\y^{1+\delta}}}$ to share bits with the other source and the other destination,
    respectively. By superposition coding, the following rates are achievable.
    \begin{align*}
      \delta \bp ss+\Delta R =
      &\log\Big(1+\frac{\frac{\z}{2}}{1+\frac{\z}{2\sqrt{\y^{1+\delta}}}}\Big)=  \min\Big(n_C, \frac{1+\delta}{2}n_I\Big)-2\\\
      \Delta R &= \log\Big(1+\frac{1}{2}\sqrt{\y^{1-\delta}}\Big)\geq  \frac{1-\delta}{2}n_I-1
    \end{align*}
    Therefore we can set
    \begin{align*}
      \delta \bp ss &= \delta n_I-3\\
      \Delta R &=  \min\Big(n_C-\delta n_I, \frac{1-\delta}{2}n_I\Big)-1.
    \end{align*}
    For the virtual channel, we use the same scheme as in the previous subregion and achieve
    \begin{align*}
      R_\A = 2\min ((n_I-1), \bp ss) \geq 2\Big(n_I-1-\frac{3}{\delta}\Big).
    \end{align*}
    Therefore the sum-rate is
    $
      \Ri {sum} \geq \min \left\{u_1', u_2'\right\}-4.
    $
    Hence $\overline{\Cii {sum} {LDM}}$ can be achieved in this case.
\end{enumerate}

\noindent{\bf Region 5:} $\frac{1}{2}\leq \frac{\x}{\y}\leq 2, \z> \x, \z>1$, and  $\y>1$.

This region corresponds to the case $n_D = n_I$ for the LDM. In LDM, the channel is degenerated and cooperation is not helpful. However, in the Gaussian case, whether the channel is degenerated further depends on the phase information of the channel, which is not captured by the LDM.

When {$\x<1$}, we have $\y\leq 2\x<2$ and get $n_D = n_I = 0$. Therefore, $\overline{\Cii {sum} {LDM}}=0$, which can be achieved trivially.  Below we assume {$\x\geq 1$}.

In this region, we have $n_I-2\leq n_D\leq n_I+2$ and $\beta_1\leq 3$. We set $\bp sd = 0, \Delta R = 0$. As in Region~2, we can set $R_\B=R_\C=(n_D-1)^+$ and $\delta \bp ss = (n_C-n_D-1)^+$.
For the virtual channel, we set rates $R_U = R_{V'} = 0$ and powers
$\sigma_W^2 = \frac{1}{2}, \Var{X_V} = \frac{1}{2}$. By Theorem~\ref{thm:IFconf}, non-negative rates which satisfy the following
conditions are achievable \footnote{Redundant conditions are not listed here.
Also, conditions corresponding to error events which involve an unwanted
message along with zero-rate messages are also not listed. For example, the
rate constraint on $R_{W_2} + R_{U_1}$ is avoided since it corresponds to
the error event of destination~3 making an error on the unwanted message
$m_{W_2}$ and the message $m_{U_1}$ which is absent in this case.}

\begin{align*}
  2R_W+R_V& \leq  \log\Big(1+\frac{\y}{2}\Big)\\
  R_W+R_V& \leq  \log\Big(1+\frac{\x}{2}\Big)\\
  R_V& \leq  \log\Big(1+\frac{\beta_1\x}{2}\Big)\wedge \bp ss
\end{align*}

Therefore, for the virtual channel, we can achieve
\begin{align*}
  R_\A &= \min \left\{
  \begin{array}{c}
  (2n_D-2)^+\\
  (n_I+\bp ss-1)^+\\
  (n_I+\log(1+\frac{\beta_1\x}{2})-1)^+
  \end{array}
  \right\}.
\end{align*}

By the assumption, it is not hard to verify that
$
  \Ri{sum} \geq \min \left\{u_2'-4, u_4'-6, u_4-6 \right\}.
$
Hence $\overline{\Cii {sum} {LDM}}$ can be achieved within 2 bits in this
region.

\end{proof}

\section{Proof of Theorem~\ref{thm:converse}}
\label{app:sumrateconverse}

We prove the outerbound by first proving an outerbound for a more general channel with generalized feedback of which ours is a special case. Specifically, we consider the following two user interference channel $p(y_1,y_2,y_3,y_4|x_1,x_2)$ whose input alphabets are ${\mathcal X}_1$, ${\mathcal X}_2$ respectively for the first and second sources, output alphabets are ${\mathcal Y}_3$,  ${\mathcal Y}_4$ respectively for first and second destinations, and ${\mathcal Y}_1$ and ${\mathcal Y}_2$ respectively are the output alphabets (of the generalized feedback) for first and second sources. Let $W_1$ and $W_2$ be the messages of the first and second sources. At time $t$, the first source's signal $X_{1,t}$ may depend only on its past outputs $Y_1^{t-1}$ and its message $W_1$, similary for the second source. We also have cost functions $c_1:{\mathcal X}_1\rightarrow {\mathbb R}_+$ and $c_2:{\mathcal X}_2\rightarrow {\mathbb R}_+$ and there are average cost constraints $P_1$ and $P_2$, respectively, on the first and second sources.
Along the lines of~\cite{TelatarTse07}, we focus on channels of the following form $p(y_1,y_2,y_3,y_4|x_1,x_2) = \sum_{u_1,u_2} p(u_1,u_2,y_1,y_2,y_3,y_4|x_1,x_2)$, where
\[  p(u_1,u_2,y_1,y_2,y_3,y_4|x_1,x_2) = p(u_1,y_2|x_1)p(u_2,y_1|x_2)\delta(y_3-f_3(x_1,u_2))\delta(y_4-f_4(x_2,u_1)),\]
where $U_1$ and $U_2$ take values in alphabets ${\mathcal U}_1$ and ${\mathcal U}_2$ respectively, and, for every $x_1\in{\mathcal X}_1$, the map $f_3(x_1,.):{\mathcal U}_2 \rightarrow {\mathcal Y}_3$ defined as $u_2 \mapsto f_3(x_1,u_2)$ is invertible, and similarly, for $f_4$. The capacity region of this channel may be defined as usual.

The following gives an outerbound on the capacity region of the above channel.
\begin{thm}\label{thm:ifc-general-outerbound}
If $(R_1,R_2)$ belongs to the capacity region of the above channel, there there is a $p(q,x_1,x_2)$ with ${\mathbb E}[c_1(X_1)]\leq P_1$ and ${\mathbb E}[c_2(X_2)]\leq P_2$ such that for the joint distribution
\[  p(u_1,u_2,y_1,y_2,y_3,y_4,x_1,x_2)= p(u_1,u_2,y_2,y_3,y_4|x_1,x_2)p(q,x_1,x_2),\]
\begin{align}
R_1 &\leq I(X_1;Y_2,Y_3|X_2,Q), \label{eq:ifc-gen-obR1}\\
R_2 &\leq I(X_2;Y_1,Y_4|X_1,Q), \label{eq:ifc-gen-obR2}\\
R_1+R_2 &\leq I(X_1,X_2;Y_3,Y_4|Q), \label{eq:ifc-gen-obsum0}\\
R_1+R_2 &\leq I(X_1;Y_2,Y_3|Y_4,X_2,Q) + I(X_1,X_2;Y_4|Q), \label{eq:ifc-gen-obsum1}\\
R_1+R_2 &\leq I(X_2;Y_4,Y_1|Y_3,X_1,Q) + I(X_1,X_2;Y_3|Q), \label{eq:ifc-gen-obsum2}\\
R_1+R_2 &\leq I(X_1,X_2;Y_1,Y_3|U_1,Y_2,Q) + I(X_1,X_2;Y_2,Y_4|U_2,Y_1,Q)).  \label{eq:ifc-gen-R1plusR2}
\end{align}
\end{thm}
\begin{IEEEproof}
The bounds \eqref{eq:ifc-gen-obR1}-\eqref{eq:ifc-gen-obR2} are simple cutset bounds. The next three \eqref{eq:ifc-gen-obsum0}-\eqref{eq:ifc-gen-obsum2} were proved in~\cite[Theorem II.1]{TuninettiITA10}. We omit the proofs here. The last one is new and its proof follows.
By Fano's inequality, for any $\epsilon>0$, we have a sufficiently large blocklength $n$ such that
\begin{align}
n(R_1 - \epsilon) &\leq I(W_1;Y_3^n\notag)\\
	&\leq I(W_1;Y_3^n,U_1^n,Y_1^n,Y_2^n)\notag\\
	&= H(Y_3^n,U_1^n,Y_1^n,Y_2^n) - H(Y_3^n,U_1^n,Y_1^n,Y_2^n|W_1)\notag\\
	&= H(U_1^n,Y_1^n,Y_2^n)  + H(Y_3^n|U_1^n,Y_1^n,Y_2^n) - H(Y_3^n,Y_1^n,Y_2^n|W_1) - H(U_1^n|Y_3^n,Y_1^n,Y_2^n,W_1).
\label{eq:ifc-gen-ob-proof1}
\end{align}
But,
\begin{align*}
H(Y_3^n|U_1^n,Y_1^n,Y_2^n) \leq& \sum_{t=1}^n H(Y_{3,t}|U_{1,t},Y_{1,t},Y_{2,t}),\\
H(Y_3^n,Y_1^n,Y_2^n|W_1) = &\sum_{t=1}^n H(Y_{3,t},Y_{1,t},Y_{2,t}|W_1,Y_3^{t-1},Y_1^{t-1},Y_2^{t-1})\\
	=&\sum_{t=1}^n H(Y_{3,t},Y_{1,t},Y_{2,t}|X_1^t,W_1,Y_3^{t-1},Y_1^{t-1},Y_2^{t-1})\\
	=&\sum_{t=1}^n H(U_{2,t},Y_{1,t},Y_{2,t}|X_1^t,W_1,U_2^{t-1},Y_1^{t-1},Y_2^{t-1})\\
	=&\sum_{t=1}^n H(U_{2,t},Y_{1,t}|X_1^t,W_1,U_2^{t-1},Y_1^{t-1},Y_2^{t-1}) + H(Y_{2,t}|X_1^t,W_1,U_2^{t},Y_1^{t},Y_2^{t-1})\\
	\stackrel{\text{(a)}}{=}&\sum_{t=1}^n H(U_{2,t},Y_{1,t}|U_2^{t-1},Y_1^{t-1},Y_2^{t-1}) + H(Y_{2,t}|X_{1,t},U_{2,t},Y_{1,t})\\
	=&\sum_{t=1}^n (H(U_{2,t},Y_{1,t},Y_{2,t}|U_2^{t-1},Y_1^{t-1},Y_2^{t-1}) - H(Y_{2,t}|U_2^t,Y_1^t,Y_2^{t-1})) \\&+ H(Y_{2,t}|X_{1,t},U_{2,t},Y_{1,t})\\
	\geq&\sum_{t=1}^n H(U_{2,t},Y_{1,t},Y_{2,t}|U_2^{t-1},Y_1^{t-1},Y_2^{t-1}) - H(Y_{2,t}|U_{2,t},Y_{1,t}) \\&+ H(Y_{2,t}|X_{1,t},X_{2,t},U_{2,t},Y_{1,t})\\
	\geq& H(U_2^n,Y_1^n,Y_2^n) - \sum_{t=1}^n  I(X_{1,t},X_{2,t};Y_{2,t}|U_{2,t},Y_{1,t}),\\
H(U_1^n|Y_3^n,Y_1^n,Y_2^n,W_1) =& H(U_1^n|X_1^n,Y_3^n,Y_1^n,Y_2^n,W_1)\\
	=& \sum_{t=1}^n H(U_{1,t}|U_1^{t-1},X_1^n,Y_3^n,Y_1^n,Y_2^n,W_1)\\
	=& \sum_{t=1}^n H(U_{1,t}|X_{1,t},Y_{2,t})\\
	=& \sum_{t=1}^n H(U_{1,t}|X_{1,t},X_{2,t},Y_{1,t},Y_{2,t},U_{2,t})\\
	=& \sum_{t=1}^n H(Y_{4,t}|X_{1,t},X_{2,t},Y_{1,t},Y_{2,t},U_{2,t}),
\end{align*}
where (a) follows from the fact that $(W_1,X_1^t) - (Y_1^{t-1},Y_2^{t-1}) - X_{2,t} - (U_{2,t},Y_{1,t})$ is a Markov chain.
Substituting in \eqref{eq:ifc-gen-ob-proof1}, we get
\begin{align*}
n(R_1-\epsilon) \leq& H(U_1^n,Y_1^n,Y_2^n) - H(U_2^n,Y_1^n,Y_2^n)
	\\&
	+ \left(\sum_{t=1}^n  I(X_{1,t},X_{2,t};Y_{2,t}|U_{2,t},Y_{1,t}) +  H(Y_{3,t}|U_{1,t},Y_{1,t},Y_{2,t}) - H(Y_{4,t}|X_{1,t},X_{2,t},Y_{1,t},Y_{2,t},U_{2,t}) \right).\\
\intertext{Similarly,}
n(R_2-\epsilon) \leq& H(U_2^n,Y_1^n,Y_2^n) - H(U_1^n,Y_1^n,Y_2^n)
	\\&
	+ \left(\sum_{t=1}^n  I(X_{1,t},X_{2,t};Y_{1,t}|U_{1,t},Y_{2,t}) +  H(Y_{4,t}|U_{2,t},Y_{1,t},Y_{2,t}) - H(Y_{3,t}|X_{1,t},X_{2,t},Y_{1,t},Y_{2,t},U_{1,t}) \right).
\end{align*}
Adding up,
\begin{align*}
n(R_1+R_2-2\epsilon) \leq &\sum_{t=1}^n  I(X_{1,t},X_{2,t};Y_{2,t}|U_{2,t},Y_{1,t}) + I(X_{1,t},X_{2,t};Y_{1,t}|U_{1,t},Y_{2,t})
	\\&
 		+ I(X_{1,t},X_{2,t};Y_{3,t}|U_{1,t},Y_{1,t},Y_{2,t}) + I(X_{1,t},X_{2,t};Y_{4,t}|U_{2,t},Y_{1,t},Y_{2,t})\\
	= &\sum_{t=1}^n I(X_{1,t},X_{2,t};Y_{1,t},Y_{3,t}|U_{1,t},Y_{2,t}) +  I(X_{1,t},X_{2,t};Y_{2,t},Y_{4,t}|U_{2,t},Y_{1,t}).
\end{align*}
Proceeding as usual by picking $Q$ to be uniformly distributed over $\{1,\ldots,n\}$ and letting $X_1=X_{1,Q}$ and so on, we obtain \eqref{eq:ifc-gen-R1plusR2}.

\end{IEEEproof}


We will use the above theorem to prove our outerbound. Without loss of generality, we may rewrite our channel (by absorbing phases into the inputs and outputs) in the following symmetric form.
\begin{align}
Y_{1,t}&= (|h_{21}|X_{2,t}+Z_{1,t})1_{S_{1,t}=0} \label{eq:symchrewrite1}\\
Y_{2,t}&= (|h_{12}|X_{1,t}+Z_{2,t})1_{S_{2,t}=0}\\
Y_{3,t}&= |h_{13}|e^{j\theta/2}X_{1,t}1_{S_{1,t}=1} + |h_{23}|X_{2,t}1_{S_{2,t}=1} + Z_{3,t},\\
Y_{4,t}&= |h_{14}|X_{1,t}1_{S_{1,t}=1} + |h_{24}|e^{j\theta/2}X_{2,t}1_{S_{2,t}=1} + Z_{4,t}. \label{eq:symchrewrite4}
\end{align}
Recall that $\theta=\theta_{13}+\theta_{24}-\theta_{14}-\theta_{23}$, and we assume $|h_{13}|^2=|h_{24}|^2=\SNR, |h_{14}|^2=|h_{23}|^2=\INR, |h_{12}|^2=|h_{21}|^2=\CNR$.
Notice that our channel fits the model of Theorem~\ref{thm:ifc-general-outerbound} if we identify the first and second sources' channel inputs as $(X_1,S_1)$ and $(X_2,S_2)$ respectively, the outputs for the two sources are $Y_1$ and $Y_2$ respectively, and $U_1 = h_{14}X_{1}1_{S_1=1} + Z_{4}$, $U_2 = h_{23}X_{2}1_{S_{2}=1} + Z_{3}$. The two destinations' channel outputs are $Y_3 = h_{13}X_11_{S_1=1} + U_2$, and $Y_4 = h_{24}X_21_{S_2=1} + U_1$ respectively. And the cost functions are $c_1(x_1,s_1)=|x_1|^21_{s_1=1}$ and $c_2(x_2,s_2)=|x_2|^2 1_{s_2=1}$ with unit power constraints $P_1=P_2=1$.

Using Theorem~\ref{thm:ifc-general-outerbound} we get an upperbound on the sum-rate, namely, the minimum of the right hand sides of  \eqref{eq:ifc-gen-obsum0}-\eqref{eq:ifc-gen-R1plusR2} and the sum of the right hand sides of \eqref{eq:ifc-gen-obR1} and \eqref{eq:ifc-gen-obR2}, maximized over $p(q,x_1,x_2)$ which satisfy the power constraints. First of all, let us notice that when the channel and the power constraints are symmetric, as is the case for the channel in \eqref{eq:symchrewrite1}-\eqref{eq:symchrewrite4}, without loss of generality, we may assume that ${\mathbb P}(S_1=1,S_2=0)={\mathbb P}(S_1=0,S_2=1)$. Let $\delta={\mathbb P}(S_1=1,S_2=1)/{\mathbb P}(S_1=1,S_2=0)$, and $\gamma={\mathbb P}(S_1=0,S_2=0)/{\mathbb P}(S_1=1,S_2=0)$. Also, let
\begin{align*}
P_{1\A}&={\mathbb E}\left[|X_1|^2 \mid S_1=S_2=1\right],&
P_{1\B}&={\mathbb E}\left[|X_1|^2 \mid S_1=1,S_2=0\right],& P_{1\C}&=0,\text{ and}\\
P_{2\A}&={\mathbb E}\left[|X_2|^2 \mid S_1=S_2=1\right],&  P_{2\B}&=0,&
P_{2\C}&={\mathbb E}\left[|X_2|^2 \mid S_1=0,S_2=1\right].
\end{align*}
We have ${\mathbb E}\left[|X_i|^21_{S_i=1}\right]=(\delta P_{i\A}+ P_{i\B} + P_{i\C})/(2+\delta+\gamma) \leq 1$, for $i=1,2$. We now derive the outerbounds:
\begin{enumerate}

\item  $Cut(\delta)$\\

From \eqref{eq:ifc-gen-obR1}-\eqref{eq:ifc-gen-obR2},
\begin{align*}
R_1+R_2 &\leq I(X_1,S_1;Y_2,Y_3|X_2,S_2,Q) +  I(X_2,S_2;Y_1,Y_4|X_1,S_1,Q)\\
	&\leq H(S_1)+H(S_2)+ I(X_1;Y_2,Y_3|X_2,Q,S_1,S_2) + I(X_2;Y_1,Y_4|X_1,Q,S_1,S_2)\\
	&\leq 2 + (I(X_1;Y_3|Q,S_1=S_2=1)+I(X_2;Y_4|Q,S_1=S_2=1)){\mathbb P}(S_1=S_2=1)
	\\&\qquad
	 + I(X_1;Y_2,Y_3|Q,S_1=1,S_2=0){\mathbb P}(S_1=1,S_2=0)
	\\&\qquad
	 + I(X_2;Y_1,Y_4|Q,S_1=0,S_2=1){\mathbb P}(S_1=0,S_2=1)\\
	&\leq 2 + \frac{\delta}{2+\delta+\gamma}\left(\log(1+xP_{1\A})+\log(1+xP_{2\A})\right)
	\\&\qquad
	 + \frac{1}{2+\delta+\gamma}\log(1+(x+z)P_{1\B}) + \frac{1}{2+\delta+\gamma}\log(1+(x+z)P_{2\C})\\
	&\leq 2 + \frac{\delta}{2+\delta}\left(\log(1+xP_{1\A})+\log(1+xP_{2\A})\right)
	\\&\qquad
	 + \frac{1}{2+\delta}\log(1+(x+z)P_{1\B}) + \frac{1}{2+\delta}\log(1+(x+z)P_{2\C}).
\end{align*}

\item $Z(\delta)$\\
From \eqref{eq:ifc-gen-obsum2},
\begin{align*}
R_1+R_2 &\leq I(X_2,S_2;Y_1,Y_4|Y_3,X_1,S_1,Q) + I(X_1,S_1,X_2,S_2;Y_3|Q)\\
	&\leq H(S_2) + H(S_1,S_2) + I(X_2;Y_1,Y_4|Y_3,X_1,Q,S_1,S_2) + I(X_1,X_2;Y_3|Q,S_1,S_2)\\
	&\leq 3 + (I(X_2;Y_4|Y_3,X_1,Q,S_1=S_2=1) + I(X_1,X_2;Y_4|Q,S_1=S_2=1)){\mathbb P}(S_1=S_2=1)
	\\&\qquad
	 + I(X_1;Y_3|Q,S_1=1,S_2=0){\mathbb P}(S_1=1,S_2=0)
	\\&\qquad
	 + I(X_2;Y_1,Y_3,Y_4|Y_3,Q,S_1=0,S_2=1){\mathbb P}(S_1=0,S_2=1)\\
	&\leq 3 + \frac{\delta}{2+\delta+\gamma}\left(\log\left(1+\frac{\x P_{2A}}{1+\y P_{2A}}\right) +  \log(1+2\x P_{1A}+2\y P_{2A})\right)
	\\&\qquad
	 +  \frac{1}{2+\delta+\gamma}\log(1+\x P_{1B})  + \frac{1}{2+\delta+\gamma}\log(1+(\x+\y+\z)P_{2C}\\
	&\leq 3 + \frac{\delta}{2+\delta}\left(\log\left(1+\frac{\x P_{2A}}{1+\y P_{2A}}\right) +  \log(1+2\x P_{1A}+2\y P_{2A})\right)
	\\&\qquad
	 +  \frac{1}{2+\delta}\log(1+\x P_{1B})  + \frac{1}{2+\delta}\log(1+(\x+\y+\z)P_{2C}.
\end{align*}

\item $V(\delta)$\\
From \eqref{eq:ifc-gen-R1plusR2},
\begin{align*}
R_1+R_2 &\leq I(X_1,S_1,X_2,S_2;Y_1,Y_3|U_1,Y_2,Q) + I(X_1,S_1,X_2,S_2;Y_2,Y_4|U_2,Y_1,Q))\\
	&\leq 2H(S_1,S_2) + I(X_1,X_2;Y_1,Y_3|U_1,Y_2,Q,S_1,S_2) + I(X_1,X_2;Y_2,Y_4|U_2,Y_1,Q,S_1,S_2)\\
	&\leq 4 + (I(X_1,X_2;Y_3|U_1,Q,S_1=S_2=1)+ I(X_1,X_2;Y_4|U_2,Q,S_1=S_2=1)){\mathbb P}(S_1=S_2=1)
	\\&\quad
	 + (I(X_1;Y_3|U_1,Y_2,Q,S_1=1,S_2=0)+ I(X_1;Y_2,U_1|Q,S_1=1,S_2=0)){\mathbb P}(S_1=1,S_2=0)
	\\&\quad
	 + (I(X_2;Y_1,U_2|Q,S_1=0,S_2=1)+ I(X_2;Y_4|U_2,Y_1,Q,S_1=0,S_2=1)){\mathbb P}(S_1=0,S_2=1)\\
	&\leq 4 + (I(X_1,X_2;Y_3|U_1,Q,S_1=S_2=1)+ I(X_1,X_2;Y_4|U_2,Q,S_1=S_2=1)){\mathbb P}(S_1=S_2=1)
	\\&\qquad
	 + I(X_1;Y_3,U_1,Y_2|Q,S_1=1,S_2=0){\mathbb P}(S_1=1,S_2=0)
	\\&\qquad
	 + I(X_2;Y_4,U_2,Y_1|Q,S_1=0,S_2=1){\mathbb P}(S_1=0,S_2=1)\\
	&\leq 4 + \frac{\delta}{2+\delta+\gamma}\left(\log\left(1+\y P_{2A}+\frac{2\x P_{1A}+\y P_{2A}}{1+\y P_{1A}}\right)+\log\left(1+\y P_{1A}+\frac{2\x P_{2A}+\y P_{1A}}{1+\y P_{2A}}\right)\right)\\
	\\&\qquad
	 + \frac{1}{2+\delta+\gamma}\log(1+(\x+\y+\z)P_{1B}) + \frac{1}{2+\delta+\gamma}\log(1+(\x+\y+\z)P_{2C})\\
	&\leq  4 + \frac{\delta}{2+\delta}\left(\log\left(1+\y P_{2A}+\frac{2\x P_{1A}+\y P_{2A}}{1+\y P_{1A}}\right)+\log\left(1+\y P_{1A}+\frac{2\x P_{2A}+\y P_{1A}}{1+\y P_{2A}}\right)\right)\\
	\\&\qquad
	 + \frac{1}{2+\delta}\log(1+(\x+\y+\z)P_{1B}) + \frac{1}{2+\delta}\log(1+(\x+\y+\z)P_{2C}).
\end{align*}

\item $Cut'(\delta)$\\
From \eqref{eq:ifc-gen-obsum0},
\begin{align*}
R_1+R_2 &\leq I(X_1,S_1,X_2,S_2;Y_3,Y_4|Q)\\
	&\leq H(S_1,S_2) + I(X_1,X_2;Y_3,Y_4|Q,S_1,S_2)\\
	&\leq 2 + I(X_1,X_2;Y_3,Y_4|Q,S_1=S_2=1){\mathbb P}(S_1=S_2=1)
	\\&\qquad
	 + I(X_1;Y_3,Y_4|Q,S_1=1,S_2=0){\mathbb P}(S_1=1,S_2=0)
	\\&\qquad
	 + I(X_2;Y_3,Y_4|Q,S_1=0,S_2=1){\mathbb P}(S_1=0,S_2=1)\\
	&\stackrel{\text{(a)}}{\leq} 2 + \frac{\delta}{2+\delta+\gamma}\left( \log(1+2(\x+\y)(P_{1A}+P_{2A})+P_{1A}P_{2A}(\x^2+\y^2-2\x\y\cos\theta))\right)
	\\&\qquad
	 +  \frac{1}{2+\delta+\gamma}\log(1+(x+y)P_{1B}) +  \frac{1}{2+\delta+\gamma}\log(1+(x+y)P_{2C})\\
	&\leq  2 + \frac{\delta}{2+\delta}\left( \log(1+2(\x+\y)(P_{1A}+P_{2A})+P_{1A}P_{2A}(\x^2+\y^2-2\x\y\cos\theta))\right)
	\\&\qquad
	 +  \frac{1}{2+\delta}\log(1+(x+y)P_{1B}) +  \frac{1}{2+\delta}\log(1+(x+y)P_{2C}),
\end{align*}
where (a) follows from the fact that
\begin{align*}
&I(X_1,X_2;Y_3,Y_4|Q,S_1=S_2=1)\\
	&= h(Y_3,Y_4|Q,S_1=S_2=1) - h(Y_3,Y_4|X_1,X_2,Q,S_1=S_2=1)\\
	&\leq \log(\det K), \text{ where $K$ is the covariance matrix of $(Y_3,Y_4)$}\\
	&\leq  1+(\x^2+\y^2)(1-|\rho|^2)P_{1A}P_{2A}+(\x+\y)(P_{1A}+P_{2A})+2Re(h_{13}h_{23}^*\rho)\sqrt{P_{1A}P_{2A}}\\
	&\qquad + 2Re(h_{14}h_{24}^*\rho)\sqrt{P_{1A}P_{2A}}-2Re(h_{13}h_{23}^*h_{14}^*h_{24})(1-|\rho|^2)P_{1A}P_{2A}\\
	&\leq 1+(\x+\y)(P_{1A}+P_{1A})+4\sqrt{\x\y}|\rho|\sqrt{P_{1A}P_{1A}}\cos\frac{\theta}{2}+(\x^2+\y^2-2\x\y\cos\theta)(1-|\rho|^2)P_{1A}P_{2A}\\
	&\leq \log(1+2(\x+\y)(P_{1A}+P_{2A})+P_{1A}P_{2A}(\x^2+\y^2-2\x\y\cos\theta)).
\end{align*}

\end{enumerate}

It remains to show that
 $\overline{\Cii {sum}{HD}}\leq\overline{\Ci {sum}}+7$.
By power constraint, we have $P_{1A}\leq
\frac{2+\delta}{\delta},P_{2A}\leq \frac{2+\delta}{\delta},
P_{1B}\leq 2+\delta, P_{2C}\leq 2+\delta$.

In $Cut(\delta), Z(\delta), Cut'(\delta)$, each term is a monotone increasing function of
$P_{iA}, P_{iB}, P_{iC}, i = 1, 2$, so

\begin{align*}
  Cut(\delta) & \leq  2+ \frac{1}{2+\delta}\Big[\delta \log(1+\x\frac{2+\delta}{\delta})+\delta \log(1+\x\frac{2+\delta}{\delta})\\
        &\log(1+(\x+\z)(2+\delta))+\log(1+(\x+\z)(2+\delta))\Big]\\
    Z(\delta) & \leq 3+ \frac{1}{2+\delta}\Big[\delta \log(1+2\x\frac{2+\delta}{\delta}+2\y\frac{2+\delta}{\delta})+\log(1+\x(2+\delta))\\
        & \qquad + \log(1+(\x+\y+\z)(2+\delta))+\delta\log(1+\frac{\x\frac{2+\delta}{\delta}}{1+\y\frac{2+\delta}{\delta}})\Big]\\
    Cut'(\delta) & \leq  2+ \frac{1}{2+\delta}\Big[\delta\log(1+2(\x+\y)(\frac{2+\delta}{\delta}+\frac{2+\delta}{\delta})+(\frac{2+\delta}{\delta})^2(\x^2+\y^2-2\x\y\cos\theta))\\
        & \qquad + \log(1+(\x+\y)(2+\delta))+\log(1+(\x+\y)(2+\delta))\Big].
\end{align*}

In $V(\delta)$, observe that
\begin{align*}
  &1+\y P_{2A}+\frac{2\x P_{1A}+\y P_{2A}}{1+\y P_{1A}}\\
  & \leq  1+\y\frac{2+\delta}{\delta}+\frac{2\x P_{1A}+\y\frac{2+\delta}{\delta}}{1+\y P_{1A}}\\
  & \leq  \max \left\{
  \begin{array}{c}
  1+\y\frac{2+\delta}{\delta}+\frac{(2\x+\y)\frac{2+\delta}{\delta}}{1+\y\frac{2+\delta}{\delta}}\\
  1+2\y\frac{2+\delta}{\delta}
  \end{array}
  \right\}
\end{align*}
So we have
\begin{align*}
  V(\delta) & \leq  4+ \frac{1}{2+\delta}\Big[\delta\log\left(\max \left\{
  \begin{array}{c}
  1+\y\frac{2+\delta}{\delta}+\frac{(2\x+\y)\frac{2+\delta}{\delta}}{1+\y\frac{2+\delta}{\delta}}\\
  1+2\y\frac{2+\delta}{\delta}
  \end{array}
  \right\}\right)+\log(1+(\x+\y+\z)(2+\delta))\\
        &
        +\delta\log\left(\max \left\{
  \begin{array}{c}
  1+\y\frac{2+\delta}{\delta}+\frac{(2\x+\y)\frac{2+\delta}{\delta}}{1+\y\frac{2+\delta}{\delta}}\\
  1+2\y\frac{2+\delta}{\delta}
  \end{array}
  \right\}\right)+\log(1+(\x+\y+\z)(2+\delta))\Big]\\
\end{align*}

Comparing them term by term with $u_i, i = 1,2, 3, 4$, then we get
\begin{align*}
  Cut(\delta)-u_1& \leq {2+} \frac{1}{2+\delta}\Big[\delta \log\frac{2+\delta}{\delta}+\delta \log\frac{2+\delta}{\delta}+\log(2+\delta)+\log(2+\delta)\Big]\\
  Z(\delta)-u_2 & \leq  {3+} \frac{1}{2+\delta}\Big[\delta \log\frac{2+\delta}{\delta}+\log(2+\delta)+\log(2+\delta)+\delta\log\frac{2+\delta}{\delta}\Big]\\
  V(\delta)-u_3 & \leq  {4+} \frac{1}{2+\delta}\Big[\delta \log\frac{2+\delta}{\delta}+\log(2+\delta)+\delta\log\frac{2+\delta}{\delta}+\log(2+\delta)\Big]\\
    Cut'(\delta)-u_4 & \leq {2+} \frac{1}{2+\delta}\Big[\delta\log\left(\frac{2+\delta}{\delta}\right)^2+\log(2+\delta)+\log(2+\delta)\Big].
\end{align*}

For $\delta\geq 0$,
\begin{align*}
  \frac{\delta}{2+\delta}\log(\frac{2+\delta}{\delta})\leq\frac{1}{e\ln2}\ ,
  \quad\frac{1}{2+\delta}\log(2+\delta)\leq \frac{1}{e\ln2}.
\end{align*}

So we can conclude that
\begin{align*}
  \overline{\Cii {sum} {HD}} &= \max_{\delta}\min(Cut(\delta), Z(\delta),
  V(\delta), Cut'(\delta))\\
  & \leq \max_{\delta}\min(u_1, u_2,
  u_3, u_4)+\frac{4}{e\ln2} {+4}\leq \overline{\Ci {sum}}+7.
\end{align*}

\section{Proof of Theorem~\ref{thm:cogachievability}}
\label{app:cogachievability}

As in the sum-rate case, we will prove this achievability result in two
steps. Instead of directly comparing $\overline{\CRo}$ with the rate
achievable by the coding scheme in section~\ref{sec:achievability}, we will
first show that the $\overline{\CRo}$ is within a constant of
$\overline{\CRoi {LDM}}$, a quantity we define below inspired by the result
for the linear deterministic model. We will then prove that the coding
scheme in section~\ref{sec:achievability} can be used to achieve an $R_1$
which is within $R_0$ of the point-to-point capacity $C_0=\log(1+SNR_1)$ of
the primary link and an $R_2$ which is within a constant of
$\overline{\CRoi{LDM}}$.  Specifically, we prove the following two lemmas
which together imply Theorem~\ref{thm:cogachievability}. To simplify the notation, let
$\x_i = \SNR_i, \y_i = \INR_i, \z = \CNR, i = 1, 2$, and define $n_i = \lfloor\log \x_i\rfloor^+,
\alpha_i=\lfloor\log \y_i\rfloor^+, \beta=\lfloor\log \z\rfloor^+, i = 1, 2. $
\begin{lem}\label{lem:cogGaussianLDMlink}
  Define
\begin{align*}
  \overline{\CRoi{LDM}} = \max_{\delta}\overline{\CRoi{LDM}}(\delta) =
\max_{\delta>0}\min(u_1'-10-2R_0, u_2'-5-R_0, u_3'-5-R_0, u_4'),
\end{align*}
\indent where
\begin{align*}
  u_1' &=  \frac{1}{1+\delta}n_2\\
  u_2' &=  \frac{1}{1+\delta}[n_2\vee\alpha_2-\alpha_2\wedge n_1+\delta(\beta\vee\alpha_2\vee n_1-n_1)]\\
  u_3' &=  \frac{1}{1+\delta}[(\alpha_1-n_1)^++(n_2-\alpha_1)^+]\\
  u_4' &=  \frac{1}{1+\delta}[(\alpha_1-n_1)^+-\alpha_2\wedge
n_1+(n_2-\alpha_1)\vee\alpha_2+\delta(\beta\vee\alpha_2\vee n_1-n_1)].
\end{align*}
Then $\overline{\CRo}<\overline{\CRoi{LDM}}+13+2R_0$.
\end{lem}
\begin{proof}
  It is easy to verify that $u_1\leq u_1'+3, u_2\leq u_2'+8+R_0, u_3\leq u_3'+8+R_0$ and $u_4\leq u_4'+13+2R_0$. So we get the result.
\end{proof}

Next we show that the secondary user can achieve $\overline{\CRoi{LDM}}$ within a constant given that the primary user achieves a rate within $R_0$ of its link capacity.

\begin{lem}\label{thm:cogGaussianLDMachievability}
For $R_0>7$, $(R_1,R_2)= (C_0-R_0,\overline{\CRoi{LDM}}-10)$ is
achievable.\qed
\end{lem}

Before proving Lemma~\ref{thm:cogGaussianLDMachievability}, we first prove the following $R_0$-capacity result for the interference channel, i.e., the cognitive rate achievable without source cooperation.
\begin{lem}\label{lem:cog_IF}
  For $R_0\geq 7$, $\overline{\Cii{cog}{IFC-LDM}}\leq \Cii{$R_0$}{IFC}+1$, where
  \begin{align*}
    \overline{\Cii{cog}{IFC-LDM}} = \min\left(
    \begin{array}{c}
    n_2\\
    n_2\vee \alpha_2-\alpha_2\wedge n_1\\
    (\alpha_1-n_1)^++(n_2-\alpha_1)^+\\
    (\alpha_1-n_1)^+-\alpha_2\wedge n_1+(n_2-\alpha_1)\vee \alpha_2
    \end{array}
    \right)
  \end{align*}
  and $C_{R_0}^{IFC}$ is the $R_0$-capacity for the interference channel.
\end{lem}
\begin{proof}
Let $\overline{C^{IFC}}$ be the outer bound to the interference channel capacity
region derived in \cite{EtkinTseWang}. From the achievability result there,
we know that given $R_1 = \log(1+SNR_1)-R_0$, $R_2$ is achievable if
\begin{align*}
  (\log(1+SNR_1)-R_0+1, R_2+1)\in \overline{\Cii{}{IFC}}.
\end{align*}

It is straightforward to verify that $R_2 = \overline{\Cii{cog}{IFC-LDM}}-1$ is achievable by considering the weak, mixed, and strong interference regions separately.
\end{proof}

Similar to the symmetric case, let
\begin{align*}
    \beta_1 &=
    \frac{\x_1\x_2+\y_1\y_2-2\sqrt{\x_1\x_2\y_1\y_2}\cos\theta}{\x_1\x_2}\\
    \beta_2 &=
    \frac{\x_1\x_2+\y_1\y_2-2\sqrt{\x_1\x_2\y_1\y_2}\cos\theta}{\y_1\y_2},
\end{align*}
and it is easy to show that when $\frac{x_1x_2}{y_1y_2}\geq 4 (\frac{x_1x_2}{y_1y_2}\leq \frac{1}{4})$, we have $\beta_1\geq \frac{1}{4}(\beta_2\geq \frac{1}{4})$. Then we can show the following lemma, which is the counterpart of Lemma~\ref{lem:cog_LDM_k} for the Gaussian case.

\begin{lem}\label{lem:cogreceivepowers}
    When $\frac{\x_1\x_2}{\y_1\y_2}\geq 4 \text{ or } \frac{\x_1\x_2}{\y_1\y_2}\leq
  \frac{1}{4}$, we have $\beta_1\x_1(1\wedge \frac{\x_2}{\y_2})\geq \frac{1}{4}[\x_1(1\wedge
  \frac{\x_2}{\y_2})]\vee[\y_1(1\wedge\frac{\y_2}{\x_2})]
\stackrel{\text{def}}{=} \frac{\tilde{k}}{4}$.
  \end{lem}
  \begin{proof}
  If $\frac{\x_1\x_2}{\y_1\y_2}\geq 4$, we have $\beta_1\geq
    \frac{1}{4}$ and $\x_1\geq4\frac{\y_1\y_2}{\x_2}$. Hence
    \begin{align*}
      \beta_1\x_1(1\wedge \frac{\x_2}{\y_2})&\geq  \frac{1}{4}\x_1(1\wedge \frac{\x_2}{\y_2})\\
      \beta_1\x_1(1\wedge \frac{\x_2}{\y_2})&\geq  \beta_1\frac{4\y_1\y_2}{\x_2}(1\wedge \frac{\x_2}{\y_2})\geq \y_1(1\wedge\frac{\y_2}{\x_2})\geq \frac{1}{4}\y_1(1\wedge\frac{\y_2}{\x_2})
    \end{align*}
    If $\frac{\x_1\x_2}{\y_1\y_2}\leq
  \frac{1}{4}$, we can rewrite the LHS as
  \begin{align*}
    \beta_1\x_1(1\wedge \frac{\x_2}{\y_2}) = \beta_2\frac{\y_1\y_2}{\x_2}(1\wedge
    \frac{\x_2}{\y_2}) =
          \beta_2\y_1(1\wedge \frac{\y_2}{\x_2}).
  \end{align*}
Now, using the fact that $\beta_2\geq \frac{1}{4}$ and $\y_1\geq
4\frac{\x_1\x_2}{\y_2}$ when $\frac{\x_1\x_2}{\y_1\y_2}\leq
\frac{1}{4}$, we can show similarly that
\begin{align*}
          \beta_2\y_1(1\wedge \frac{\y_2}{\x_2})\geq
\frac{1}{4}[\x_1(1\wedge \frac{\x_2}{\y_2})]\vee[\y_1(1\wedge
\frac{\y_2}{\x_2})].
\end{align*}
\end{proof}

\begin{proof}[Proof of Lemma~\ref{thm:cogGaussianLDMachievability}]
  When $\z\leq \x_1\vee \y_2, \y_2\leq 1, \x_1\leq 1$ or $\x_2\leq 1$, it is easy to see from the LDM that the cooperate is not needed and $\overline{\CRoi{LDM}}$ can be achieved by the scheme for the interference channel. So we assume $z>x_1\vee y_2$ and $x_1, x_2, y_2>1$ below.

  When $\frac{1}{4}\leq \frac{x_1x_2}{y_1y_2}\leq 4$, it corresponds to the region $n_1+n_2 = \alpha_1+\alpha_2$ for the LDM. As the channel gains are aligned, the cooperation is also not helpful. In fact, $\overline{\CRoi{LDM}}$ is dominated by $u_1'$ and $u_3'$ in this region, and it is not hard to verify that it is smaller than $\CRoi{IFC}$ using Lemma~\ref{lem:cog_IF}. Hence $\overline{\CRoi{LDM}}$ can be achieved by the scheme for the interference channel. Below we further assume that $\frac{x_1x_2}{y_1y_2}\geq 4$ or $\frac{x_1x_2}{y_1y_2}\geq\frac{1}{4}$.

  We assume that $y_1>1$. According to the LDM, we set $\delta_\A=1, \delta_\B=\delta$, and $\delta_C=0$, and cooperation is achieved through cooperative-private messages. For simplicity, we will require that $R_{1\B}, R_{1\A}\geq \log(1+\x_1)-R_0$.

  In mode \B, source 1 uses power $\frac{1}{\x_1}$ to share bits with source 2 and power $1-\frac{1}{\x_1}$ to send data to destination 3. Under the natural order of superposition coding, the following rates are suppported.
  \begin{align*}
    R_{1B} &=  \log(1+\frac{(1-\frac{1}{\x_1})\x_1}{2}) = \log(1+\x_1)-1\\
    \frac{\bp 12}{\delta} &=  \log(1+\frac{\z}{\x_1})\geq \beta -n_1-1.
  \end{align*}

  For the virtual channel, source 1 uses three messages $W_1, U_1, V_1$ and source 2 uses two messages $W_2, U_2$. For source~1, we allocate powers $\sigma_{W_1}^2= \frac{1}{3}, \sigma_{U_1}^2 = \frac{1}{3\y_2}, \Var{X_{V_1}} = \frac{1}{3}(1\wedge \frac{\x_2}{\y_2})$, and for source~2, $\sigma_{W_2}^2= \frac{1}{3}, \sigma_{U_2}^2 = \frac{1}{3\y_1}, \Var{X_{V_2}} = \frac{\y_2}{\x_2}\Var{X_{V_1}}=\frac{1}{3}(1\wedge \frac{\y_2}{\x_2})$. Destination 1 gets $W_1, U_1, V_1, W_2, U_2$ with powers $\frac{\x_1}{3}, \frac{\x_1}{3\y_2}, \frac{\beta_1\x_1}{3}(1\wedge \frac{\x_2}{\y_2}), \frac{\y_1}{3}, \frac{1}{3}$, resp., and $U_2$ is treated as noise. Destination 2 gets $W_2, U_2, W_1, U_1$ with powers $\frac{\x_2}{3}, \frac{\x_2}{3\y_1}, \frac{\y_2}{3}, \frac{1}{3}$, resp., and $U_1$ is treated as noise.
  Using lemma~\ref{lem:cogreceivepowers}, it is easy to verify that the  following constraints on non-negative rates imply all the relevant constraints in Theorem~\ref{thm:IFconf}.
  \begin{align*}
    R_{W_1}+R_{U_1}+R_{W_2}+R_{V_1}&\leq \log(1+\frac{\x_1+\y_1}{4})\\
    R_{U_1}+R_{W_2}+R_{V_1}&\leq \log(1+\frac{\y_1+\tilde{k}/4}{4})\\
    R_{W_1}+R_{U_1}+R_{V_1}&\leq \log(1+\frac{\x_1+\tilde{k}/4}{4})\\
    R_{W_1}+R_{U_1}&\leq \log(1+\frac{\x_1}{4})\\
    R_{U_1}+R_{W_2}&\leq \log(1+\frac{\frac{\x_1}{\y_2}+\y_1}{4})\\
    R_{U_1}+R_{V_1}&\leq \log(1+\frac{\tilde{k}/4}{4})\\
    R_{U_1}&\leq \log(1+\frac{\x_1}{4\y_2})\\
    R_{V_1}&\leq \bp 12\\
    R_{W_1}+R_{W_2}+R_{U_2}&\leq \log(1+\frac{\x_2+\y_2}{4})\\
    R_{W_1}+R_{U_2}&\leq \log(1+\frac{\frac{\x_2}{\y_1}+\y_2}{4})\\
    R_{W_2}+R_{U_2}&\leq \log(1+\frac{\x_2}{4})\\
    R_{U_2}&\leq \log(1+\frac{\x_2}{4\y_1}).
  \end{align*}

  First we will get the condition on $R_0$ such that $R_{1A} =  \log(1+\x_1)-R_0$ is supported by the above constraints. Set $R_2=0$. In the worst case, we have $\bp 12 = 0$ when $R_{V_1} = 0$. So at least we can achieve $R_{1A} = R_{W_1}+R_{U_1}$, where non-negative $R_{W_1}$ and $R_{U_1}$ satisfy the constraints
  \begin{align*}
      R_{W_1}+R_{U_1}&\leq  \log(1+\frac{\x_1}{4})\\
      R_{U_1}&\leq  \log(1+\frac{\x_1}{16\y_2})\\
      R_{W_1}&\leq  \log(1+\frac{\x_2+\y_2}{4}).
  \end{align*}
  Hence a rate $R_{1A}$ which is the minimum of $\log(1+\frac{\x_1}{4})$ and  $\log(1+\frac{\x_1}{16\y_2})+\log(1+\frac{\x_2+\y_2}{4})$ is acheivable.  Thus,  we may conclude that $R_{1A} = (\log(1+\x_1)-R_0)^+$ is achievable when  $R_0\geq 7$.

  Now in the original constraints, set $R_{1\A} = (\log(1+\x_1)-R_0)^+$. Then by Fourier-Motzkin elimination, we can show that $R_{2\A} = \min(v_1-9, v_2+\bp 12-7+R_0, v_3-19, v_4+\bp 12-16+R_0)$
  is achievable, where $v_i, i = 1, 2, 3, 4$ are defined in Proposition~\ref{prop:LDcogIFC}. When $R_0\geq 7$, using the fact that $\bp 12\geq \delta(\beta-n_1-1)$, we get
  \begin{align*}
    R_2(\delta) = \frac{1}{1+\delta}R_{2\A} \geq\min(u_1'-9,
    u_2'-7+R_0-1, u_3'-19, u_4'-16+R_0-1).
  \end{align*}
  Hence $\overline{\CRoi{LDM}}$ can be achieved within 10 bits in this region.

  The case $y_1\leq 1$ is similar and we can show that $\overline{\CRoi{LDM}}$ can be achieved in this region. The proof is omitted due to space limit.
\end{proof}

\section{Proof of Theorem~\ref{thm:cog_outerbound}}\label{app:cogouterboundtheorem}

We prove the outerbound by first proving an outerbound for a more general channel with generalized feedback of which ours is a special case. Specifically, we consider the following two user cognitive interference channel $p(y_2,y_3,y_4|x_1,x_2)$ whose input alphabets are ${\mathcal X}_1$, ${\mathcal X}_2$ respectively for primary and secondary sources, output alphabets are ${\mathcal Y}_3$,  ${\mathcal Y}_4$ respectively for primary and secondary destinations, and  ${\mathcal Y}_2$ is the output alphabet for the secondary source. Let $W_1$ and $W_2$ be the messages of the primary and secondary sources. At time $t$, the secondary sources signal $X_{2,t}$ may depend only on its past outputs $Y_2^{t-1}$ and its message $W_2$. We also have cost functions $c_1:{\mathcal X}_1\rightarrow {\mathbb R}_+$ and $c_2:{\mathcal X}_2\rightarrow {\mathbb R}_+$ and there are average cost constraints $P_1$ and $P_2$, respectively, on the primary and secondary sources.
Along the lines of~\cite{TelatarTse07}, we focus on channels of the following form $p(y_2,y_3,y_4|x_1,x_2) = \sum_{u_1,u_2} p(u_1,u_2,y_2,y_3,y_4|x_1,x_2)$, where
\[  p(u_1,u_2,y_2,y_3,y_4|x_1,x_2) = p(u_1,y_2|x_1)p(u_2|x_2)\delta(y_3-f_3(x_1,u_2))\delta(y_4-f_4(x_2,u_1)),\]
where $U_1$ and $U_2$ take values in alphabets ${\mathcal U}_1$ and ${\mathcal U}_2$ respectively, and, for every $x_1\in{\mathcal X}_1$, the map $f_3(x_1,.):{\mathcal U}_2 \rightarrow {\mathcal Y}_3$ defined as $u_2 \mapsto f_3(x_1,u_2)$ is invertible, and similarly, for $f_4$. The capacity region of this channel may be defined as usual.

The following gives an outerbound on the capacity region of the above channel.
\begin{thm}\label{thm:cog-gen-ob}
If $(R_1,R_2)$ belongs to the capacity region of the above channel, there there is a $p(q,x_1,x_2)$ with ${\mathbb E}[c_1(X_1)]\leq P_1$ and ${\mathbb E}[c_2(X_2)]\leq P_2$ such that for the joint distribution
\[  p(u_1,u_2,y_2,y_3,y_4,x_1,x_2)= p(u_1,u_2,y_2,y_3,y_4|x_1,x_2)p(q,x_1,x_2),\]
\begin{align}
R_2 &\leq I(X_2;Y_4|X_1,Q), \label{eq:cog-gen-obR2}\\
R_1+R_2 &\leq I(X_1;Y_2,Y_3|Y_4,X_2,Q) + I(X_1,X_2;Y_4|Q), \label{eq:cog-gen-obsum1}\\
R_1+R_2 &\leq I(X_2;Y_4|Y_3,X_1,Q) + I(X_1,X_2;Y_3|Q), \label{eq:cog-gen-obsum2}\\
2R_1+R_2 &\leq I(X_1,X_2;Y_3|Q) + I(X_1;Y_2|Q) + I(X_1,X_2;Y_4|U_2,Y_2,Q) + I(X_1;Y_3|X_2,Y_2,Y_4,Q). \label{eq:cog-gen-2R1plusR2}
\end{align}
\end{thm}
\begin{IEEEproof}
The first bound \eqref{eq:cog-gen-obR2} is a simple cutset bound. The next two \eqref{eq:cog-gen-obsum1}-\eqref{eq:cog-gen-obsum2} were proved in~\cite[Theorem II.1]{TuninettiITA10}. We omit the proofs here. The last one is new and its proof follows.

By Fano's inequality, for any $\epsilon>0$, we have a sufficiently large blocklength $n$ such that
\begin{align*}
n(R_1 - \epsilon) &\leq I(W_1;Y_3^n)= H(Y_3^n) - H(Y_3^n|W_1) = H(Y_3^n) - H(Y_3^n|X_1^n,W_1).
\end{align*}
But, $H(Y_3^n|X_1^n,W_1) = H(U_2^n|X_1^n,W_1) \geq H(U_2^n|Y_2^n,X_1^n,W_1) = H(U_2^n|Y_2^n)$, where the last equality follows from the facts that $U_2^n - X_2^n - (W_2,Y_2^n) - (W_1,X_1^n)$ is a Markov chain and $W_1,W_2$ are independent. Hence,
\begin{align}
n(R_1-\epsilon) &\leq H(Y_3^n) - H(U_2^n|Y_2^n). \label{eq:cog-gen-ob-proof1}
\end{align}
Another application of Fano's inequality gives
\begin{align}
n(R_1-\epsilon) &\leq I(W_1;Y_3^n)\notag\\ &\leq I(W_1;Y_3^n,Y_2^n,Y_4^n,W_2)\notag\\
 		&= I(W_1;Y_3^n,Y_2^n,Y_4^n|W_2)\notag\\
                &= H(Y_2^n|W_2) + H(Y_4^n|Y_2^n,W_2) + H(Y_3^n|Y_2^n,Y_4^n,W_2) - H(Y_2^n,Y_3^n,Y_4^n|W_2). \label{eq:cog-gen-ob-proof2}
\end{align}
Again, using Fano's inequality,
\begin{align*}
n(R_2-\epsilon) &\leq I(W_2;Y_4^n)\notag\\ &\leq I(W_2;Y_4^n,Y_2^n,U_2^n)\notag\\
		&\stackrel{\text{(a)}}{=} I(W_2;Y_4^n,U_2^n|Y_2^n)\\
		&= H(Y_4^n,U_2^n|Y_2^n) - H(Y_4^n,U_2^n|Y_2^n,W_2)\\
		&= H(U_2^n|Y_2^n) + H(Y_4^n|U_2^n,Y_2^n) - H(Y_4^n|Y_2^n,W_2) - H(U_2^n|Y_4^n,Y_2^n,W_2),
\end{align*}
where (a) follows from the fact that $Y_2^n - X_1^n - W_1 - W_2$ is a Markov chain and $W_1$ is independent of $W_2$. Furthermore, $H(U_2^n|Y_4^n,Y_2^n,W_2) = H(U_2^n|X_2^n,Y_4^n,Y_2^n,W_2) = H(U_2^n|X_2^n)$, where the first equality is due to the fact that $X_2^n$ is a deterministic function of $(W_2,Y_2^n)$ and the second equality follows from $U_2^n - X_2^n - (Y_4^n,Y_2^n,W_2)$ being a Markov chain. Thus,
\begin{align}
n(R_2-\epsilon) &\leq H(U_2^n|Y_2^n) + H(Y_4^n|U_2^n,Y_2^n) - H(Y_4^n|Y_2^n,W_2) - H(U_2^n|Y_4^n,Y_2^n,W_2). \label{eq:cog-gen-ob-proof3}
\end{align}

Adding up \eqref{eq:cog-gen-ob-proof1}-\eqref{eq:cog-gen-ob-proof3}, we have
\begin{align}
n(2R_1+R_2 - 3\epsilon) &\leq H(Y_3^n) + H(Y_2^n|W_2) + H(Y_3^n|Y_2^n,Y_4^n,W_2) + H(Y_4^n|U_2^n,Y_2^n)
	\notag\\&\qquad
	- H(Y_2^n,Y_3^n,Y_4^n|W_2) - H(U_2^n|Y_4^n,Y_2^n,W_2) \label{eq:cog-gen-ob-proof4}
\end{align}
But,
\begin{align*}
H(Y_3^n) &\leq \sum_{t=1}^n H(Y_{3,t}),\\
H(Y_2^n|W_2) &= \sum_{t=1}^n H(Y_{2,t})\\
H(Y_3^n|Y_2^n,Y_4^n,W_2) &= H(Y_3^n|X_2^n,Y_2^n,Y_4^n,W_2) = H(Y_3^n|X_2^n,Y_2^n,Y_4^n) \leq \sum_{t=1}^n H(Y_{3,t}|X_{2,t},Y_{2,t},Y_{4,t}),\\
H(Y_4^n|U_2^n,Y_2^n) &\leq \sum_{t=1}^n H(Y_{4,t}|U_{2,t},Y_{2,t}),\\
H(Y_2^n,Y_3^n,Y_4^n|W_1,W_2) &= H(Y_2^n|W_1,W_2) + H(Y_3^n,Y_4^n|Y_2^n,W_1,W_2)\\
	&= H(Y_2^n|X_1^n,W_1,W_2) + H(Y_3^n,Y_4^n|X_1^n,X_2^n,Y_2^n,W_1,W_2)\\
	&= \sum_{t=1}^n \left(H(Y_{2,t}|X_{1,t}) + H(Y_{3,t},Y_{4,t}|X_{1,t},X_{2,t},Y_{2,t})\right)\\
	&= \sum_{t=1}^n \left(H(Y_{2,t}|X_{1,t}) + H(Y_{3,t}|X_{1,t},X_{2,t},Y_{2,t}) + H(Y_{4,t}|X_{1,t},X_{2,t},Y_{2,t})\right),\\
	&= \sum_{t=1}^n \left(H(Y_{2,t}|X_{1,t}) + H(Y_{3,t}|X_{1,t},X_{2,t},Y_{2,t},Y_{4,t}) + H(Y_{4,t}|X_{1,t},X_{2,t},U_{2,t},Y_{2,t})\right),\\
H(U_2^n|Y_4^n,Y_2^n,W_2) &= H(U_2^n|X_2^n,Y_4^n,Y_2^n,W_2)\\ &= H(U_2^n|X_2^n)\\
	&\geq H(U_2^n|X_1^n,X_2^n)\\
	&= H(Y_3^n|X_1^n,X_2^n)\\
	&= \sum_{t=1}^n H(Y_{3,t}|X_{1,t},X_{2,t}).
\end{align*}
Substituting in \eqref{eq:cog-gen-ob-proof4}, we get
\begin{align*}
n(2R_1&+R_2 - 3\epsilon)\\ &\geq \sum_{t=1}^n I(X_{1,t},X_{2,t};Y_{3,t}) + I(X_{1,t};Y_{2,t}) + I(X_{1,t},X_{2,t};Y_{4,t}|U_{2,t},Y_{2,t}) + I(X_{1,t};Y_{3,t}|X_{2,t},Y_{2,t},Y_{4,t}).
\end{align*}
Proceeding as usual by picking $Q$ to be uniformly distributed over $\{1,\ldots,n\}$ and letting $X_1=X_{1,Q}$ and so on, we obtain \eqref{eq:cog-gen-2R1plusR2}.

\end{IEEEproof}

We will use the above theorem to prove our outerbound. Notice that our channel fits the model if we identify the primary and secondary sources' channel inputs as $X_1$ and $(X_2,S_2)$ respectively, the output for the secondary source is $Y_2$, and $U_1 = h_{14}X_{1} + Z_{4}$, $U_2 = h_{23}X_{2}1_{S_{2}=1} + Z_{3}$. The primary and secondary destinations' channel outputs are $Y_3 = h_{13}X_1 + U_2$, and $Y_4 = h_{24}X_21_{S_2=1} + U_1$ respectively. And the cost functions are $c_1(x_1)=|x_1|^2$ and $c_2(x_2,s_2)=|x_2|^2 1_{s_2=1}$ with unit power constraints $P_1=P_2=1$.

In Theorem~\ref{thm:cog-gen-ob}, let $\delta = {\mathbb P}(S_2=0)/{\mathbb P}(S_2=1)$. Also, let
\begin{align*}
P_{1\A}&={\mathbb E}\left[|X_1|^2 \mid S_2=1\right],&
P_{1\B}&={\mathbb E}\left[|X_1|^2\mid S_2=0\right],\text{ and}\\
P_{2\A}&={\mathbb E}\left[|X_2|^2\mid S_2=1\right],& P_{2\B}&=0.
\end{align*}
We have ${\mathbb E}\left[|X_1|^2\right]=(P_{1\A}+\delta P_{1\B})/(1+\delta) \leq 1$, and
 ${\mathbb E}\left[|X_2|^21_{S_2=1}\right]=P_{1\A}/(1+\delta) \leq 1$. We now derive the outerbounds:
\begin{enumerate}
\item $R_2$\\
From \eqref{eq:cog-gen-obR2},
\begin{align*}
R_2 &\leq I(X_2,S_2;Y_4|X_1,Q)\\
	&\leq H(S_2) + I(X_2;Y_4|X_1,Q,S_2=1){\mathbb P}(S_2=1)\\
	&\leq 1 + \frac{1}{1+\delta}\log(1+\x_2P_{2A}).
\end{align*}

\item $R_1+R_2$\\
From \eqref{eq:cog-gen-obsum1},
\begin{align*}
R_1+R_2 &\leq I(X_1;Y_2,Y_3|Y_4,X_2,S_2,Q) + I(X_1,X_2,S_2;Y_4|Q)\\
	&\leq I(X_1;Y_2,Y_3|Y_4,Q,S_2=0){\mathbb P}(S_2=0) + I(X_1;Y_3|Y_4,X_2,Q,S_2=1){\mathbb P}(S_2=1)
	\\&\qquad
	+ H(S_2) +  I(X_1;Y_4|Q,S_2=0){\mathbb P}(S_2=0) +  I(X_1,X_2;Y_4|Q,S_2=1){\mathbb P}(S_2=1)\\
	&= H(S_2) + I(X_1;Y_2,Y_3,Y_4|Q,S_2=0){\mathbb P}(S_2=0)
	\\&\qquad
	+ ( I(X_1,X_2;Y_4|Q,S_2=1) + I(X_1;Y_3|Y_4,X_2,Q,S_2=1)){\mathbb P}(S_2=1)\\
	&\leq 1 + \frac{\delta}{1+\delta}\log(1+(\x_1+\y_2+\z)P_{1B})
	\\&\qquad
	+\frac{1}{1+\delta}\left(\log (1+2\x_2P_{2A}+2\y_2P_{1A}) + \log\left(1+\frac{\x_1P_{1A}}{1+\y_2P_{1A}}\right)\right)
\end{align*}

\item $R_1+R_2$\\
From \eqref{eq:cog-gen-obsum2},
\begin{align*}
R_1+R_2 &\leq I(X_2,S_2;Y_4|Y_3,X_1,Q) + I(X_1,X_2,S_2;Y_3|Q)\\
	&\leq H(S_2) + I(X_2;Y_4|Y_3,X_1,Q,S_2=1){\mathbb P}(S_2=1)
	\\&\qquad
	+ H(S_2) + I(X_1;Y_3|Q,S_2=0){\mathbb P}(S_2=0) + I(X_1;Y_3|Q,S_2=1){\mathbb P}(S_2=1)\\
	&= 2H(S_2) + I(X_1;Y_3|Q,S_2=0){\mathbb P}(S_2=0)
	\\&\qquad
	+ (I(X_2;Y_4|Y_3,X_1,Q,S_2=1)+I(X_1;Y_3|Q,S_2=1)){\mathbb P}(S_2=1)\\
	&\leq 2 + \frac{\delta}{1+\delta}\log(1+\x_1P_{1B})
	 + \frac{1}{1+\delta}\left(\log\left(1+\frac{\x_2P_{2A}}{1+\y_1P_{2A}}\right) + \log(1+2\x_1P_{1A}+2\y_1P_{2A})\right)
\end{align*}

\item $2R_1+R_2$\\
From \eqref{eq:cog-gen-2R1plusR2},
\begin{align*}
2&R_1+R_2\\
	&\leq I(X_1,X_2,S_2;Y_3|Q) + I(X_1;Y_2|Q) + I(X_1,X_2,S_2;Y_4|Q,U_2,Y_2) + I(X_1;Y_3|Q,X_2,S_2,Y_2,Y_4)\\
	&\leq I(X_1,X_2,S_2;Y_3|Q) + I(X_1,S_2;Y_2|Q) + I(X_1,X_2,S_2;Y_4|Q,U_2,Y_2) + I(X_1;Y_3|Q,X_2,S_2,Y_2,Y_4)\\
	 &\leq 3H(S_2) +  (I(X_1;Y_3|Q,S_2=0) +  I(X_1;Y_3,Y_2,Y_4|Q,S_2=0)){\mathbb P}(S_2=0)
	\\&\quad
	 + (I(X_1,X_2;Y_3|Q,S_2=1) +  I(X_1,X_2;Y_4|Q,U_2,S_2=1) + I(X_1;Y_3|Q,X_2,Y_4,S_2=1)){\mathbb P}(S_2=1)\\
	&\leq 3 + \frac{\delta}{1+\delta}\left(\log(1+\x_1P_{1B})+\log(1+(\x_1+\y_2+\z)P_{1B})\right)
	\\&\quad
	 +  \frac{1}{1+\delta}\left(\log(1+2\x_1P_{1A}+2\y_1P_{2A}) + +\log(1+\y_2P_{1A}+\frac{2\x_2P_{2A}+\y_2P_{1A}}{1+\y_1P_{2A}}) + \log(1+\frac{\x_1P_{1A}}{1+\y_2P_{1A}})\right).
\end{align*}

\end{enumerate}

\section{Proof of Lemma~\ref{lem:cog_outerbound}}\label{app:cogouterboundlemma}

The power constraint implies that we have $P_{1A}\leq 1+\delta, P_{2A}\leq
1+\delta, P_{1B}\leq \frac{1+\delta}{\delta}$.  In the upper bound of $R_2$ and
$R_1+R_2$, each term is a monotone increasing function of $P_{1A}, P_{2A},
P_{1B}$. So
\begin{align*}
  R_2&\leq  {1+} \frac{1}{1+\delta}\log(1+\x_2(1+\delta))\leq  {1+}  \frac{1}{1+\delta}\log(1+\x_2)+\frac{1}{1+\delta}\log(1+\delta),\\
  R_1+R_2&\leq {1+}  \frac{1}{1+\delta}\Bigg[\log (1+2\x_2(1+\delta)+2\y_2(1+\delta))+\delta\log\left(1+(\x_1+\y_2+\z)\frac{1+\delta}{\delta}\right)\\
  &\qquad\qquad\qquad	+\log\left(1+\frac{\x_1(1+\delta)}{1+\y_2(1+\delta)}\right)\Bigg]\\
  &\leq  {1+}
  \frac{1}{1+\delta}\left[\log(1+2\x_2+2\y_2)+\delta\log(1+(\x_1+\y_2+\z))+\log\left(1+\frac{\x_1}{1+\y_2}\right)\right]\\
  &\qquad	+\frac{\delta}{1+\delta}\log\left(\frac{1+\delta}{\delta}\right)+\frac{2}{1+\delta}\log(1+\delta),\\
  R_1+R_2&\leq  {2+} \frac{1}{1+\delta}\left[\log(1+2\x_1(1+\delta)+2\y_1(1+\delta))+\delta \log(1+\x_1\frac{1+\delta}{\delta})
  +\log (1+\frac{\x_2(1+\delta)}{1+\y_1(1+\delta)})\right]\\
  &\leq  {2+}  \frac{1}{1+\delta}\left[\log(1+2\x_1+2\y_1)+\delta \log(1+\x_1)+\log
  \left(1+\frac{\x_2}{1+\y_1}\right)\right]\\
  &\qquad	+\frac{\delta}{1+\delta}\log\left(\frac{1+\delta}{\delta}\right)+\frac{2}{1+\delta}\log(1+\delta).
\end{align*}

In the upper bound for $2R_1+R_2$, observe that
\begin{align*}
  1+\y_2P_{1A}+\frac{2\x_2P_{2A}+\y_2P_{1A}}{1+\y_1P_{2A}}
  &\leq  1+\y_2(1+\delta)+\frac{2\x_2P_{2A}+\y_2(1+\delta)}{1+\y_1P_{2A}}\\
  &\leq  \max \left\{
  \begin{array}{c}
  1+\y_2(1+\delta)+\frac{(2\x_2+\y_2)(1+\delta)}{1+\y_1(1+\delta)},\\
  1+2\y_2(1+\delta)
  \end{array}
  \right\}\\
  &\leq  (1+\delta)\max \left\{
  \begin{array}{c}
  1+\y_2+\frac{2\x_2+\y_2}{1+\y_1},\\
  1+2\y_2
  \end{array}
  \right\}.
\end{align*}

So we have
\begin{align*}
  2R_1+R_2&\leq  {3+}  \frac{1}{1+\delta}\Big[\log(1+2\x_1(1+\delta)+2\y_1(1+\delta))+\delta
  \log\left(1+\x_1\frac{1+\delta}{\delta}\right)+\log\left(1+\frac{\x_1(1+\delta)}{1+\y_2(1+\delta)}\right)\\
  &\qquad+\log\left(
  \max \left\{
  \begin{array}{c}
  1+\y_2(1+\delta)+\frac{(2\x_2+\y_2)(1+\delta)}{1+\y_1(1+\delta)},\\
  1+2\y_2(1+\delta)
  \end{array}
  \right\}
  \right)+\delta \log\left(1+(\x_1+\y_2+\z)\frac{1+\delta}{\delta}\right)\Big]\\
  &\leq  {3+}  \frac{1}{1+\delta}\Big[\log(1+2\x_1+2\y_1)+\delta
  \log(1+\x_1)+\log(1+\frac{\x_1}{1+\y_2})\\
  &\qquad+\max\left(\log\left(1+\y_2+\frac{2\x_2+\y_2}{1+\y_1}\right),\;\; \log(1+2\y_2)\right)+\delta
  \log(1+(\x_1+\y_2+\z))\Big]\\
  &\qquad+\frac{2\delta}{1+\delta}\log\left(\frac{1+\delta}{\delta}\right)+\frac{3}{1+\delta}\log(1+\delta).
\end{align*}
We finish the proof by noticing that for $\delta\geq 0$,
\begin{align*}
  \frac{\delta}{1+\delta}\log(\frac{1+\delta}{\delta})\leq\frac{1}{e\ln2}
\quad\text{ and }
  \quad\frac{1}{1+\delta}\log(1+\delta)\leq \frac{1}{e\ln2}.
\end{align*}

\bibliographystyle{IEEEtran}

\begin{thebibliography}{10}



\bibitem{ADT08}
S.~Avestimehr, S.~N.~Diggavi, and D.~Tse, ``Wireless network information flow: a deterministic approach", S. Avestimehr, S. Diggavi and D. Tse, {\em IEEE Transactions on Information Theory}, Vol 57, No 4, April 2011.


\bibitem{BT08}
G. Bresler and D. Tse. (Jul. 2008). ``The two-user Gaussian interference channel: a deterministic view", European Transactions in Telecommunications, Vol.~19, pp.~333-354, April 2008.

\bibitem{IFalign}
V. R. Cadambe and S. A. Jafar,  ``Interference alignment and the
degrees of freedom for the K user interference channel,'' {\em IEEE
Transactions on Information Theory}, vol. 54, no. 8, pp. 3425--3441, Aug 2008.

\bibitem{CaoChen07}
Y.~Cao and B.~Chen. ``An achievable region for interference channel with
conferencing,'' in {\em Proceedings of IEEE International Symposium on Information Theory},
2007, pp. 1251--1255.

\bibitem{CaoChen09}
Y. Cao and B. Chen. (Oct. 2009). ``Capacity bounds for two-hop interference networks.'' in {\em Proceedings
of Forty-Seventh Annual Allerton Conference on Communication, Control, and Computing }, 2009, pp. 272-279.

\bibitem{Carleial75}
A.~B.~Carleial, ``A case where interference does not reduce capacity,''
{\em IEEE Transactions on Information Theory}, vol. IT-21, no. 5, pp. 569--570, Sep. 1975.

\bibitem{CoverElGamal}
T.~Cover~and~A.~El~Gamal,
\newblock ``Capacity theorems for the relay channel,''
\newblock {\em IEEE Transactions on Information Theory}, vol. 25, no. 5, pp. 572--584, Sep. 1979.

\bibitem{ElGamalCosta}
A. A. El Gamal and M. H. M. Costa, ``The capacity region of a class
of deterministic interference channels,'' {\em IEEE Transactions on Information Theory}, vol. IT-28, no. 2, pp. 343--346, Mar. 1982.

\bibitem{EtkinTseWang}
R.~Etkin,~D.~Tse,~and~H.~Wang, ``Gaussian interference channel capacity to within one bit,''
{\em IEEE Transactions on Information Theory}, vol. 54, no. 12, pp. 5534--5562, Dec. 2008.

\bibitem{HanKobayashi}
T.~S.~Han and K.~Kobayashi, ``A new achievable rate region for the interference channel,''
{\em IEEE Transactions on Information Theory}, vol. 27, no. 1, pp. 49--60, Jan. 1981.

\bibitem{Host-Madsen06}
A.~H{\o}st-Madsen, ``Capacity bounds for COoperative diversity,''
{\em IEEE Transactions on Information Theory}, vol. 52, no. 4, pp. 1522--1544, Apr. 2006.

\bibitem{JovicicViswanath09}
A. Jovicic and P. Viswanath, ``Cognitive radio: an information-theoretic perspective,''
{\em IEEE Transactions on Information Theory}, vol. 55, no. 9, pp. 3945--3958, Sep. 2009.


\bibitem{Coop_Diversity}
J.~N.~Laneman, D. Tse and G.~W.~Wornell, ``Cooperative diversity in wireless networks: efficient protocols and outage
behavior,'' {\em IEEE Transactions on Information Theory}, vol. 50, no. 12, pp.
3062--3080, Dec. 2004.

\bibitem{Z_channel}
N. Liu and S. Ulukus, ``On the capacity region of the Gaussian Z-channel, '' in {\em Procedings of Global Telecommunications Conference}, 2004, pp.
415-419, vol. 1.

\bibitem{MYK07}
I. Maric, R.D. Yates, and G. Kramer,
``Capacity of interference channels with partial transmitter cooperation,''
{\em IEEE Transactions on Information Theory}, vol. 53, no. 10, pp. 3536-3548, Oct. 2007.

\bibitem{Mohajer08}
S.~Mohajer, S.~N.~Diggavi, C.~Fragouli, and D.~Tse, ``Transmission
techniques for relay-interference networks,'' in {\em Proceedings of Forty-Sixth Annual Allerton Conference on Communication, Control, and Computing},
2008, pp. 467--474.

\bibitem{OLT07}
A. Ozgur, O. Leveque, and D. Tse, ``Hierarchical cooperation achieves optimal capacity scaling in ad hoc networks,'' {\em IEEE Transactions on Information Theory},
vol. 53, no. 10, pp. 3549--3572, Oct. 2007.

\bibitem{HD_Coop}
Y. Peng and D. Rajan, ``Capacity bounds of half-duplex Gaussian cooperative interference
channel,'' in {\em Proceedings of IEEE International Symposium on Information Theory},
2009, pp. 2081--2085.

\bibitem{PrabhakaranVi09S}
V. Prabhakaran and P. Viswanath, ``Interference channels with
source cooperation,'' {\em IEEE Transactions on Information Theory}, vol. 57, no. 1, pp.
156--186, Jan. 2011.

\bibitem{PrabhakaranVi09D}
V. Prabhakaran and P. Viswanath, ``Interference channels with
destination cooperation,'' {\em IEEE Transactions on Information Theory}, vol. 57, no. 1, pp. 187--209, Jan. 2011.

\bibitem{RFL09}
P. Rost, G. Fettweis, and J. N. Laneman, ``Opportunities, constraints, and benefits of relaying in the presence of interference,''
in {\em Proceedings of the 2009 IEEE international conference on Communications}, 2009,
pp. 1--5.

\bibitem{SSNPS07}
O. Simeone, O. Somekh, Y. Bar-Ness, H. V. Poor, and S. Shamai,
``Capacity of linear two-hop mesh networks with rate splitting,
decodeand-forward relaying and cooperation,'' in {\em Proceedings of Forty-Fifth Annual Allerton Conference on Commmunication, Control, and Computing},
2007, pp. 1127--1134.

\bibitem{TelatarTse07}
E. Telatar and D. Tse, ``Bounds on the capacity region of a class of interference channels,'' in {\em Proceedings of IEEE International Symposium on Information Theory}, 2007, pp. 2871--2874.

\bibitem{Thejaswi08}
C.~Thejaswi, A.~Bennatan, J.~Zhang, R.~Calderbank, D.~Cochran.
``Rate-achievability strategies for two-hop interference flows,'' presented at
Forty-Sixth Annual Allerton Conference on Communication, Control, and Computing, Monticello, IL, 2008.

\bibitem{Tuninetti07}
D.~Tuninetti, ``On interference channels with generalized feedback,'' in {\em Proceedings of IEEE International Symposium on Information Theory}, 2007, pp. 2861--2865.

\bibitem{TuninettiITA10}
D.~Tuninetti, ``An outer bound region for interference channels with generalized feedback,'' in {\em Proceedings of Information Theory and Applications Workshop}, 2010. Available: {\tt http://ita.ucsd.edu/workshop/10/files/paper/paper\_264.pdf}

\bibitem{HD_Relay}
S. Vishwanath, S. Jafar and S. Sandhu, ``Half-duplex relays: cooperative communication strategies and outer
bounds,'' in {\em Proceedings of International Conference on Wireless Networks, Communications and Mobile
Computing}, 2005, pp. 1455--1459.

\bibitem{WT11}
I.-H. Wang and D. Tse, ``Interference mitigation through limited
receiver cooperation,'' {\em IEEE Transactions on Information
Theory}, vol. 57, no. 5, pp. 2913-2940, May 2011.

\bibitem{CRC}
Z. Wu and M. Vu, ``Partial Decode-Forward Binning Schemes for
the Causal Cognitive Relay Channels.'' Available: {\tt http://arxiv.org/pdf/1111.3966v2.pdf}

\bibitem{YangTuninetti08}
S.~Yang and D.~Tuninetti, ``A new achievable region for interference
channel with generalized feedback,'' in {\em Proceedings of Annual Conference on
Information Sciences and Systems}, 2008, pp. 803--808.

\bibitem{YangTuninetti11}
S.~Yang and D.~Tuninetti, ``Interference Channel With Generalized Feedback (a.k.a. With Source Cooperation): Part I: Achievable Region'' {\em IEEE Transactions on Information Theory}, vol. 57, no. 5, pp. 2686 - 2710, May 2011.

\bibitem{YuZhou08}
W. Yu and L. Zhou, ``Gaussian z-interference channel with a relay
link: achievability region and asymptotic sum capacity,'' {\em IEEE Transactions on Information Theory}, vol. 58, no. 4, pp. 2413-2426, April 2012.

\end{thebibliography}

\ 

Rui Wu is a Ph.D. candidate in the Department of Electrical and Computer Engineering at University of Illinois at Urbana-Champaign. He received the B.S. degree in Department of Electronic Engineering from Tsinghua University in 2008, and the M.S. degree in Department of Electrical and Computer Engineering at University of Illinois at Urbana-Champaign in 2011. He is a visiting student of the DYOGEN team at Inria, Paris during the fall semester 2013. He is a recipient of the James M. Henderson Fellowship at University of Illinois at Urbana-Champaign in 2009. His research interests include information theory, networking and machine learning. \\

Vinod M. Prabhakaran received his Ph.D. in 2007 from the EECS Department, University of California, Berkeley. He was a Postdoctoral Researcher at the Coordinated Science Laboratory, University of Illinois, Urbana-Champaign from 2008 to 2010 and at Ecole Polytechnique F\'ed\'erale de Lausanne, Switzerland in 2011. In Fall 2011, he joined the Tata Institute of Fundamental Research, Mumbai, where he currently holds the position of a Reader. His research interests are in information theory, wireless communication, cryptography, and signal processing.
He has received the Tong Leong Lim Pre-Doctoral Prize and the Demetri Angelakos Memorial Achievement Award from the EECS Department, University of California, Berkeley, and the Ramanujan Fellowship from the Department of Science and Technology, Government of India. \\

Pramod Viswanath received the Ph.D. degree in EECS from the University of California at Berkeley in 2000. He was a member of technical staff at Flarion Technologies until August 2001 before joining the ECE department at the University of Illinois, Urbana-Champaign. He is a recipient of the Xerox Award for Faculty Research from the College of Engineering at UIUC (2010), the Eliahu Jury Award from the EECS department of UC Berkeley (2000), the Bernard Friedman Award from the Mathematics department of UC Berkeley (2000), and the NSF CAREER Award (2003). He was an associate editor of the IEEE Transactions on Information Theory for the period 2006-2008.\\

Yi Wang received the M.S.E.E. and Ph.D. degrees in information engineering department from Beijing University of Posts and Telecommunications, China, in 1997 and 2000 respectively. He has worked at Tsinghua University and the University of Kiel in Germany as post-doctor. Since 2005 he joined Huawei Technologies Co., Ltd. he led a series of research projects including beyond 3G, superposition coding for LTE-advanced system, distributed antenna system, cloud RAN, and massive MIMO. Currently Dr. Wang is the principal engineer at Huawei leading high frequency research. His research interests cover mobile communications techniques and related standards.

\end{document}